%% file: tpds2019.tex
\documentclass[10pt,journal,compsoc]{IEEEtran}
\ifCLASSOPTIONcompsoc
  \usepackage[nocompress]{cite}
\else
  \usepackage{cite}
\fi
\ifCLASSINFOpdf
\else
\fi

\input{tkmk_header.tex}

\renewcommand\mathsf{\mathit}

\providetoggle{proofs}
\settoggle{proofs}{false}

\begin{document}
%
\title{On Mixing Eventual and Strong Consistency: Acute Cloud Types}
%
%
%
%

\author{Maciej~Kokoci\'nski,
        Tadeusz~Kobus, 
        Pawe{\l}~T.~Wojciechowski
\IEEEcompsocitemizethanks{\IEEEcompsocthanksitem 
The authors are with the Institute of Computing Science, Poznan University of Technology, 
60-965 Pozna\'n, Poland. \protect\\
E-mail: \{Maciej.Kokocinski,Tadeusz.Kobus,Pawel.T.Wojciechowski\}{\hfil\break}@cs.put.edu.pl 
\IEEEcompsocthanksitem 
This work was supported by the Foundation for Polish Science, within the 
TEAM programme co-financed by the European Union under the European 
Regional Development Fund (grant No. POIR.04.04.00-00-5C5B/17-00). 
Kokoci\'nski and Kobus were also supported by the Polish National Science 
Centre (grant No. DEC-2012/07/B/ST6/01230) and partially by the internal 
funds of the Faculty of Computing, Poznan University of Technology.

%
 }}

\IEEEtitleabstractindextext{%
\begin{abstract}
\input{0-abstract.tex}

\end{abstract}

}

\maketitle

\IEEEdisplaynontitleabstractindextext

%
\IEEEpeerreviewmaketitle

\input{algorithms/pseudocode_settings.tex}
\input{histories/histories-elements.tex}

\input{0-macros.tex}

\input{1-introduction.tex}

\input{2-problem_statement.tex}

\input{3a-system_model.tex}

\input{3b-framework.tex}

\input{4-correctness_criteria.tex}
\input{5-correctness_of_bayou.tex}

\input{6-impossibility.tex}

\input{7-related_work.tex}

\input{8-conclusions.tex}


\bibliographystyle{IEEEtran}
\bibliography{bibliography}

%
%

\ifCLASSOPTIONcaptionsoff
  \newpage
\fi

\vspace{-1.225cm}
\begin{IEEEbiography}[{\includegraphics[width=1in,height=1.25in,clip,keepaspectratio]{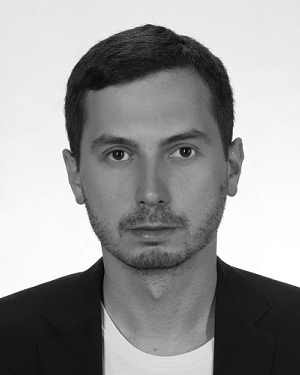}}]{Maciej Kokoci\'nski}
is currently pursuing a Ph.D. degree and working as a Research Assistant 
in the Institute of Computing Science, Poznan University of Technology, 
Poland.
In the past, he was a summer intern at 
Microsoft in Redmond and collaborated on research projects for Egnyte. His research interests 
include theory of distributed systems and transactional memory, and design of
efficient in-memory data stores.
\end{IEEEbiography}

\vspace{-1.0cm}
\begin{IEEEbiography}[{\includegraphics[width=1in,height=1.25in,clip,keepaspectratio]{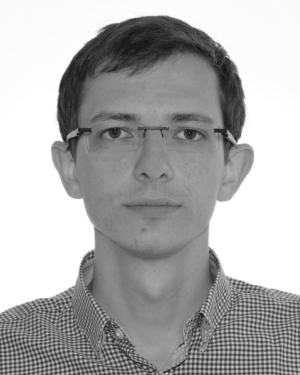}}]{Tadeusz Kobus}
obtained his Ph.D. in Computer Science from Poznan University of Technology,
Poland in 2017, where he is currently an Assistant Professor in the Institute 
of Computing Science. In the past, he was an intern at IBM T. J. Watson 
Research Center and worked on research projects for Egnyte. His research 
interests include fault tolerant distributed algorithms, concurrent 
data structures and, most recently, non-volatile memory.
\end{IEEEbiography}



\vspace{-1.15cm}
\begin{IEEEbiography}[{\includegraphics[width=1in,height=1.25in,clip,keepaspectratio]{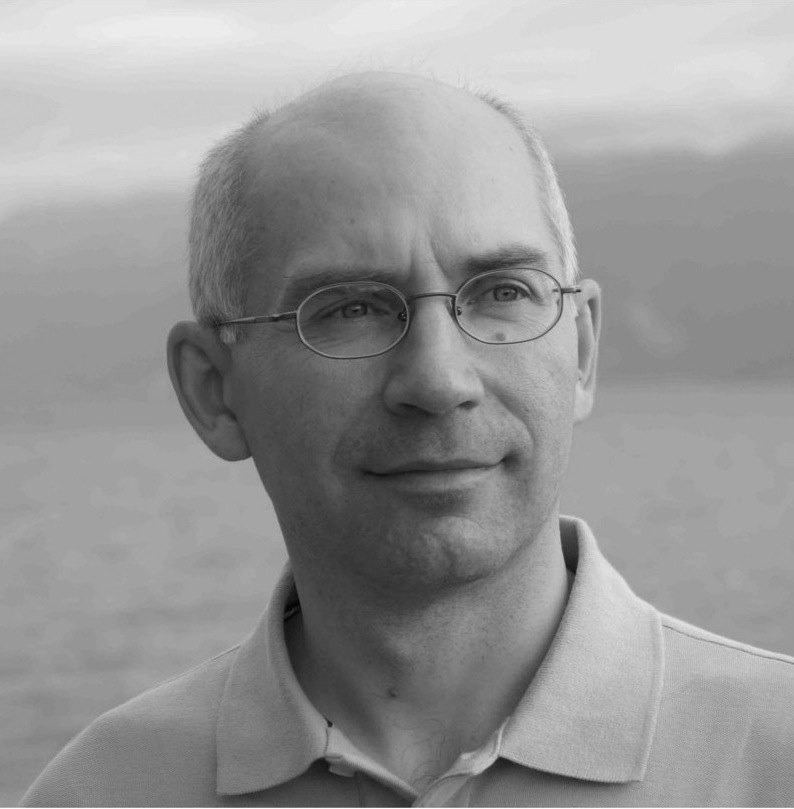}}]{Pawe{\l} T. Wojciechowski}
received the Habilitation degree from Poznan University of Technology,
Poland, in 2008, and the PhD degree in computer science from the
University of Cambridge, in 2000. He was a postdoctoral researcher
with the School of Computer and Communication Sciences, \'Ecole
Polytechnique F\'ed\'erale de Lausanne (EPFL), Switzerland, from 2001
to 2005. He is currently an associate professor with the Institute
of Computing Science, Poznan University of Technology. He has led many
research projects and coauthored dozens of papers. His research
interests span topics in concurrency, distributed computing, and programming
languages.
\end{IEEEbiography}

\balance

\clearpage

\appendices

\input{a-appendix.tex}

\end{document}

%% file: tkmk_header.tex
\usepackage{latexsym}
\usepackage{anyfontsize}
\usepackage{url}

\usepackage{lipsum}

\usepackage{amsthm,amssymb,amsmath}
\usepackage{environ}
\usepackage[polish, english]{babel}
\usepackage{graphicx} 
\usepackage[utf8]{inputenc}
\usepackage{amsthm} 
\makeatletter
\newcommand{\deftheorem}[3]{\@ifundefined{#1}{\newtheorem{#1}{#2}[#3]}{}}
\makeatother
\deftheorem{theorem}{Theorem}{}
\deftheorem{lemma}{Lemma}{}
\deftheorem{corollary}{Corollary}{}
\deftheorem{claim}{Claim}{}

\usepackage{xspace}
\usepackage{mathtools}
\usepackage{thmtools, thm-restate}

\usepackage{multirow}

\usepackage{etoolbox}
\providetoggle{sketches}
\settoggle{sketches}{false}

\providetoggle{extended}
\settoggle{extended}{true}

\usepackage{balance}

\usepackage{etoolbox}
\providetoggle{sketches}
\settoggle{sketches}{false}

\usepackage{tikz}
\usetikzlibrary{arrows,
                calc,
                petri,
                topaths,
                positioning,
                shapes,
                patterns,
                decorations.pathmorphing}
\usepackage{tabularx}
\usepackage{savesym}
\usepackage{newlfont}

\newcommand{\HistoryNameColumn}[1]{\vspace{-1.2cm} #1}
\newcommand{\HistoryDiagramColumn}[1]{\scalebox{0.85}{#1}}%

\newcolumntype{H}{>{\collectcell\HistoryNameColumn} >{\centering\arraybackslash}m{0.25in}  <{\endcollectcell}}
\newcolumntype{D}{>{\collectcell\HistoryDiagramColumn}{l}<{\endcollectcell}}

\newcommand\figurescale{1.0}    
\newcommand{\todo}[1]{}
\newcommand\etal{\emph{et al.}\xspace}

\usepackage{collcell}
\usepackage{comment}

\usepackage{multirow}
\usepackage{centernot}
\usepackage{multicol}
\usepackage{algorithm}
\usepackage{algpseudocode}

\definecolor{yellow-green}{rgb}{0.6, 0.8, 0.2}

\makeatletter
\newcommand{\algcolor}[2]{%
  \hskip-\ALG@thistlm\colorbox{#1}{\parbox{\dimexpr\linewidth-2\fboxsep}{\hskip\ALG@thistlm\relax #2}}%
}

\makeatother

%% file: 0-abstract.tex
In this article we study the properties of distributed systems that mix eventual 
and strong consistency. We formalize such systems through \emph{acute cloud 
types} (ACTs), abstractions similar to conflict-free replicated data types 
(CRDTs), which by default work in a highly available, eventually consistent 
fashion, but which also feature strongly consistent operations for tasks which
require global agreement. Unlike other mixed-consistency solutions, ACTs can 
rely on efficient quorum-based protocols, such as Paxos. Hence, ACTs 
gracefully tolerate machine and network failures also for the strongly 
consistent operations. We formally study ACTs and demonstrate phenomena 
which are neither present in purely eventually consistent nor strongly 
consistent systems. In particular, we identify \emph{temporary operation 
reordering}, which implies interim disagreement between replicas on the relative 
order in which the client requests were executed. When not handled carefully, 
this phenomenon may lead to undesired anomalies, 
including circular causality. We 
prove an impossibility result which states that temporary operation reordering 
is unavoidable in mixed-consistency systems with sufficiently 
complex semantics. Our result is startling, because it shows 
that apparent \emph{strengthening} of the semantics of a system (by introducing 
strongly consistent operations to an eventually consistent system) results in 
the weakening of the guarantees on the eventually consistent operations. 



%% file: algorithms/pseudocode_settings.tex
   \algnewcommand\algorithmicoperation{\textbf{operation}}
\algdef{SE}[MESSAGE]{Operation}{EndOperation}                                   
   [2]{\algorithmicoperation\ \textproc{#1}\ifthenelse{\equal{#2}{}}{}{(#2)}}%
   {\algorithmicend\ \algorithmicoperation}%

\algnewcommand\algorithmicempty{}
\algdef{SE}[MESSAGE]{Empty}{EndEmpty}                                   
   [2]{\algorithmicempty\ \textproc{#1}\ifthenelse{\equal{#2}{}}{}{(#2)}}%
   {\algorithmicend\ \algorithmicempty}%

\algnewcommand\algorithmicoperator{\textbf{operator}}
\algdef{SE}[OPERATOR]{Operator}{EndOperator}                                   
   [2]{\algorithmicoperator\ \textproc{#1}\ifthenelse{\equal{#2}{}}{}{(#2)}}%
   {\algorithmicend\ \algorithmicoperator}%

\algnewcommand\algorithmicmessage{\textbf{message}}
\algdef{SE}[MESSAGE]{Message}{EndMessage}                                   
   [2]{\algorithmicmessage\ \textproc{#1}\ifthenelse{\equal{#2}{}}{}{(#2)}}%
   {\algorithmicend\ \algorithmicmessage}%
              
\algnewcommand\algorithmicreceive{\textbf{receive}}
\algdef{SE}[MESSAGE]{Receive}{EndReceive}                                   
   [2]{\algorithmicreceive\ \textproc{#1}\ifthenelse{\equal{#2}{}}{}{(#2)}}%
   {\algorithmicend\ \algorithmicreceive}%

\algnewcommand\algorithmicupon{\textbf{upon}}
\algdef{SE}[UPON]{Upon}{EndUpon}                                   
   [2]{\algorithmicupon\ \textproc{#1}\ifthenelse{\equal{#2}{}}{}{(#2)}}%
   {\algorithmicend\ \algorithmicupon}%
   
\algnewcommand\algorithmicperiodically{\textbf{periodically}}
\algdef{SE}[PERIODICALLY]{Periodically}{EndPeriodically}                                   
   [2]{\algorithmicperiodically\ \textproc{#1}\ifthenelse{\equal{#2}{}}{}{(#2)}}%
   {\algorithmicend\ \algorithmicperiodically}%

\algnewcommand\senddesc{\textbf{send}}
\algnewcommand\Send{\senddesc{} }

\algnewcommand\rbcastdesc{\textbf{rbcast}}
\algnewcommand\Rbcast{\rbcastdesc{} }

\algnewcommand\tobcastdesc{\textbf{tobcast}}
\algnewcommand\Tobcast{\tobcastdesc{} }

\algnewcommand\structdesc{\textbf{struct}}
\algnewcommand\Struct{\structdesc{} }

\algnewcommand\vardesc{\textbf{var}}
\algnewcommand\Var{\vardesc{} }

\algnewcommand\lockstartdesc{\textbf{lock \{}}
\algnewcommand\LockStart{\lockstartdesc{} }

\algnewcommand\lockenddesc{\textbf{\}}}
\algnewcommand\LockEnd{\lockenddesc{} }

\algnewcommand\algindentdesc{\hspace{2.8em}}
\algnewcommand\AlgIndent{\algindentdesc{} }

\algnewcommand\algindentsmalldesc{\hspace{1.3em}}
\algnewcommand\AlgIndentSmall{\algindentsmalldesc{} }

\algnewcommand\algindentindentdesc{\hspace{4.0em}}
\algnewcommand\AlgIndentIndent{\algindentindentdesc{} }

\algnewcommand{\IIf}[1]{\State\algorithmicif\ #1\ \algorithmicthen}
\algnewcommand{\EElse}[1]{\algorithmicelse}
\algnewcommand{\EndIIf}{}

\algnotext{EndFor}
\algnotext{EndIf}
\algnotext{EndUpon}
\algnotext{EndOperator}
\algnotext{EndMessage}
\algnotext{EndOperation}
\algnotext{EndReceive}
\algnotext{EndFunction} 
\algnotext{EndProcedure} 
\algnotext{EndWhile}
\algnotext{EndEmpty}
\algnotext{EndPeriodically}

\newcommand{\algrule}[1][.2pt]{\par\vskip.5\baselineskip\hrule height
#1\par\vskip.5\baselineskip}

\newcommand{\LineComment}[1]{\hfill\textit{// #1}}

%% file: histories/histories-elements.tex
\newcommand{\myseg}[2]{
   \draw[-] ($ (#1) - (0, 0.07) $) -- ($ (#1) + (0, 0.07) $) -- ($ (#2) + (0, 0.07) $) -- ($ (#2) - (0, 0.07) $) -- ($ (#1) - (0, 0.07) $) ;
}

\newcommand{\mysegfill}[2]{
   \draw[-, fill=black, draw=black] ($ (#1) - (0, 0.06) $) -- ($ (#1) + (0, 0.06) $) -- ($ (#2) + (0, 0.06) $) -- ($ (#2) - (0, 0.06) $) -- ($ (#1) - (0, 0.06) $) ;
}

\newcommand{\mysegdraw}[2]{
   \draw[fill=white,draw=black] ($ (#1) - (0, 0.06) $) -- ($ (#1) + (0, 0.06) $) -- ($ (#2) + (0, 0.06) $) -- ($ (#2) - (0, 0.06) $) -- ($ (#1) - (0, 0.06) $) ;
}

\newcommand{\mysegfillcustom}[3]{
   \draw[-, draw=black, fill=white] ($ (#1) - (0, 0.06) $) -- ($ (#1) + (0, 0.06) $) -- ($ (#2) + (0, 0.06) $) -- ($ (#2) - (0, 0.06) $) -- ($ (#1) - (0, 0.06) $) ;
   \draw[-, draw=black, pattern=#3] ($ (#1) - (0, 0.06) $) -- ($ (#1) + (0, 0.06) $) -- ($ (#2) + (0, 0.06) $) -- ($ (#2) - (0, 0.06) $) -- ($ (#1) - (0, 0.06) $) ;
}

\newcommand{\abseg}[2]{
   \node[label=above:#2] (midlabel) at ($ (#1) $) {};
   \draw[|-|, thick] ($ (#1) - (0.4, 0) $) -- ($ (#1) + (0.4, 0) $);
   \fill [black] ($ (#1) - (0.4, 0.035) $) rectangle ($ (#1) + (0.4, 0.035) $);
}

\newcommand{\belseg}[2]{
   \node[label=below:#2] (midlabel) at ($ (#1) $) {};
   \draw[|-|, very thick] ($ (#1) - (0.4, 0) $) -- ($ (#1) + (0.4, 0) $);
   \fill [black] ($ (#1) - (0.4, 0.035) $) rectangle ($ (#1) + (0.4, 0.035) $);
}

\newcommand{\abseglong}[2]{
   \node[label=above:#2] (midlabel) at ($ (#1) $) {};
   \draw[|-|, thick] ($ (#1) - (0.7, 0) $) -- ($ (#1) + (0.7, 0) $);
   \fill [black] ($ (#1) - (0.7, 0.035) $) rectangle ($ (#1) + (0.7, 0.035) $);
}

\newcommand{\belseglong}[2]{
   \node[label=below:#2] (midlabel) at ($ (#1) $) {};
   \draw[|-|, very thick] ($ (#1) - (0.7, 0) $) -- ($ (#1) + (0.7, 0) $);
   \fill [black] ($ (#1) - (0.7, 0.035) $) rectangle ($ (#1) + (0.7, 0.035) $);
}

\newcommand{\abdot}[2]{
   \node[label=above:#2, fill=black, circle, inner sep=1.5pt] (midlabel) at ($ (#1) $) {};
}

\newcommand{\beldot}[2]{
   \node[label=below:#2, fill=black, circle, inner sep=1.5pt] (midlabel) at ($ (#1) $) {};
}

\newcommand{\xcross}[1]{
   \node[cross out, draw=black, inner sep=3pt, thick] at ($ (#1) $) {};
}

\newcommand{\xcrosssmall}[1]{
   \node[cross out, draw=black, inner sep=1.5pt, thick] at ($ (#1) $) {};
}

\newcommand{\abdotname}[3]{
   \node [label=above:#3, fill=black, circle, inner sep=1.5pt] (#1) at ($ (#2) $) {};
}

\newcommand{\beldotname}[3]{
   \node [label=below:#3, fill=black, circle, inner sep=1.5pt] (#1) at ($ (#2) $) {};
}

\newcommand{\emptydot}[2]{
   \node (#1) at ($ (#2) $) {};
}

\newcommand{\abemptydotname}[3]{
   \node [label=above:#3] (#1) at ($ (#2) $) {};
}

\newcommand{\belemptydotname}[3]{
   \node [label=below:#3] (#1) at ($ (#2) $) {};
}

\newcommand{\absegarb}[4]{
   \node[label=above:#3] (midlabel) at ($ (#4) $) {};
   \draw[|-|, thick] ($ (#1) $) -- ($ (#1) + (#2) $);
}

\newcommand{\absegarbdashed}[4]{
   \node[label=above:#3] (midlabel) at ($ (#4) $) {};
   \draw[|-|, thick, dashed] ($ (#1) $) -- ($ (#1) + (#2) $);
}

\newcommand{\belsegarb}[4]{
   \node[label=below:#3] (midlabel) at ($ (#4) $) {};
   \draw[|-|, thick] ($ (#1) $) -- ($ (#1) + (#2) $);
}

\newcommand{\belsegarbdashed}[4]{
   \node[label=below:#3] (midlabel) at ($ (#4) $) {};
   \draw[|-|, thick, dashed] ($ (#1) $) -- ($ (#1) + (#2) $);
}

\newcommand{\belleftsegarbdashed}[4]{
   \node[label=below:#3] (midlabel) at ($ (#4) $) {};
   \draw[|-, dashed] ($ (#1) $) -- ($ (#1) + (#2) $);
}

\newcommand{\osolid}[1]{
    \tikz[baseline=(todotted.base)]{
        \node[inner sep=1pt,outer sep=0pt] (todotted) {#1\!};
        \draw (todotted.north west) -- (todotted.north east);
    }%
}%

\newcommand{\odashed}[1]{
    \tikz[baseline=(todotted.base)]{
        \node[inner sep=1pt,outer sep=0pt] (todotted) {#1\!};
        \draw[densely dashed] (todotted.north west) -- (todotted.north east);
    }%
}%

\newcommand{\odotted}[1]{
    \tikz[baseline=(todotted.base)]{
        \node[inner sep=1pt,outer sep=0pt] (todotted) {#1\!};
        \draw[dotted] (todotted.north west) -- (todotted.north east);
    }%
}%

\newcommand{\usolid}[1]{
    \tikz[baseline=(todotted.base)]{
        \node[inner sep=1pt,outer sep=0pt] (todotted) {#1};
        \draw (todotted.south west) -- (todotted.south east);
    }%
}%

\newcommand{\udashed}[1]{
    \tikz[baseline=(todotted.base)]{
        \node[inner sep=1pt,outer sep=0pt] (todotted) {#1};
        \draw[densely dashed] (todotted.south west) -- (todotted.south east);
    }%
}%

\newcommand{\udotted}[1]{
    \tikz[baseline=(todotted.base)]{
        \node[inner sep=1pt,outer sep=0pt] (todotted) {#1};
        \draw[dotted] (todotted.south west) -- (todotted.south east);
    }%
}%

\newcommand{\udashedsolid}[1]{%
    \tikz[baseline=(todotted.base)]{
        \node[inner sep=1pt,outer sep=0pt] (todotted) {#1};
        \draw[densely dashed] (todotted.south west) -- (todotted.south east);
        \draw ([yshift=-1.5pt]todotted.south west) -- ([yshift=-1.5pt]todotted.south east);
    }%
}%

\newcommand{\usoliddotted}[1]{%
    \tikz[baseline=(todotted.base)]{
        \node[inner sep=1pt,outer sep=0pt] (todotted) {#1};
        \draw (todotted.south west) -- (todotted.south east);
        \draw[dotted] ([yshift=-1.5pt]todotted.south west) -- ([yshift=-1.5pt]todotted.south east);
    }%
}%

\newcommand{\udashedsoliddotted}[1]{%
    \tikz[baseline=(todotted.base)]{
        \node[inner sep=1pt,outer sep=0pt] (todotted) {#1};
        \draw[densely dashed] (todotted.south west) -- (todotted.south east);
        \draw ([yshift=-1.5pt]todotted.south west) -- ([yshift=-1.5pt]todotted.south east);
        \draw[dotted] ([yshift=-3.0pt]todotted.south west) -- ([yshift=-3.0pt]todotted.south east);
    }%
}%

%% file: 0-macros.tex
\NewEnviron{myeq}{%
\begin{equation}\begin{split}
  \BODY
\end{split}\nonumber\end{equation}}

\newcommand{\Mod}[1]{\ (\mathrm{mod}\ #1)}

\mathchardef\mhyphen="2D 

\newcommand{\RO}{RO\xspace}
\newcommand{\RWW}{RW\xspace}

\renewcommand{\qedsymbol}{\rule{0.7em}{0.7em}}
\newcommand{\instruction}{\mathsf{instruction}\xspace}

\newcommand\mydash[1]{\overline{#1}}

\newcommand{\eqdef}{\stackrel{\text{def}}{=}}
\newcommand{\iffdef}{\stackrel{\text{def}}{\Longleftrightarrow}}
\renewcommand{\iff}{\Leftrightarrow}

\newcommand{\NCC}{\textsc{NCC}\xspace}
\newcommand{\BEC}{\textsc{BEC}\xspace}
\newcommand{\EV}{\textsc{EV}\xspace}
\newcommand{\RVAL}{\textsc{RVal}\xspace}
\newcommand{\FEC}{\textsc{FEC}\xspace}
\newcommand{\CPERC}{\textsc{CPar}\xspace}
\newcommand{\SA}{\textsc{Sess\-Arb}\xspace}
\newcommand{\SC}{\textsc{Seq}\xspace}
\newcommand{\SO}{\textsc{Sin\-Ord}\xspace}

\newcommand{\ar}{\mathsf{ar}}
\newcommand{\blockingops}{\mathsf{blockingops}}
\newcommand{\context}{\mathsf{context}}
\newcommand{\fcontext}{\mathsf{fcontext}}
\newcommand{\corr}{\mathsf{corr}}
\newcommand{\csh}{\mathsf{crash}}
\newcommand{\Events}{\mathsf{Events}}
\providecommand{\false}{}
\renewcommand{\false}{\mathsf{false}}
\newcommand{\hb}{\mathsf{hb}}
\newcommand{\id}{\mathsf{id}}
\newcommand{\lvl}{\mathsf{lvl}}
\newcommand{\lvlmap}{\mathsf{lvlmap}}
\newcommand{\Operations}{\mathsf{Ope\-ra\-tions}}
\newcommand{\ok}{\mathsf{ok}}
\newcommand{\op}{\mathsf{op}}
\newcommand{\proc}{\mathsf{proc}}
\newcommand{\rb}{\mathsf{rb}}
\newcommand{\readonlyops}{\mathsf{readonlyops}}
\newcommand{\rela}{\mathsf{rel}}
\newcommand{\rank}{\mathsf{rank}}
\newcommand{\rval}{\mathsf{rval}}
\newcommand{\sib}{\mathsf{sib}}
\providecommand{\so}{}
\renewcommand{\so}{\mathsf{so}}
\renewcommand{\sp}{\mathsf{sp}}
\newcommand{\sss}{\mathsf{ss}}
\providecommand{\true}{}
\renewcommand{\true}{\mathsf{true}}
\newcommand{\updateops}{\mathsf{updateops}}
\newcommand{\pureupdateops}{\mathsf{pureupdateops}}
\newcommand{\Values}{\mathsf{Values}}
\newcommand{\vis}{\mathsf{vis}}
\newcommand{\words}{\mathsf{words}}

\newcommand{\crap}{\mathsf{cr-agnostic}}
\newcommand{\arflipped}{\mathsf{ar{\mhyphen}flipped}}
\newcommand{\desc}{\mathsf{desc}}

\newcommand{\symmetric}{\mathsf{symmetric}}
\newcommand{\reflexive}{\mathsf{reflexive}}
\newcommand{\irreflexive}{\mathsf{irreflexive}}
\newcommand{\transitive}{\mathsf{transitive}}
\newcommand{\acyclic}{\mathsf{acyclic}}
\newcommand{\total}{\mathsf{total}}

\renewcommand{\natural}{\mathsf{natural}}
\newcommand{\partialorder}{\mathsf{partialorder}}
\newcommand{\totalorder}{\mathsf{totalorder}}
\newcommand{\enumeration}{\mathsf{enumeration}}
\newcommand{\equivalencerelation}{\mathsf{equivalencerelation}}

\newcommand{\eqclasses}{\slash\!\approx}

\renewcommand{\AA}{\mathcal{A}}
\newcommand{\CC}{\mathcal{C}}
\newcommand{\EE}{\mathcal{E}}
\newcommand{\FF}{\mathcal{F}}
\newcommand{\HH}{\mathcal{H}}
\newcommand{\PP}{\mathcal{P}}
\newcommand{\RR}{\mathcal{R}}

\newcommand\sx{\kern-1ex}
\newcommand\pf{\sx$(\FF)$\xspace}

\newcommand{\RA}{\rightarrow}

\newcommand{\raar}{\xrightarrow{\ar}}
\newcommand{\radel}{\xrightarrow{\del}}
\newcommand{\rakd}{\xrightarrow{\kd}}
\newcommand{\raeo}{\xrightarrow{\eo}}
\newcommand{\raeog}{\xrightarrow{\eo_G}}
\newcommand{\rarb}{\xrightarrow{\rb}}
\newcommand{\raro}{\xrightarrow{\ro}}
\newcommand{\rarog}{\xrightarrow{\ro_G}}
\newcommand{\raso}{\xrightarrow{\so}}
\newcommand{\ravis}{\xrightarrow{\vis}}
\newcommand{\rahb}{\xrightarrow{\hb}}
\newcommand{\rarela}{\xrightarrow{\rela}}
\newcommand{\raaa}[1]{\xrightarrow{\mathsf{#1}}}
\newcommand{\rapar}[1]{\xrightarrow{\perc(\mathsf{#1})}}

\newcommand{\app}{\mathsf{append}}
\newcommand{\rd}{\mathsf{read}}

\newcommand{\FFreg}{\FF_\mathit{reg}}
\newcommand{\FFseq}{\FF_\mathit{seq}}
\newcommand{\FFamvr}{\FF_\mathit{AMVR}}
\newcommand{\FFmvr}{\FF_\mathit{MVR}}
\newcommand{\FFmvrs}{\FF_\mathit{MVRS}}
\newcommand{\FForset}{\FF_\mathit{orset}}
\newcommand{\FFkvs}{\FF_\mathit{KVS}}
\newcommand{\FFnnc}{\FF_\mathit{NNC}}
\newcommand\pfseq{\sx$(\FFseq)$\xspace}

\newcommand{\updateContext}{\mathsf{updateContext}}

\newcommand{\SEC}{\textsc{SEC}\xspace}

\newcommand{\CONV}{\textsc{CONV}\xspace}

\newcommand{\CA}{\textsc{CA}\xspace}
\newcommand{\CV}{\textsc{CV}\xspace}
\newcommand{\CY}{\textsc{Causality}\xspace}
\newcommand{\CAC}{\textsc{CC}\xspace}
\newcommand{\NUBEC}{\textsc{NUBEC}\xspace}
\newcommand{\NUCAC}{\textsc{NUCC}\xspace}
\newcommand{\FACAC}{\textsc{FACC}\xspace}
\newcommand{\SNUBEC}{\textsc{SNUBEC}\xspace}
\newcommand{\OCC}{\textsc{OCC}\xspace}
\newcommand{\FAOCC}{\textsc{FAOCC}\xspace}
\newcommand{\NAT}{\textsc{Nat}\xspace}
\newcommand{\NatCC}{\textsc{NatCC}\xspace}
\newcommand{\NOCC}{\textsc{NOCC}\xspace}
\newcommand{\FANOCC}{\textsc{FANOCC}\xspace}

\newcommand{\PAEV}{\textsc{PAEV}\xspace}
\newcommand{\PABEC}{\textsc{PABEC}\xspace}
\newcommand{\CABEC}{\textsc{CABEC}\xspace}
\newcommand{\FABEC}{\textsc{FABEC}\xspace}
\newcommand{\XBEC}{\textsc{XBEC}\xspace}

\newcommand{\ECONV}{\textsc{Even\-tu\-al\-Con\-ver\-gen\-ce}\xspace}
\newcommand{\ERVAL}{\textsc{ERVal}\xspace}
\newcommand{\FRVAL}{\textsc{FRVal}\xspace}
\newcommand{\LIN}{\textsc{Lin}\xspace}

\newcommand{\COMPLETE}{\textsc{complete}}
\newcommand{\RW}{\textsc{Re\-le\-vant\-Writes}\xspace}
\newcommand{\rw}{\mathit{rw}}
\newcommand{\MRshort}{\textsc{MR}\xspace}
\newcommand{\EMRshort}{\textsc{EMR}\xspace}
\newcommand{\RMW}{\textsc{Read\-My\-Writes}\xspace}
\newcommand{\RMWshort}{\textsc{RMW}\xspace}
\newcommand{\ERMWshort}{\textsc{ERMW}\xspace}
\newcommand{\RYWshort}{\textsc{RYW}\xspace}
\newcommand{\ERYWshort}{\textsc{ERYW}\xspace}
\newcommand{\MW}{\textsc{Mo\-no\-to\-nic\-Writes}\xspace}
\newcommand{\MWshort}{\textsc{MW}\xspace}
\newcommand{\MWA}{\textsc{Mo\-no\-to\-nic\-Writes\-In\-Ar\-bi\-tra\-tion}\xspace}
\newcommand{\MWAshort}{\textsc{MWA}\xspace}
\newcommand{\MWV}{\textsc{Mo\-no\-to\-nic\-Writes\-In\-Vi\-si\-bi\-li\-ty}\xspace}
\newcommand{\MWVshort}{\textsc{MWV}\xspace}
\newcommand{\WFR}{\textsc{Writes\-Follow\-Reads}\xspace}
\newcommand{\WFRshort}{\textsc{WFR}\xspace}
\newcommand{\WFRA}{\textsc{Writes\-Follow\-Reads\-In\-Ar\-bi\-tra\-tion}\xspace}
\newcommand{\WFRAshort}{\textsc{WFRA}\xspace}
\newcommand{\WFRV}{\textsc{Writes\-Follow\-Reads\-In\-Vi\-si\-bi\-li\-ty}\xspace}
\newcommand{\WFRVshort}{\textsc{WFRV}\xspace}
\newcommand{\SG}{\textsc{SG}\xspace}
\newcommand{\RT}{\textsc{RT}\xspace}
\newcommand{\TQC}{\textsc{True\-Quie\-scent\-Con\-si\-st\-en\-cy}\xspace}
\newcommand{\QEC}{\textsc{Quie\-scent\-Even\-tu\-al\-Con\-si\-st\-en\-cy}\xspace
}
\newcommand{\QO}{\textsc{Quie\-scent\-Or\-der}\xspace}
\newcommand{\STD}{\textsc{Stan\-dard\-Even\-tu\-al\-Con\-sis\-ten\-cy}\xspace}

\newcommand{\CPRE}{\textsc{CP}\xspace}

\newcommand{\Overwrites}{\textsc{Overwrites}\xspace}
\newcommand{\Irrelevant}{\textsc{Irrelevant}\xspace}
\newcommand{\ConcHide}{\textsc{ConcHide}\xspace}

\newcommand{\para}{\mathsf{parallel}}
\newcommand{\perc}{\mathsf{par}}

\newcommand{\Abot}{$A_\bot$\xspace}
\newcommand{\Aeconv}{$A_\textsc{EC}$\xspace}
\newcommand{\Aqec}{$A_\textsc{QEC}$\xspace}
\newcommand{\Astd}{$A_\textsc{STD}$\xspace}
\newcommand{\Asec}{$A_\textsc{SEC}$\xspace}
\newcommand{\Abec}{$A_\textsc{BEC}$\xspace}
\newcommand{\Abecrmw}{$A_\textsc{BECR}$\xspace}
\newcommand{\Afec}{$A_\textsc{FEC}$\xspace}
\newcommand{\Afecrmw}{$A_\textsc{FECR}$\xspace}
\newcommand{\Asc}{$A_\textsc{Seq}$\xspace}
\newcommand{\Alin}{$A_\textsc{LIN}$\xspace}

\newcommand{\multlinesplit}{\vspace{-1.0cm}}

\newcommand{\clock}{\mathit{clock}\xspace}
\newcommand{\crash}{\mathit{crash}\xspace}
\newcommand{\callret}{\mathit{callret}\xspace}
\newcommand{\receive}{\mathit{receive}\xspace}
\newcommand{\delivered}{\mathit{delivered}\xspace}

\newcommand{\eo}{\mathsf{eo}\xspace}
\newcommand{\tr}{\mathsf{tr}\xspace}

\newcommand{\act}{\mathsf{act}\xspace}
\newcommand{\decode}{\mathsf{decode}\xspace}
\newcommand{\del}{\mathsf{del}\xspace}
\newcommand{\getUpdate}{\mathsf{getUpdate}\xspace}
\newcommand{\kd}{\mathsf{kd}\xspace}
\newcommand{\key}{\mathsf{key}}
\newcommand{\mangle}{\mathsf{mangle}\xspace}
\newcommand{\maxUpdate}{\mathsf{maxUpdate}\xspace}
\newcommand{\maxx}{\mathsf{max}\xspace}
\newcommand{\minx}{\mathsf{min}\xspace}
\newcommand{\msg}{\mathsf{msg}\xspace}
\newcommand{\ops}{\mathsf{ops}}
\newcommand{\ppar}{\mathsf{par}}
\newcommand{\pre}{\mathsf{pre}}
\newcommand{\post}{\mathsf{post}}
\newcommand{\sspre}{\mathsf{sspre}}
\newcommand{\sspost}{\mathsf{sspost}}
\newcommand{\pos}{\mathsf{pos}}
\newcommand{\rcv}{\mathsf{rcv}\xspace}
\newcommand{\replica}{\mathsf{replica}\xspace}
\newcommand{\rep}{\mathsf{rep}\xspace}
\newcommand{\req}{\mathsf{req}\xspace}
\newcommand{\ro}{\mathsf{ro}}
\newcommand{\prop}{\mathsf{prop}}
\newcommand{\snd}{\mathsf{snd}\xspace}
\newcommand{\totalize}{\mathsf{totalize}\xspace}
\newcommand{\trace}{\mathsf{trace}\xspace}

\newcommand{\ctx}{\mathit{ctx}}
\newcommand{\fair}{\mathit{fair}}
\newcommand{\opn}{\mathit{op}}
\newcommand{\rid}{\mathit{rid}}
\newcommand{\suf}{\mathit{suf}}

\newcommand{\Rep}{\mathsf{Replica}}
\newcommand{\rall}{\mathsf{r_{all}}}

\newcommand{\NEV}{\textsc{NEV}\xspace}
\newcommand{\NUEV}{\textsc{NUEV}\xspace}
\newcommand{\SNUEV}{\textsc{SNUEV}\xspace}
\newcommand{\FAEV}{\textsc{FAEV}\xspace}
\newcommand{\CAEV}{\textsc{CAEV}\xspace}
\newcommand{\SFAEV}{\textsc{SFAEV}\xspace}

\newcommand{\obs}{\mathit{obs}}
\newcommand{\notobs}{\mathit{notobs}}

\newcommand{\OpRecord}{\mathsf{OpRec}}
\newcommand{\OpRecords}{\mathsf{OpRecs}}
\newcommand{\makeContext}{\mathsf{makeContext}}
\newcommand{\updates}{\mathsf{updates}}
\providecommand{\visible}{}
\renewcommand{\visible}{\mathsf{visible}}

\newcommand{\aapp}{\mathsf{append}}
\newcommand{\rrd}{\mathsf{read}}
\newcommand{\kvput}{\mathsf{put}}
\newcommand{\kvget}{\mathsf{get}}
\newcommand{\wwr}{\mathsf{write}}
\newcommand{\uupd}{\mathsf{update}}
\newcommand{\aadd}{\mathsf{add}}
\newcommand{\rrem}{\mathsf{remove}}
\newcommand{\ook}{\mathsf{ok}}
\newcommand{\kvundef}{\mathsf{undef}}

\newcommand{\Perform}{\mathsf{Perform}}

\newcommand{\clockk}{\mathsf{clock}}
\newcommand{\ridd}{\mathsf{rid}}

\newcommand{\Notify}{\mathsf{Notify}}

\newcommand{\nat}{\mathsf{nat}\xspace}
\newcommand{\oo}{\mathsf{o}\xspace}
\newcommand{\sett}{\mathsf{set}\xspace}
\newcommand{\opFunction}{\mathsf{opFun}\xspace}
\newcommand{\recUpdates}{\mathsf{recUpdates}\xspace}
\newcommand{\recVisible}{\mathsf{recVisible}\xspace}


\newcommand{\boxx}{\mathsf{state}}
\newcommand{\db}{\mathsf{db}}
\newcommand{\undoLog}{\mathsf{undoLog}}
\newcommand{\readsetMap}{\mathsf{readsetMap}}
\newcommand{\writesetMap}{\mathsf{writesetMap}}
\newcommand{\currentCommitted}{\mathsf{currentCommitted}}
\newcommand{\currentEventNumber}{\mathsf{currEventNo}}
\newcommand{\currentTime}{\mathsf{currTime}}
\renewcommand{\dot}{\mathsf{id}}
\newcommand{\ctxsf}{\mathsf{ctx}}
\newcommand{\Request}{\mathrm{Req}}
\newcommand{\client}{\mathrm{client}}
\newcommand{\VersionId}{\mathrm{VersionId}}
\newcommand{\CommittedVersionId}{\mathrm{CommittedVersionId}}
\newcommand{\vid}{\mathsf{vid}}
\newcommand{\Id}{\mathrm{Id}}
\newcommand{\Value}{\mathrm{Value}}
\newcommand{\response}{\mathsf{response}}
\newcommand{\res}{\mathsf{res}}
\newcommand{\Response}{\mathrm{Resp}}
\newcommand{\Entry}{\mathrm{Entry}}
\newcommand{\rr}{\mathsf{r}}
\newcommand{\strongOp}{\mathsf{strongOp}}
\newcommand{\timestamp}{\mathsf{timestamp}}

\newcommand{\committed}{\mathsf{committed}}
\newcommand{\tentative}{\mathsf{tentative}}

\newcommand{\executed}{\mathsf{executed}}
\newcommand{\toBeExecuted}{\mathsf{toBeExecuted}}
\newcommand{\toBeExecutedOrReexecuted}{\mathsf{toBeExecutedOrReexecuted}}
\newcommand{\toBeRolledBack}{\mathsf{toBeRolledBack}}
\newcommand{\possibleConflicts}{\mathsf{possibleConflicts}}
\newcommand{\tentativelyExecuted}{\mathsf{tentativelyExecuted}}
\newcommand{\requestsAwaitingResponse}{\mathsf{reqsAwaitingResp}}
\newcommand{\lastResponse}{\mathsf{lastResp}}
\newcommand{\strongOpsToCheck}{\mathsf{strongOpsToCheck}}
\newcommand{\contextNotReadyOps}{\mathsf{missingContextOps}}
\newcommand{\readyToScheduleOps}{\mathsf{readyToScheduleOps}}

\newcommand{\adjustExecution}{\mathrm{adjustExecution}}
\newcommand{\promoteToCommitted}{\mathrm{promoteToCommitted}}

\newcommand{\readset}{\mathrm{readset}}
\newcommand{\writeset}{\mathrm{writeset}}
\newcommand{\invoke}{\mathrm{invoke}}
\newcommand{\responserm}{\mathrm{response}}

\newcommand{\bcast}{\mathrm{BE{\text-}cast}}
\newcommand{\bdeliver}{\mathrm{BE{\text-}deliver}}

\newcommand{\rbcast}{\mathrm{RB{\text-}cast}}
\newcommand{\rbdeliver}{\mathrm{RB{\text-}deliver}}

\newcommand{\tobcast}{\mathrm{TOB{\text-}cast}}
\newcommand{\tobdeliver}{\mathrm{TOB{\text-}deliver}}

\newcommand{\frbcast}{\mathrm{FIFO{\text-}RB{\text-}cast}}
\newcommand{\frbdeliver}{\mathrm{FIFO{\text-}RB{\text-}deliver}}

\newcommand{\cabcast}{\mathrm{CAB{\text-}cast}}
\newcommand{\cabdeliver}{\mathrm{CAB{\text-}deliver}}

\newcommand{\pput}{\mathrm{put}}
\newcommand{\gget}{\mathrm{get}}
\newcommand{\rremove}{\mathrm{remove}}
\newcommand{\ccontains}{\mathrm{contains}}
\newcommand{\incAndGet}{\mathrm{incAndGet}}

\newcommand{\reverse}{\mathrm{reverse}}
\newcommand{\event}{\mathsf{event}}
\newcommand{\tob}{\mathsf{tob}}
\newcommand{\tobNo}{\mathsf{tobNo}}
\newcommand{\rbdel}{\mathsf{rbdel}}
\newcommand{\tobdel}{\mathsf{tobdel}}
\newcommand{\RBdel}{\mathsf{RBdel}}
\newcommand{\TOBdel}{\mathsf{TOBdel}}
\newcommand{\sent}{\mathsf{sent}}
\newcommand{\exec}{\mathsf{exec}}
\newcommand{\ret}{\mathsf{ret}}
\newcommand{\TOB}{\mathsf{TOB}}
\newcommand{\RB}{\mathsf{RB}}

\newcommand{\newOrder}{\mathsf{newOrder}}

\newcommand{\execute}{\mathrm{execute}}
\newcommand{\rollback}{\mathrm{rollback}}
\newcommand{\primaryCommit}{\mathrm{primaryCommit}}

\newcommand{\mgs}{\mathsf{mgs}}

\newcommand{\append}{\mathsf{append}}
\newcommand{\duplicate}{\mathsf{duplicate}}
\newcommand{\replace}{\mathsf{replace}}
\newcommand{\getRoom}{\mathsf{reserveRoom}}
\newcommand{\getStable}{\mathsf{getStable}}

\newcommand{\causal}{\mathsf{causal}}
\newcommand{\eventual}{\mathsf{eventual}}
\newcommand{\weak}{\mathsf{weak}}
\newcommand{\strong}{\mathsf{strong}}

\newcommand{\val}{\mathsf{value}}
\newcommand{\RegisterRecord}{\mathsf{RegRec}}
\newcommand{\writefull}{\mathsf{write}}
\newcommand{\readfull}{\mathsf{read}}
\newcommand{\current}{\mathsf{myReg}}
\newcommand{\update}{\mathsf{update}}

\newcommand{\StateObject}{StateObject\xspace}

\newcommand{\causalContext}{\mathsf{causalCtx}}
\newcommand{\dvv}{\mathrm{dvv}}
\newcommand{\hasBeenAlreadyExecuted}{\mathsf{hasBeenAlreadyExecuted}}
\newcommand{\rollbackRequired}{\mathsf{rollbackRequired}}
\newcommand{\rs}{\mathsf{rs}}
\newcommand{\ws}{\mathsf{ws}}
\newcommand{\executeAsCommitted}{\mathsf{executeAsCommitted}}
\newcommand{\isCommitted}{\mathsf{isCommitted}}
\newcommand{\pendingStrongRequests}{\mathsf{pendingStrongReqs}}
\newcommand{\unconfirmed}{\mathsf{unconfirmed}}
\newcommand{\committedExtension}{\mathsf{committedExt}}

\newcommand{\tobcastProposal}{\mathrm{tobcastProposal}}
\newcommand{\voteOK}{\mathrm{voteOK}}
\newcommand{\checkUnconfirmed}{\mathrm{checkUnconfirmed}}
\newcommand{\adjustTentativeOrder}{\mathrm{insertIntoTentative}}
\newcommand{\checkPendingStrongRequests}{\mathrm{checkPendingStrongReqs}}
\newcommand{\commit}{\mathrm{commit}}
\newcommand{\remove}{\mathrm{remove}}
\newcommand{\cleanUp}{\mathrm{cleanUp}}
\newcommand{\getTentative}{\mathrm{getTentative}}
\newcommand{\getCommitted}{\mathrm{getCommitted}}
\newcommand{\insertTentative}{\mathrm{insertTentative}}
\newcommand{\insertCommitted}{\mathrm{insertCommitted}}

\newcommand{\ISSUE}{\mathrm{ISSUE}}
\newcommand{\COMMIT}{\mathrm{COMMIT}}
\newcommand{\primary}{\mathrm{primary}}

\newcommand{\front}{\mathsf{front}}
\providecommand{\last}{}
\renewcommand{\last}{\mathsf{last}}
\newcommand{\listt}{\mathsf{list}}

\newcommand{\previous}{\mathsf{previous}}
\newcommand{\subsequent}{\mathsf{subsequent}}

\newcommand{\newTentative}{\mathsf{newTentative}}
\newcommand{\inOrder}{\mathsf{inOrder}}
\newcommand{\outOfOrder}{\mathsf{outOfOrder}}

\newcommand{\head}{\mathsf{head}}
\newcommand{\tail}{\mathsf{tail}}

\newcommand{\undoMap}{\mathsf{undoMap}}
\newcommand{\idd}{\mathsf{id}}


\newcommand{\msgs}{\mathsf{msgs}}

\newcommand{\propose}{\mathsf{propose}}
\newcommand{\decide}{\mathsf{decide}}

\newcommand{\idSeq}{\mathsf{idSeq}}
\newcommand{\idSet}{\mathsf{idSet}}
\newcommand{\ids}{\mathsf{ids}}
\newcommand{\ordered}{\mathsf{ordered}}
\providecommand{\received}{}
\renewcommand{\received}{\mathsf{received}}
\newcommand{\test}{\mathsf{test}}
\newcommand{\unordered}{\mathsf{unordered}}

\providecommand{\Message}{}
\renewcommand{\Message}{\mathrm{Message}}
\newcommand{\Predicate}{\mathrm{Predicate}}

\newcommand{\AB}{AB\xspace}
\newcommand{\CAB}{CAB\xspace}

\newcommand{\checkDependencies}{\mathsf{checkDep}}
\newcommand{\cc}{\mathsf{cc}}
\newcommand{\find}{\mathsf{find}}


\newcommand{\ACT}{$\mathsf{ACT}$\xspace}
\newcommand{\ACTc}[1]{$\mathsf{ACT}_\mathsf{#1}$\xspace}
\newcommand{\ACTBayou}{AcuteBayou\xspace}
\newcommand{\ACTnnc}{ANNC\xspace}

\newcommand{\sd}{\mathsf{strongSub}}
\newcommand{\si}{\mathsf{strongAdd}}
\newcommand{\wi}{\mathsf{weakAdd}}
\newcommand{\resp}{\mathsf{resp}}
\newcommand{\type}{\mathsf{type}}

\newcommand{\add}{\mathrm{add}}
\newcommand{\subtract}{\mathrm{subtract}}
\newcommand{\get}{\mathrm{get}}
\newcommand{\addit}{\mathit{add}}
\newcommand{\subtractit}{\mathit{subtract}}
\newcommand{\getit}{\mathit{get}}

\newcommand{\ord}{\mathsf{ord}}
\newcommand{\foldr}{\mathsf{foldr}}
\newcommand{\sort}{\mathsf{sort}}
\newcommand{\fnnc}{\mathit{f_{NNC}}}
\newcommand{\deliveredAdds}{\mathit{rbDeliveredAdds}}

\newcommand{\invokeit}{\mathit{invoke}}
\newcommand{\executeit}{\mathit{execute}}
\newcommand{\rollbackit}{\mathit{rollback}}
\newcommand{\rbcastit}{\mathit{rbcast}}
\newcommand{\tobcastit}{\mathit{tobcast}}
\newcommand{\rbdeliverit}{\mathit{rbdeliver}}
\newcommand{\tobdeliverit}{\mathit{tobdeliver}}

\newcommand{\blueapp}{\mathsf{blueAppend}}
\newcommand{\redapp}{\mathsf{redAppend}}
\newcommand{\bluerd}{\mathsf{blueRead}}

\newcommand{\ts}{\mathsf{ts}}
\newcommand{\seq}{\mathsf{seq}}
\newcommand{\delReqs}{\mathsf{deliveredReqs}}
\newcommand{\lc}{\mathsf{lc}}

\newcommand{\Rubis}{RUBiS\xspace}

%% file: 1-introduction.tex
\IEEEraisesectionheading{\section{Introduction}\label{sec:introduction}}

The massive scalability and high availability of the complex (geo-replicated) 
distributed systems that power today's Internet often hinges on the 
use of
eventually consistent data stores. These systems extensively 
employ specialized 
data structures, e.g., last-write-wins registers (LWW-registers), multi-value 
registers (MVRs), observed-remove sets (OR-sets) or other 
conflict-free replicated data types (CRDTs) \cite{DHJK+07} \cite{SPBZ11} 
\cite{SPBZ11a}.
These data structures are replicated on multiple machines (\emph{replicas}) and 
can be read or modified independently on each site without prior 
synchronization with other replicas. It means that replicas can promptly 
respond to the clients. The communication between the replicas happens solely 
using a gossip protocol. By design replicas are guaranteed to be able to 
converge to a single state, automatically resolving any inconsistencies between 
them.

Unfortunately, the semantics of such data structures are very limited. 
To provide high availability, low-latency responses and eventual state 
convergence, these data structures require either that all operations commute 
or that there exist commutative, associative, and idempotent procedures for 
merging replica states. This is why these mechanisms are not suitable for all 
use cases. For example, consider a simple non-negative integer counter. The 
addition operation can be trivially implemented in a conflict-free 
manner, as the addition operations are commutative. However, implementing the 
subtraction operation requires global agreement to ensure that the value of 
the counter never drops below 0. In a similar way, in an auction system 
concurrent bids can be considered independent operations and thus their 
execution does not need to be synchronized. However, the operation that closes 
the auction requires solving distributed consensus to select the single winning 
bid \cite{PBS18}.

Due to the inherent shortcomings of CRDTs and solutions similar to them, 
recently there have been several attempts, both at academia (e.g., 
\cite{SPAL11} \cite{LPC+12} \cite{BGY13} \cite{TPKB+13} \cite{BDFR+15} 
\cite{GYF+16} \cite{LPR18}) and in the industry (e.g., 
\cite{AmazonDynamoDBConsistentReads} \cite{CosmosDBConsistencyLevels} 
\cite{CassandraLWTs} \cite{RiakConsistencyLevels}) to enrich the semantics of 
eventually consistent systems by allowing some operations to be performed with 
stronger consistency guarantees or by introducing (quasi) transactional 
support. Crucially, none of the \emph{mixed-consistency} approaches we are 
aware of is flexible enough to: (a) account for very weak consistency models 
(weaker than causal consistency, which is known to be costly to achieve in 
practice \cite{BFGH+12}), (b) admit strongly consistent operations which do not 
require all replicas to be operational in order to complete (so to gracefully 
tolerate failures), and (c) provide clearly stated semantics that enables easy 
reasoning about the system-wide guarantees. The latter trait is especially 
important when the same data can be accessed at the same time both in a 
strongly and eventually consistent fashion, what is notoriously difficult to 
implement. For example, in Apache Cassandra using the \emph{light weight 
transactions} on data that are accessed at the same time in the regular, 
eventually consistent fashion leads to undefined behaviour \cite{CASS11000}.

In this article we introduce \emph{acute cloud types} (ACTs), a family of 
specialized mixed-consistency data structures designed primarily for high 
availability and low latency, but which also seamlessly integrate on-demand 
strongly consistent semantics. ACT feature two kinds of operations: 
\begin{itemize}
\item \emph{weak operations}--targeted for unconstrained scalability and low 
latency responses (as operations in CRDTs), 
\item \emph{strong operations}--used when eventually consistent guarantees are 
insufficient, require consensus-based inter-replica synchronization prior to 
execution.
\end{itemize}

Weak operations are guaranteed to progress, and are handled in such a way that 
the replicas eventually converge to the same state within each network 
partition, even when strongly consistent operations cannot complete due to 
network and process failures. On the other hand, strong operations can provide 
guarantees even as strong as linearizability \cite{HW90} wrt. the already 
completed strong operations and a precisely defined subset of completed weak 
operations. Crucially, strong operations are \emph{non-blocking}: they can 
leverage efficient, quorum-based synchronization protocols, such as Paxos 
\cite{Lam98}, and thus gracefully tolerate machine and network failures. 
Both weak and strong operations can be arbitrarily complex, but they must be 
deterministic. 

Our approach is more robust than other mixed-consistency solutions. Most 
notably, unlike classic \emph{cloud types} \cite{BFLW12} and \emph{global 
sequence protocol} (GSP) \cite{BLPF15}, ACTs are \emph{symmetrical} in the 
sense that they do not assume the existence of a server or servers that mediate 
all communication between remote replicas. This has several advantages: a 
failure of a replica or a group of replicas cannot impede the ability of other 
ACT replicas to execute weak operations and propagate the resulting updates. 
Also, ACTs can better tolerate network splits by allowing the replicas in the 
minority partitions to execute weak operations and exchange resulting updates. 
Furthermore, unlike the RedBlue consistency model \cite{LPC+12} and approaches 
similar to it (e.g., \cite{LLSG92} \cite{GYF+16} \cite{LPR18}), ACTs support 
consistency guarantees weaker than causal consistency, so account for a wider 
range of systems. Crucially, ACTs do not require all replicas to be operational 
in order for the strong operations to complete, contrary to the approaches 
mentioned above. This latter trait has been fundamental to the design of ACTs.

In order to provide an easy to understand yet flexible consistency model that 
allows weak operations to be executed in a highly available and scalable 
manner, we require that in any run of an ACT, logically, there always exists a 
single global order $S$ of all operations. During execution, strong operations 
are guaranteed to observe the prefix of $S$ up to their position in $S$. 
A weak operation may observe a serialization $S'$ of operations 
that diverges from $S$, but only by a finite number of elements. Thus weak and 
strong operations are interconnected in a non-trivial way, which intuitively 
ensures \emph{write stabilization}: once a strong operation, during its 
execution, observes some weak operations $\opn_i$, $\opn_j$ in that order, all 
subsequent strong operations, and eventually all weak operations, will also 
observe $\opn_i$, $\opn_j$ in that order. It is so even though
weak operations never have to directly synchronize with strong 
operations (e.g., by blocking on the completion of strong operations).

We propose a framework that enables formal reasoning about ACTs and their 
guarantees. We express the dependencies between operations through the 
\emph{visibility} and \emph{arbitration} relations, similarly to \cite{B14}, 
but we allow each operation to observe the arbitration in a temporarily 
inconsistent (but eventually convergent) form. In order to capture the unique 
properties of ACTs and write stabilization in particular, we define a novel 
correctness condition called \emph{fluctuating eventual consistency} (\FEC) 
that is strictly weaker than Burckhardt's Basic Eventual Consistency (\BEC) 
\cite{BGY13}.

By formally specifying ACTs, we uncovered several interesting phenomena unique 
to mixed-consistency systems (they are never exhibited by popular NoSQL 
systems, which only guarantee eventual consistency, nor by strongly consistent 
solutions). Crucially, some ACTs exhibit a phenomenon that we call 
\emph{temporary operation reordering}, which happens when the replicas 
temporarily disagree on the relative order in which the requests (modelled as 
operations) submitted to the system were executed. When not handled carefully, 
temporary operation reordering may lead to all kinds of undesired situations, 
e.g., circular causality among the responses observed by the clients. As we 
formally prove, temporary operation reordering is not present in all ACTs but 
in some cases cannot be avoided. This impossibility result is startling, 
because it shows that apparent \emph{strengthening} of the semantics of a 
system (by introducing strong operations to an eventually-consistent system) 
results in the weakening of the guarantees on the eventually-consistent 
operations.

In order to illustrate our concepts and analysis, we present an ACT for a 
non-negative counter and also revisit Bayou \cite{TTPD+95}, a seminal, always 
available, eventually consistent data store. Bayou combines timestamp-based 
eventual consistency \cite{V09} and serializability \cite{P79} by speculatively 
executing transactions submitted by clients and having a primary replica to 
periodically \emph{stabilize} the transactions (establish the final transaction 
execution order). We show how Bayou can be improved to form a general-purpose
ACT.

\subsection{Contribution summary}

\begin{enumerate}
\item We define \emph{acute cloud types}, a family of specialized 
mixed-consistency data structures designed primarily for high availability and 
low latency, but which also seamlessly integrate on-demand strongly consistent 
semantics achieved through quorum-based consensus protocols. Weak and strong 
operations in ACTs are interconnected in a non-trivial way, which intuitively 
ensures \emph{write stabilization}.

\item We identify a range of traits unique to some ACTs. Most importantly, we 
define \emph{temporary operation reordering}, a situation in which there is an 
interim disagreement between replicas on the relative order in which the client 
requests were executed. 

\item We propose a framework that enables formal reasoning about ACTs and their 
guarantees. In particular, our framework allows us to formalize temporary 
operation reordering and propose a correctness condition called 
\emph{fluctuating eventual consistency} which adequately captures the 
guarantees provided by ACTs that exhibit this phenomenon.

\item We use our framework to prove a number of formal results regarding ACTs. 
Crucially, we show an impossibility result that states that temporary operation 
reordering is not present in all ACTs, but in some cases cannot be avoided. 

\item We revisit the seminal Bayou system, study its consistency guarantees, 
and show how it can be improved to form a general-purpose ACT.

\end{enumerate}

\subsection{Article structure}

The article is organized as follows. In Section~\ref{sec:bayou} 
we explain ACTs through examples: an acute non-negative counter and 
adaptation of Bayou that forms a general-purpose ACT. 
We formally define ACTs in 
Section~\ref{sec:act}, and introduce the formal framework for reasoning 
about their correctness in Section~\ref{sec:framework}. In 
Section~\ref{sec:guarantees} we define FEC, our new correctness criterion and 
prove the correctness of our example ACTs. Next, in 
Section~\ref{sec:impossibility}, we give our impossibility result. We 
discuss related work in Section~\ref{sec:related_work}, and 
conclude in Section~\ref{sec:conclusions}.

A brief announcement of this article appeared in \cite{KKW19a}.

%% file: 2-problem_statement.tex

\section{Acute cloud types by examples} \label{sec:bayou}

\subsection{Acute non-negative counter}

{\renewcommand\small{\footnotesize}%
    \input{algorithms/act_counter_modif.tex}

}

As we mentioned in Section~\ref{sec:introduction}, a non-negative integer 
counter cannot be implemented as a classic CRDT because the subtraction 
operation requires global coordination to ensure 
that the value of the counter never drops below 0. In 
Algorithm~\ref{alg:actnnc} we present an \emph{acute non-negative integer 
counter} (\ACTnnc), a simple ACT implementing such a counter. The $\add$ (line~\ref{alg:actpc:add}) and $\get$ (line~\ref{alg:actpc:get}) operations 
can be weak and thus always ensure low latency responses, whereas $\subtract$ 
(line~\ref{alg:actpc:subtract}) must be a strong operation to ensure
the semantics of a non-negative counter. The crux of \ACTnnc lies in using 
two complementary protocols for exchanging updates (a gossip one and 
one that establishes the ultimate operation serialization), and calculating
the state of the counter by liberally counting $\add$ operations and 
conservatively counting the $\subtract$ operations.


To track the execution of weak and strong operations, each \ACTnnc 
replica maintains three variables (line~\ref{alg:actpc:vars}): one 
for subtraction operations ($\sd$) and two for the addition operations ($\wi$ 
and $\si$). The replicas exchange the information about new ADD requests (weak 
updating operations) using a gossip protocol (modelled using 
\emph{reliable broadcast, RB} \cite{CGR11}) as well as a protocol that involves 
inter-replica synchronization (modelled using \emph{total order broadcast, TOB} 
\cite{PGS98}, which can be efficiently implemented using quorum-based 
protocols, such as Paxos \cite{Lam98}; lines 
\ref{alg:actpc:addrb}-\ref{alg:actpc:addtob}). 
The $\subtract$ operation, which does not commute unlike the $\add$ 
operation, solely uses TOB. Upon receipt of a $\tobcast$ SUBTRACT message, the 
subtract operation completes successfully only if we are certain that the value 
of the counter does not drop below 0, i.e., when the aggregated value of all 
confirmed addition operations ($\wi$) is greater or equal to the aggregated 
value of all subtract operations ($\sd$) increased by $\val$ (lines 
\ref{alg:actpc:subcalcstart}-\ref{alg:actpc:subcalcend}). 


We ensure that on any replica and for any ADD request $r$, the $\rbdeliver(r)$ 
event always happens before the $\tobdeliver(r)$ event (lines 
\ref{alg:actpc:rbretstart}--\ref{alg:actpc:rbretend} and 
\ref{alg:actpc:tobrbinvstart}--\ref{alg:actpc:tobrbinvend}). This way $\wi \ge 
\si$. Hence, we solely use $\wi$ as the approximation of the total value added 
to \ACTnnc when calculating the return value for the $\get$ operations.

Using a gossip protocol allows us to achieve propagation of weak updating 
operations within network partitions, when synchronization involving solving 
distributed consensus is not possible. On the other hand, when solving 
distributed consensus is possible, replicas can agree on the final order in 
which operations will be visible. This way weak operations $\add$ and $\get$ 
are highly available, i.e., they always execute in a constant number of steps
and do not depend on waiting on communication with other replicas. Crucially, 
the return value of the $\get$ operation always reflects all the 
$\add$ operations performed locally and, eventually, all $\add$ 
operations performed within the network partition to which the replica belongs, 
if such a partition exists. On the other hand, the strong $\subtract$ 
operation is applied only if the replicas agree that it is safe to do so.

\ACTnnc guarantees a property which is a conjunction
of \emph{basic eventual consistency} (\BEC) \cite{BGY13} \cite{B14} for weak 
operations ($\add$ and $\get$) and \emph{linearizability} (\LIN) \cite{HW90} 
for strong operations ($\subtract$). We formalize \BEC and \LIN
in Sections \ref{sec:guarantees:bec} and \ref{sec:guarantees:sc},  
and prove the correctness of \ACTnnc in Section~\ref{sec:guarantees:bayou}.

\subsection{Bayou} 

Bayou was an experimental system, so was never optimized for performance. 
However, due to its unique approach 
to speculative execution of transactions and their later \emph{stabilization} 
(establishing the final transaction execution order by a primary replica), 
examining Bayou allows us to discuss various problematic phenomena that stem 
from having both weak and strong semantics in a single system. 
We improve Bayou to form a general-purpose, albeit not 
performance-optimized ACT.


\subsubsection{Protocol overview}

\todo{Add a sentence about read-only operations that do need to be
broadcast to primary and the rest of the replicas.}
Below we give a high-level description of the Bayou protocol. An interested
reader may find a detailed description of Bayou (together with a pseudocode) in 
Appendix~\ref{sec:bayou:details}.

In order to make our analysis more general, we abstract certain aspects 
of the original protocol. Crucially, we allow clients to submit to Bayou 
replicas deterministic, arbitrarily complex (also as complex as, e.g., SQL 
transactions) operations that can provide the clients with a return value. Each 
operation is either \emph{weak} or \emph{strong}, similarly to operations in 
\ACTnnc. Any weak operation is non-blocking with respect to network 
communication, because it is executed locally without any coordination with 
other replicas, but its ultimate impact on the system's state might differ from 
what the client can infer from the return value (if the stabilized execution 
happened differently than the speculative one). On the other hand, the return 
value of a strong operation results from a prior inter-replica synchronization 
and thus can be trusted to never change.

In Bayou, each server speculatively total-orders all received operations
using a simple timestamp-based mechanism and without prior agreement with other 
servers. A unique timestamp is generated by the replica upon receipt of an 
operation from the client. An operation tagged with the timestamp is then 
sent
to all other replicas using some gossip protocol. When a replica 
has a new operation $\opn$ ($\opn$ was directly submitted by a client or it has 
been received 
from other replica), firstly the 
replica  determines the suitable execution order for $\opn$. If $\opn$ 
has the highest timestamp of all operations executed so far by the replica, it 
simply executes $\opn$ (and, if $\opn$ was submitted to the replica by a client 
and $\opn$ is a weak operation, the replica provides the client with a return 
value). Otherwise, the replica rolls back all operations that have higher 
timestamps than $\opn$ (starting from the one with the highest timestamp), 
executes $\opn$ and reexecutes the rolled back operations according to their 
timestamps. This way a single total order consistent with operation timestamps 
is always maintained by all replicas.

The above approach has two major downsides. The first one concerns the 
performance: every time a replica receives an operation with a relatively low 
timestamp (compared to the timestamps of the operations executed most 
recently), 
to maintain the correct execution order, many operations 
need to be rolled back and reexecuted. The second downside is related to the 
provided guarantees: a client that submitted an operation $\opn$ and already 
received a response can never be sure that there will be no other operation 
$\opn'$ with a lower timestamp than $\opn$, which will eventually cause $\opn$ 
to be reexecuted, thus producing possibly a different return value.

In order to mitigate the above two problems, one of the replicas, called the 
\emph{primary}, periodically communicates to the other replicas the \emph{final 
operation execution order}, which is a growing prefix of the operations already 
executed by the primary. Other replicas always honour the decision made by the 
primary, which may force them to adjust their local operation execution orders 
by rolling back and reexecuting some operations. When there are no major 
communication delays between replicas, the final operation execution order 
established by the primary does not deviate much from the order resulting from 
operation timestamps. Hence, replicas do not need to perform many rollbacks and 
reexecutions. Moreover, the operation values obtained during speculative 
execution are mostly correct, i.e., the same as the return values obtained 
during execution of the operations according to the final operation execution 
order. Once the final operation execution for some $\opn$ (weak or 
strong) is established, it will never be reexecuted again. If $\opn$ is a 
strong operation, it is now safe to provide the client with the return value.

Intuitively, the replicas converge to the same state, which is reflected by the 
prefix of operations established by the primary (called the $\committed$ 
list of operations) and the sequence of other operations ordered according to 
their timestamps (the $\tentative$ list of operations). More precisely, when 
the stream of operations incoming to the system ceases and there are no network 
partitions (the replicas can 
reach with the primary), the $\committed$ 
lists at all replicas will be the same, whereas the $\tentative$ lists will be 
empty. On the other hand, when there are partitions, some operations might not 
be successfully committed by the primary, but will be disseminated within a 
partition using a gossip protocol. Then all replicas within the same partition 
will have the same $\committed$ and (non-empty) $\tentative$ lists.

\subsubsection{Anomalies} \label{sec:bayou:anomalies}

\begin{figure*}[t]
\vspace{-0.2cm}
\newcommand{\cwidth}{2.7cm}
\vspace{-0.2cm}
\center

\scalebox{0.85}{
\begin{tabular}{D p{\cwidth} p{\cwidth} p{\cwidth}}
\begin{minipage}{13cm}
\scalebox{1.2}{
{\input{histories/bayou_history_tpds.tex}} 
}
\end{minipage}
& 
\begin{minipage}{\cwidth}
\begin{algorithmic}
\Empty{$u_1\ \{$}{}
\State{x = 1}
\If{y = 1}
    \State{z = 1}
\EndIf
\EndEmpty
\State{$\ \}$}
\end{algorithmic}
\end{minipage}
& 
\begin{minipage}{\cwidth}
\begin{algorithmic}
\Empty{$u_2\ \{$}{}
\State{y = 1}
\If{x = 1}
    \State{z = 2}
\EndIf
\EndEmpty
\State{$\ \}$}
\end{algorithmic}
\end{minipage}
&
\begin{minipage}{\cwidth}
\begin{algorithmic}
\Empty{$q_1/q_2\ \{$}{}
\State{\Return z}
\EndEmpty
\State{$\ \}$}
\State{}
\State{}
\end{algorithmic}
\end{minipage}
\end{tabular}
}
\caption{Example execution of Bayou showing temporary operation reordering and 
circular causality.}
\label{fig:tor}
\end{figure*}
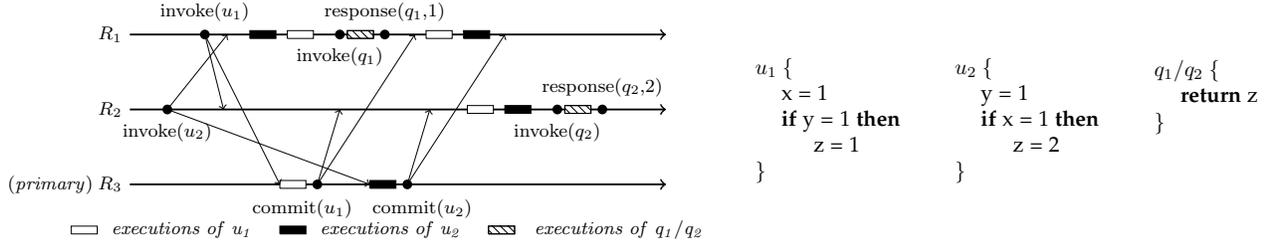

\todo{Say that $u_1$ and $u_2$ are weak.} 
Consider the example in Figure~\ref{fig:tor}, which shows an execution of a 
three-replica Bayou system. Initially, replica $R_1$ executes updating 
operations $u_1$ and $u_2$ in order $u_2, u_1$, which corresponds to $u_1$'s 
and $u_2$'s timestamps. This operation execution order is observed by the 
client who issues query $q_1$. On the other hand, $R_2$ executes the 
operations according to the final execution order ($u_1, u_2$), as 
established by the primary replica $R_3$. Hence, the client who issued query 
$q_2$ observes a different execution order than the client who issued $q_1$. 
Note that replicas execute the operations with a delay (e.g., due to CPU 
being busy) and that $R_1$ reexecutes the operations once it gets to 
know the final order.

Clearly, the clients that issued the operations can infer from the return 
values the order in which Bayou executed the operations. The observed operation 
execution orders differ between the clients accessing $R_1$ and $R_2$. This 
kind of anomaly, which we call \emph{temporary operation reordering}, cannot be 
avoided since Bayou uses two, inconsistent with each other, ways in which 
operations are ordered (the timestamp order and the order established by the 
primary, which may occasionally differ from the timestamp order). 
This behaviour is not present in strongly consistent systems, which ensure that 
a single global ordering of operation execution is always respected (e.g., 
\cite{Lam78} \cite{CBPS10}). The majority of eventually consistent systems 
which trade consistency for high availability are also free of this anomaly, 
as they only use one method for ordering concurrent operations (e.g., 
\cite{Tho79} \cite{BGY13}), or support only commutative operations (as in 
\emph{strong eventual consistency} \cite{SPBZ11}, e.g. \cite{SPBZ11a} 
\cite{AEM15}). There are also protocols that allow some operations to perceive 
the past events in different (but still legal) orders (e.g., \cite{HP86} 
\cite{PRR+12} \cite{LPC+12}). But, unlike Bayou, they do not require the 
replicas to eventually agree on a single execution order for all operations.
Interestingly, temporary operation reordering is not present in \ACTnnc, 
because weak updating operations ($\add$) commute and do not provide clients 
with the return values. 

Bayou exhibits another anomaly, which comes as very non-intuitive, i.e., 
\emph{circular causality}. By analysing the return values of queries $q_1$ and 
$q_2$ one may conclude that there is a circular dependency between $u_1$ 
and $u_2$: $u_1$ depends on $u_2$ as evidenced by $q_1$'s response, while $u_2$ 
depends on $u_1$ as evidenced by $q_2$'s response 
(the cycle of 
causally related operations can contain more operations). Interestingly, 
as we show later, circular causality does not directly follow from temporary 
operation reordering but is rather a result of the way Bayou rolls back and 
reexecutes some operations. 

In the original Bayou protocol, application-specific conflict detection and 
resolution is accomplished through the use of \emph{dependency checks} and 
\emph{merge procedure} mechanisms. Since we allow operations 
with arbitrary complex semantics, the dependency checks and the merge 
procedures can be emulated by the operations themselves, by simply 
incorporating \emph{if-else} statements: the dependency check as 
the \emph{if} condition, and the merge procedure in the \emph{else} 
branch (as suggested in the original paper \cite{TTPD+95}). Hence, these 
mechanisms do not alleviate the anomalies outlined above.

\subsubsection{Correctness guarantees} \label{sec:bayou:guarantees}

Because of the phenomena described above, the guarantees provided by Bayou 
cannot be formalized using the correctness criteria used for contemporary 
eventually consistent systems based on CRDTs. E.g., \emph{basic eventual 
consistency} (\BEC) by Burckhardt \etal \cite{BGY13} \cite{B14} (mentioned 
briefly when discussing \ACTnnc's guarantees) directly 
forbids circular causality (see Section~\ref{sec:guarantees:bec} for 
definition of \BEC). \BEC also requires the relative order of any two 
operations, as perceived by the client, to be consistent and to never change. 
Similarly, \emph{strong eventual consistency} (\SEC) by Shapiro \etal 
\cite{SPBZ11} requires any two replicas that delivered the same updates to have 
equivalent states.\footnote{\BEC can be seen as a refinement of \SEC, which 
abstracts away from CRDTs implementation details and ensures that no return 
value is constructed out of thin air.} Obviously, Bayou neither satisfies \BEC 
nor \SEC (as evidenced by Figure~\ref{fig:tor}). On the other hand informal 
definitions of eventual consistency which admit temporal reordering, such as 
\cite{V09}, involve only liveness guarantees, which is insufficient. Bayou  
fulfills the operational specification in \cite{FGL+96}. However, we are 
interested in declarative specifications, similar in the style to popular 
consistency criteria, such as \emph{sequential consistency} \cite{L79}, or 
\emph{serializability} \cite{P79}, through which we can concisely define the 
behaviour of a wide class of systems. Hence we introduce a new 
correctness criterion, \emph{fluctuating eventual consistency} (\FEC), 
which can be viewed as a generalization of \BEC (see 
Section~\ref{sec:guarantees:fec} for definition). \FEC relaxes \BEC, so 
different operations can perceive different operation orders. However, we 
require that the different perceived operation orders converge to one final 
execution order. Hence, \FEC is suitable for systems that feature temporary
operation reordering.

Similarly to \ACTnnc, Bayou also ensures linearizability for strong operations 
(a response of a strong operation $\opn$ always reflects the serial execution 
of all stabilized operations up to the point of $\opn$'s commit). 
In Section~\ref{sec:guarantees:bayou} we formally prove that the Bayou-derived 
general-purpose ACT satisfies the above correctness criteria.

In Appendix~\ref{sec:bayou:liveness}, an interested reader may find a brief 
analysis of Bayou's liveness guarantees.

\subsubsection{Fault-tolerance}

Bayou's reliance on the primary means that it provides only limited 
fault-tolerance. Even though the primary may recover, when it is down, 
operations do not stabilize, and thus no strong operation can complete. Hence, 
the primary is the single point of failure. Alternatively, the primary could be 
replaced by a distributed commit protocol. If two-phase-commit (2PC) 
\cite{BHG87} is used, the phenomena illustrated in Figure~\ref{fig:tor} are not
possible. However, in this approach, a failure of any replica blocks the 
execution of strong operations (in 2PC all the replicas need to be 
operational in order to reach distributed agreement). 
On the other hand, if a \emph{non-blocking} commit protocol, e.g., one that 
utilizes a quorum-based implementation of TOB is used (as in \ACTnnc), the 
system may stabilize operations despite (a limited number of) 
failures.\footnote{Sharded 2PC \cite{CG18} can be considered non-blocking, 
if within each shard at least one process remains 
operational at all times. Then, in such a scheme not every process needs to be 
contacted to commit a transaction, thus it falls under the quorum-based 
category.} As we prove later, ACTs (which do not depend on synchronous 
communication with all the replicas and thus can operate despite failures of 
some of them) with general-purpose semantics similar to Bayou, are necessarily 
prone to the phenomena described above.

\subsubsection{The improved Bayou protocol} \label{sec:bayou:modified}

Bayou can be improved to make it more fault-tolerant 
and free of some of the phenomena described above.

Firstly, we use TOB in place of the primary to establish the final operation 
execution order. More precisely, each time a replica receives an operation 
$\opn$ from a client, it still disseminates $\opn$ using a gossip protocol (so 
it can reach at least all replicas within the same network partition) but it 
also broadcasts the operation using TOB (so in a similar way in which weak 
updating operations are handled in \ACTnnc). Since TOB guarantees that all 
replicas deliver the same set of messages in the same order, all replicas will 
stabilize the same set of operations in the same order. As we argued earlier, 
TOB can be implemented in a way that avoids a single point of failure 
\cite{Lam98}.

The second modification is aimed at eliminating circular causality in Bayou.
To this end (1) strong operations are broadcast using TOB and never a gossip 
protocol, and (2) upon being submitted, a weak operation $\opn$ is executed 
immediately on the current state to produce the return value; if $\opn$'s 
timestamp is not the highest timestamp of all already executed operations, 
$\opn$ is then rolled back and eventually executed in the order consistent with 
its timestamp. In Appendix~\ref{sec:bayou_modified} we formally prove that 
above changes to the protocol allow us to avoid circular causality. 

The modified variant of Bayou does not ensure that subsequent 
operations invoked by the same client observe the effects of previous ones, 
even if 
they are issued on the same replica (the \emph{read-your-writes} 
session guarantee \cite{TDPS+94}).\footnote{In the original Bayou system, 
clients were colocated with replicas and the \emph{read-your-writes} guarantee 
was naturally provided. In our approach, such guarantees can be provided on 
the client side.} 

With the above modifications the improved Bayou protocol becomes the 
general-purpose ACT called \ACTBayou.

\subsection{ANNC vs \ACTBayou}

\ACTnnc and \ACTBayou greatly differ in the offered semantics and complexity.
Note that we could trivially implement a non-negative integer counter using 
\ACTBayou by executing each counter operation as a separate \ACTBayou 
operation, albeit such an implementation would be suboptimal: in some cases 
the operations would have to be rolled back and temporary operation reordering 
would be possible again. Still, we can consider \ACTBayou as a generic ACT, 
capable of executing any set of weak and strong operations. 

Despite the many differences between \ACTnnc and \ACTBayou, they share several 
design assumptions, which are common to all ACT implementations. Firstly, 
in order to facilitate high availability and low latency responses
(which are essential in geo-replicated environments), 
frequently invoked operations should be defined as weak operations and replicas 
should process them similarly to operations in CRDTs (automatically resolve 
conflicts between concurrent updates; converge to the same state within a 
network partition). To enforce this behaviour without resorting to distributed 
agreement, we impose the same assumptions as Attiya \etal for highly available 
eventually consistent data stores in \cite{AEM15} (see also 
Section~\ref{sec:act:restrictions}).
Secondly, when weak consistency guarantees are insufficient, strong operations 
can be used. Strong operations use a global agreement protocol for 
inter-replica synchronization, e.g., TOB. We require that strong operations do 
not block the execution of weak operations and that they do not require all 
replicas to be operational at all times in order to complete (as 
in 2PC).

ACTs are meant to provide the programmer with a modular abstraction layer that 
handles all the complexities of replication, while enabling flexibility, high 
performance and clear mixed-consistency semantics. In the next section we 
specify ACTs formally.

%% file: algorithms/act_counter_modif.tex
\begin{algorithm}[t] 
\caption{Acute non-negative counter (\ACTnnc) for replica $R_i$}
\label{alg:actnnc}  
  \footnotesize
\begin{algorithmic}[1]
\State{\Struct $\Request$($\type$ : ${\{}$ADD, SUBTRACT${\}}$, $\val$ : int, $\idd$ : pair$\langle$int, int$\rangle$)}
\State{\Var $\wi$, $\si$, $\sd$, $\currentEventNumber$ : int} \label{alg:actpc:vars}
\State{\Var $\requestsAwaitingResponse$ : set$\langle$pair$\langle$int, int$\rangle\rangle$}
\State{\Var $\deliveredAdds$ : set$\langle$pair$\langle$int, int$\rangle\rangle$}
\Procedure{$\add$}{$\val$ : int} \LineComment{weak operation} \label{alg:actpc:add}
    \State{$\currentEventNumber = \currentEventNumber + 1$}
    \State{$\wi = \wi + \val$}
    \State{$r = \Request($ADD, $\val, (i, \currentEventNumber))$}
    \State{$\rbcast(r)$} \label{alg:actpc:addrb}
    \State{$\tobcast(r)$} \label{alg:actpc:addtob}
\EndProcedure
\Function{$\subtract$}{$\val$ : int} \LineComment{strong operation} \label{alg:actpc:subtract}
    \State{$\currentEventNumber = \currentEventNumber + 1$}
    \State{$r = \Request($SUBTRACT, $\val, (i, \currentEventNumber))$}
    \State{$\tobcast(r)$}
    \State{$\requestsAwaitingResponse = \requestsAwaitingResponse \cup \{ \idd \}$}
\EndFunction
\Upon{$\rbdeliver$}{$r$ : $\Request($ADD, $\val, \idd)$} 
    \If{$r.\dot.\mathit{first} = i \vee r.\dot \in \deliveredAdds$} \LineComment{$r$ issued locally} \label{alg:actpc:rbretstart}
        \State{\Return} \LineComment{or TOB-deliver happened before RB-deliver}
    \EndIf \label{alg:actpc:rbretend}
    \State{$\deliveredAdds = \deliveredAdds \cup \{ r.\dot \}$}
    \State{$\wi = \wi + \val$}
\EndUpon
\Upon{$\tobdeliver$}{$r$ : $\Request($ADD, $\val, \idd)$} 
    \If{$r.\dot \not\in \deliveredAdds$} \label{alg:actpc:tobrbinvstart}
        \State{invoke $\rbdeliver(r)$} \LineComment{RB-deliver always before TOB-deliver}
    \EndIf \label{alg:actpc:tobrbinvend}
    \State{$\si = \si + \val$} \label{alg:actpc:tobinc}
\EndUpon
\Upon{$\tobdeliver$}{$r$ : $\Request($SUBTRACT, $\val, \idd)$} 
    \State{\Var $\res = \si \geq \sd + \val$} \label{alg:actpc:subcalcstart}
    \If{$\res$} \label{alg:actpc:tobdecstart}
        \State{$\sd = \sd + \val$} \label{alg:actpc:subcalcend}
    \EndIf \label{alg:actpc:tobdecend}
    \If{$\idd \in \requestsAwaitingResponse$}
        \State{$\requestsAwaitingResponse = \requestsAwaitingResponse \setminus \{ \idd \}$}
        \State{return $\res$ to client} 
    \EndIf
\EndUpon
\Function{$\get$()}{} \LineComment{read-only, weak operation} \label{alg:actpc:get}
    \State{\Return $\wi - \sd$} \label{alg:actpc:getc}
\EndFunction
\end{algorithmic}
\end{algorithm}

%% file: histories/bayou_history_tpds.tex
\begin{tikzpicture}[font=\sf\footnotesize] 

\tikzstyle{time line} = [->, thick];
\tikzstyle{bcast line} = [->, dashed, thick];
\tikzstyle{tobcast line} = [->];
\tikzstyle{dot} = [circle, fill, inner sep=1pt];

\newcommand{\dist}{1.15}
\newcommand{\procdist}{2.3}
\newcommand{\tlength}{0.4}
\newcommand{\mypattern}{north west lines}

\node (r1) at (0, 1 * \procdist) [left] {$R_1$};
\node (r2) at (0, 0.5 * \procdist) [left] {$R_2$};
\node (r3) at (0, 0.0 * \procdist) [left] {$\mathit{(primary)}\ R_3$};

\draw [time line] (r1) -- (7.15 * \dist, 1 * \procdist);
\draw [time line] (r2) -- (7.15 * \dist, 0.5 * \procdist);
\draw [time line] (r3) -- (7.15 * \dist, 0.0 * \procdist);


\abdotname{it1}{1.0 * \dist, 1.0 * \procdist}{$\invoke(u_1)$};

\emptydot{r1rect2}{1.3 * \dist, 1.0 * \procdist};

\emptydot{r1t2}{1.6 * \dist, 1.0 * \procdist};
\emptydot{r1t2e}{$(r1t2) + (\tlength, 0)$};
\mysegfill{r1t2}{r1t2e};

\emptydot{r1t1}{2.1 * \dist, 1.0 * \procdist};
\emptydot{r1t1e}{$(r1t1) + (\tlength, 0)$};
\mysegdraw{r1t1}{r1t1e};

\beldotname{iq}{2.8 * \dist, 1.0 * \procdist}{$\invoke(q_1)$};

\emptydot{r1q}{2.9 * \dist, 1.0 * \procdist};
\emptydot{r1qe}{$(r1q) + (\tlength, 0)$};
\mysegfillcustom{r1q}{r1qe}{\mypattern};

\abdotname{responseQ}{3.4 * \dist, 1.0 * \procdist}{$\responserm(q_1{,} 1)$};

\emptydot{r1recct1}{3.8 * \dist, 1.0 * \procdist};
\emptydot{r1recct2}{5.0 * \dist, 1.0 * \procdist};

\emptydot{r1t1r}{3.95 * \dist, 1.0 * \procdist};
\emptydot{r1t1re}{$(r1t1r) + (\tlength, 0)$};
\mysegdraw{r1t1r}{r1t1re};

\emptydot{r1t2r}{4.45 * \dist, 1.0 * \procdist};
\emptydot{r1t2re}{$(r1t2r) + (\tlength, 0)$};
\mysegfill{r1t2r}{r1t2re};


\beldotname{it2}{0.5 * \dist, 0.5 * \procdist}{$\invoke(u_2)$};

\emptydot{r2rect1}{1.25 * \dist, 0.5 * \procdist};

\emptydot{r2recct1}{2.8 * \dist, 0.5 * \procdist};
\emptydot{r2recct2}{4.0 * \dist, 0.5 * \procdist};

\emptydot{r2t1}{4.5 * \dist, 0.5 * \procdist};
\emptydot{r2t1e}{$(r2t1) + (\tlength, 0)$};
\mysegdraw{r2t1}{r2t1e};

\emptydot{r2t2}{5.0 * \dist, 0.5 * \procdist};
\emptydot{r2t2e}{$(r2t2) + (\tlength, 0)$};
\mysegfill{r2t2}{r2t2e};

\beldotname{invokeQ}{5.7 * \dist, 0.5 * \procdist}{$\invoke(q_2)$};

\emptydot{r2q}{5.8 * \dist, 0.5 * \procdist};
\emptydot{r2qe}{$(r2q) + (\tlength, 0)$};
\mysegfillcustom{r2q}{r2qe}{\mypattern};

\abdotname{responseQ}{6.3 * \dist, 0.5 * \procdist}{$\responserm(q_2{,} 2)$};


\emptydot{r3t1}{2.0 * \dist, 0.0 * \procdist};
\emptydot{r3t1e}{$(r3t1) + (\tlength, 0)$};
\mysegdraw{r3t1}{r3t1e};

\beldotname{ct1}{2.5 * \dist, 0.0 * \procdist}{};
\belemptydotname{ct1d}{2.3 * \dist, 0.0 * \procdist}{$\commit(u_1)$};

\emptydot{r3t2}{3.2 * \dist, 0.0 * \procdist};
\emptydot{r3t2e}{$(r3t2) + (\tlength, 0)$};
\mysegfill{r3t2}{r3t2e};

\beldotname{ct2}{3.7 * \dist, 0.0 * \procdist}{};
\belemptydotname{ctd2}{3.9 * \dist, 0.0 * \procdist}{$\commit(u_2)$};


\draw [tobcast line] (it1.center) -- (r2rect1.center);
\draw [tobcast line] (it1.center) -- (r3t1.center);

\draw [tobcast line] (it2.center) -- (r1rect2.center);
\draw [tobcast line] (it2.center) -- (r3t2.center);

\draw [tobcast line] (ct1.center) -- (r1recct1.center);
\draw [tobcast line] (ct1.center) -- (r2recct1.center);

\draw [tobcast line] (ct2.center) -- (r1recct2.center);
\draw [tobcast line] (ct2.center) -- (r2recct2.center);


\begin{scope}[shift={(-0.1, 0.3)}]
\emptydot{a1}{-0.8, -1};
\emptydot{a1e}{-0.4, -1};
\mysegdraw{a1}{a1e};
\node (t1) at (0.9, -1.0) {$\mathit{executions\ of\ u_1}$};

\emptydot{a1}{2.4, -1};
\emptydot{a1e}{2.8, -1};
\mysegfill{a1}{a1e};
\node (t1) at (4.1, -1) {$\mathit{executions\ of\ u_2}$};

\emptydot{a1}{5.6, -1};
\emptydot{a1e}{6.0, -1};
\mysegfillcustom{a1}{a1e}{\mypattern};
\node (t1) at (7.6, -1) {$\mathit{executions\ of\ q_1/q_2}$};
\end{scope}

\end{tikzpicture}

%% file: 3a-system_model.tex
\section{Acute Cloud Types} \label{sec:act}

\subsection{Definition}

\emph{An acute cloud type} is an abstract data type, implemented as a 
replicated data structure, that offers a precisely defined set of operations, 
divided into two groups: weak and strong. The operations can be either updating 
or read-only (\RO), and all operations are allowed to provide a return value 
(in Section~\ref{sec:framework} we show how the semantics of operations can be 
specified formally). 
We impose the following \emph{implementation restrictions} over ACTs:
\emph{invisible reads}, \emph{input-driven processing}, 
\emph{op-driven messages}, \emph{highly available weak operations} and 
\emph{non-blocking strong operations}. The first four, are adapted from the 
definition of \emph{write-propagating data stores} \cite{AEM15}. They 
guarantee \emph{genuine}, low-latency, eventually-consistent processing for 
weak operations (as in, e.g., CRDTs \cite{SPBZ11}). The last restriction 
guarantees that strong operations are implemented using a non-blocking 
agreement protocol, instead of a non-fault-tolerant approach requiring all the 
replicas to be operational. In Sections \ref{sec:act:model} and 
\ref{sec:act:restrictions} we formalize the system model and provide precise 
definitions of the implementation restrictions.

\subsection{System model} \label{sec:act:model}

\subsubsection{Replicas and clients}

We consider a system consisting of $n \ge 2$ processes called \emph{replicas}, 
which maintain full copies of an ACT\footnote{Partial 
replication is orthogonal to our work. We assume full replication for 
simplicity.} and to which external clients submit requests in the form of 
operations to be executed. Each operation invoked by a client is marked either 
\emph{weak} or \emph{strong}. Replicas communicate with each other through 
message passing. We assume the availability of a gossip protocol, which is used 
when ordering constraints are not necessary, and some global agreement 
protocol, used for tasks that require solving distributed consensus. For 
simplicity, as in Algorithm~\ref{alg:actnnc}, we formalize these protocols 
using \emph{reliable broadcast (RB)} \cite{CGR11}, and TOB, respectively. 
Replicas can implement point-to-point communication simply by ignoring messages 
for which they are not the intended recipient. We model replicas as 
deterministic state machines, which execute atomic steps in reaction to 
external events (e.g., operation invocation or message delivery), and
can execute internal events (e.g., scheduled processing of rollbacks). 
A specific event is \emph{enabled} on a replica, if its preconditions are 
met (e.g., an $\rbdeliver(m)$ event is enabled on a replica $R$, if $m$ was 
previously $\rbcast$ and $R$ has not yet delivered the message $m$). Replicas 
have access to a local clock, which advances monotonically, but we 
make no assumptions 
on the bound on clock drift between replicas.

We model crashed replicas as if they stopped all computation (or compute 
infinitely slowly). We say that a replica is \emph{faulty} if it crashes (in an 
infinite execution it executes only a finite number of steps). Otherwise, it 
is \emph{correct}.

\subsubsection{Network properties} 
\label{sec:act:model:network}

In a fully asynchronous system, a crashed replica is indistinguishable to its 
peers from a very slow one, and it is impossible to solve the distributed 
consensus problem \cite{FLP85}. Real distributed systems which exhibit some 
amount of \emph{synchrony} can usually overcome this limitation. For example, 
in a quasi-synchronous model \cite{VA95}, the system is considered to be 
synchronous, but there exist a non-negligible probability that timing 
assumptions can be broken. We are interested in the behaviour of protocols, 
both in the fully \emph{asynchronous} environment, when timing assumptions are 
consistently broken (e.g. because of prevalent network partitions), and in a 
\emph{stable} one, when the minimal amount of synchrony is available so that 
consensus eventually terminates. Thus, we consider two kinds of \emph{runs}: 
\emph{asynchronous runs} and \emph{stable runs}. Replicas are not aware which 
kind of a run they are currently executing. In stable runs, we augment the 
system with the failure detector $\Omega$ (which is an abstraction for the 
synchronous aspects of the system). We do so implicitly by allowing the 
replicas to use TOB through the $\tobcast$ and $\tobdeliver$ primitives. 
Since, TOB is known to require a failure detector at least as strong as 
$\Omega$ to terminate \cite{CHT96}, we guarantee it achieves progress only in 
stable runs.

In both asynchronous and stable runs we guarantee the basic 
properties of reliable message passing \cite{CGR11}, i.e.: 
\begin{itemize}
\item if a message is $\rbdeliver$ed, or $\tobdeliver$ed, then it was, 
respectively, $\rbcast$, or $\tobcast$, by some replica, 
\item no message is $\rbdeliver$ed, or $\tobdeliver$ed, more than once by 
the same replica, 
\item if a correct replica $\rbcast$s some message, then eventually it 
$\rbdeliver$s it, 
\item if a correct replica $\rbdeliver$s some message, then eventually all 
correct replicas $\rbdeliver$ it, 
\item if any (correct or faulty) replica $\tobdeliver$s some message, then 
eventually all correct replicas, $\tobdeliver$ it, 
\item messages are $\tobdeliver$ed by all replicas in the same total 
order. 
\end{itemize}
We define $\tobNo(m)$ as the sequence number of the $\tobdeliver(m)$ event 
(among other $\tobdeliver$ events in the execution)
on any replica  (we leave it undefined, i.e., $\tobNo(m) = \bot$, if $m$ is 
never $\tobdeliver$ed by any replica). 

Solely in stable runs, we also guarantee the following: 
\begin{itemize}
\item if a correct replica $\tobcast$s some message, then eventually all 
correct replicas $\tobdeliver$ it.
\item if a message $m$ was both $\rbcast$ and $\tobcast$ by some 
(correct or faulty) replica, and $m$ was $\rbdeliver$ed by some correct 
replica, then eventually all correct replicas $\tobdeliver$ it. 
\end{itemize}
The last guarantee is non-standard for a total-order broadcast, but could be 
easily emulated by the application itself. We include it to simplify 
presentation of certain algorithms, such as \ACTnnc and \ACTBayou.

%

\subsubsection{Fair executions} \label{sec:act:model:correctness}

An execution is \emph{fair}, if each replica, has a chance to execute its steps 
(all replicas execute infinitely many steps of each type of an enabled event, 
e.g., infinitely many $\rbdeliver$ events for infinitely many messages 
$\rbcast$).

We analyze the correctness of a protocol by evaluating a single arbitrary 
infinite fair execution of the protocol, similarly to \cite{B14} and 
\cite{AEM17}. If the execution satisfies the desired properties, then all the 
executions of the protocol (including finite ones and the ones featuring 
crashed replicas) satisfy all the safety aspects verified (\emph{nothing bad 
ever happens} \cite{L77} \cite{AS87}). Additionally, all fair executions of the 
protocol satisfy liveness aspects (\emph{something good eventually happens}).

\subsection{Implementation restrictions}
\label{sec:act:restrictions}

Below we state the five rules that ACTs need to adhere to.

\textbf{1. Invisible reads.} Replicas do not change their state due to 
an invocation of a weak read-only operation. Formally, for each weak read-only 
operation $\opn$ invoked on a replica $R$ in state $\sigma$, the state of $R$ 
after a response for $\opn$ is returned is equal $\sigma$. Note that, the 
consequence of this is that weak read-only operations need to return a response 
to the client immediately in the invoke event, without executing any other 
steps. We allow strong read-only operations to change the state of a replica, 
because sometimes it is necessary to synchronize with other replicas, and the 
replica needs to note down that a response is pending. 

\textbf{2. Input-driven processing.} Replicas execute a series of 
steps only in response to some external stimulus, e.g., an operation 
invocation or a received message. 
A state $\sigma$ of a replica $R$ 
is \emph{passive} if none of the internal events on the replica are enabled in 
$\sigma$. Initially each replica is in a passive state. An external event may 
bring a replica to an \emph{active} state $\sigma'$ in which it has some 
internal events enabled. Then, after executing a finite number of internal 
events (when no new external events are executed), the replica enters a passive 
state. More formally, for each replica $R$, we require that in a given 
execution, either there is only a finite number of internal events executed on 
$R$, or there is an infinite number of external events executed on $R$. We say 
that $R$ is \emph{passive}, if it is in a passive state, otherwise it is 
\emph{active}.

\textbf{3. Op-driven messages.} RB or TOB messages are only generated and 
sent as a result of some not read-only client operation, and not spontaneously 
or in response to a received message. 
More formally, a message can be $\rbcast$ or $\tobcast$ by a replica $R$, if 
previously some not read-only operation was invoked on $R$, and since then $R$ 
did not enter a passive state.

\textbf{4. Highly available weak operations.} Weak operations need to 
eventually return a response without communicating with other replicas. A weak 
operation $\opn$ may remain pending only if the execution is finite, and the 
executing replica remains active since the invocation of $\opn$ (in an infinite 
execution a pending weak operation is never allowed). 


\textbf{5. Non-blocking strong operations.} Strong operations need 
to eventually return a response if a global agreement has been reached. 
More formally, for a strong operation $\opn$ invoked on a replica $R$, let 
$\msgs$ be the set of all messages $\tobcast$ by $R$ since the invocation of 
$\opn$ but before $R$ enters a passive state. Then, $\opn$ may remain pending 
only if:
\begin{itemize}
\item the execution is finite, and $R$ remains active since the invocation of 
$\opn$, or $R$ remains active because of the delivery of any message $m \in 
\msgs$, or 
\item 
there exists a message $m \in \msgs$, which has not been $\tobdeliver$ed by $R$
yet. 
\end{itemize} 
It means that in order to execute a strong operation replicas may 
synchronize by $\tobcast$ing multiple messages, but once TOB completes, the 
response must be returned in a finite number of steps.
%

All the above requirements are commonly met by various eventually consistent 
data stores and CRDTs (when we consider them as ACTs with only weak operations 
and using our communication model\footnote{In case of geo-replicated systems 
which are weakly consistent between data centers, but feature state machine 
replication within a data center to simulate reliable processes, we can 
consider the whole data center as a single replica.}), e.g., \cite{BGHS13} 
\cite{BGYZ14} \cite{DHJK+07} \cite{PZBD+14} \cite{LFKA11} \cite{SPBZ11} 
\cite{DIRZ14} \cite{AEM15} \cite{AEM17}. The restrictions 1--4 are 
inspired by the ones defined for write-propagating data stores \cite{AEM15}, 
but modified appropriately to accommodate for the more complex nature of ACTs. 
In particular, we allow implementations which do not execute each invoked 
operation in one atomic step, but divide the execution between many internal 
steps (e.g., see the pseudocode of Bayou in Appendix~\ref{sec:bayou:details}). 
On the other hand, the 5th requirement concerns strong operations, and so is 
specific for ACTs. As discussed at length in \cite{AEM15} \cite{AEM17}, 
requirements 1--4 preclude implementations that offer stronger consistency 
guarantees but do not provide a real value to the programmer (and still 
fall short of the guarantees possible to ensure if global agreement can be 
reached). For example, without invisible reads, it is possible to propose an 
implementation of a register's read operation, which returns the most up to 
date value written to the register only after it returned stale values to a 
similar call for a fixed number of times (even though the newest value was 
already available). On the other hand, \emph{with} the above restrictions, it 
is still possible to attain causal consistency and variants of it, such as 
\emph{observable causal consistency} \cite{AEM15}.

%% file: 3b-framework.tex

\section{Formal framework} \label{sec:framework}

Below we provide the formalism that allows us to reason about execution 
histories and correctness criteria. We extended the framework by  
Burckhardt \etal \cite{BGY13}\iftoggle{proofs}{,}{\cite{B14} }
(also used by several other researchers, e.g., \cite{BEH14} 
\cite{AEM15} \cite{AEM17} \cite{VV16}\iftoggle{proofs}{, see also \cite{B14} 
for a textbook tutorial).}{).}

\subsection{Preliminaries} \label{sec:framework:preliminaries}

\noindent \textbf{Relations:} A binary relation $\rela$ over set $A$ is a 
subset $\rela \subseteq A \times A$. For $a,b \in A$, we use the notation $a 
\raaa{\rela} b$ to denote $(a,b) \in \rela$, and the notation $\rela(a)$ to 
denote $\{b \in A : a \raaa{\rela} b \}$. We use the notation $\rela^{-1}$ to 
denote the inverse relation, i.e. $(a \raaa{\rela^{-1}} b) \iff (b \raaa{\rela} 
a)$. Therefore, $\rela^{-1}(b) = \{a \in A : a \raaa{\rela} b \}$. Given two 
binary relations $\rela$, $\rela'$ over $A$, we define the composition $\rela ; 
\rela' = \{(a,c) : \exists b \in A : a \raaa{\rela} b \raaa{\rela'} c \}$. We 
let $\id_A$ be the identity relation over $A$, i.e., $(a \raaa{\id_A} b) \iff 
(a \in A) \wedge (a = b)$. For $n \in \mathbb{N}_0$, we let $\rela^n$ be the 
n-ary composition $\rela; \rela ... ; \rela$, with $\rela^0 = \id_A$. We let 
$\rela^+ = \bigcup_{n \geq 1} \rela^n$ and $\rela^* = \bigcup_{n \geq 0} 
\rela^n$. For some subset $A' \subseteq A$, we define the restricted relation 
$\rela|_{A'} = \rela \cap (A' \times A')$. In Figure~\ref{fig:relations} 
we summarize various properties of relations. 

We define by $\words(A)$ the set of all sequences (words) containing only 
elements from the set $A$. When not ambiguous we use $A^*$ to denote 
$\words(A)$ (i.e. when $A$ is not a binary relation).

Let $\rank$ be a function that encounts elements of a set $A$ that are in 
relation $\rela$ to element $a \in A$: $\rank(A, \rela, a) = |\{x \in A : 
x \raaa{\rela} a\}|$. Thus, $\rank(A, \rela, a) = |\rela^{-1}(a) \cap A|$.

We also define two operators $\sort$ and $\foldr$. $A.\sort(\rela) \in 
A^*$ arranges in an ascending order the elements of set $A$ according to the 
total order $\rela$. $\foldr(a_0, f, w) \in A$ reduces sequence $w \in B^*$ 
by one element at a time using the function $f : A \times B \to A$ and 
accumulator $a_0 \in A$: 
\begin{myeq}
\foldr(a_0, f, w) &= 
\left\{
\begin{array}{@{}ll@{}}
     a_0 & \text{if}\ w = \epsilon \\
     f(\foldr(a_0, f, w'), b) & \text{if}\ w = w'b
\end{array}\right. \\
\end{myeq}

\renewcommand\figurescale{0.85}
\begin{figure*}[t]
\begin{tabular}{cc}
\scalebox{\figurescale}{
\hspace{1.0cm}
\begin{tabular}{l | l | l}
Property & Element-wise Definition & Algebraic Definition\\
         & $\forall x,y,z \in A :$ &\\
\hline
$\symmetric$ & $x \raaa{\rela} y \Rightarrow y \raaa{\rela} x$ & $\rela = 
\rela^{-1}$\\
$\reflexive$ & $x \raaa{\rela} x$ & $\id_A \subseteq \rela$\\
$\irreflexive$ & $x \centernot{\raaa{\rela}} x$ & $\id_A \cap \rela = 
\emptyset$\\
$\transitive$ & $(x \raaa{\rela} y \raaa{\rela} z) \Rightarrow (x \raaa{\rela} 
z)$ &
$(\rela; \rela) \subseteq \rela$\\
$\acyclic$ & $\neg(x \raaa{\rela} ... \raaa{\rela} x)$ & $\id_A \cap \rela^+ = 
\emptyset$\\
$\total$ & $ x \neq y \Rightarrow (x \raaa{\rela} y \vee y \raaa{\rela} x)$ & 
$\rela \cup \rela^{-1} \cup \id_A = A \times A$
\end{tabular}
}
&
\hspace{-0.2cm}
\raisebox{-0.1cm}{
\scalebox{\figurescale}{
\renewcommand{\arraystretch}{1.3}
\begin{tabular}{l | l}
Property & Definition\\
\hline
$\natural$ & $\forall x \in A : |\rela^{-1}(x)| < \infty$\\
$\partialorder$ & $\irreflexive \wedge \transitive$\\
$\totalorder$ & $\partialorder \wedge \total$\\
$\enumeration$ & $\totalorder \wedge \natural$\\
$\equivalencerelation$ & $\reflexive \wedge \transitive$ \\[-5pt]
& $\wedge \symmetric$
\end{tabular}}
}
\end{tabular}
\caption{Definitions of common properties of a binary 
relation $\rela \subseteq A \times A$.}
\label{fig:relations}
\end{figure*}

\smallskip

\noindent \textbf{Event graphs:} To reason about executions of a distributed
system we encode the information about events that occur in the system and 
about various dependencies between them in the form of an \emph{event graph}.
An \emph{event graph} $G$ is a tuple $(E, d_1, ...., d_n)$, 
where $E \subseteq \Events$ is a finite or countably infinite set of events
drawn from universe $\Events$, 
$n \geq 1$, and each $d_i$ is an attribute or a relation over $E$. 
\emph{Vertices} in $G$ represent events that occurred 
at some point during the execution and are interpreted as opaque identifiers. 
\emph{Attributes} label vertices with information pertinent to the 
corresponding event, e.g., operation performed, or the value returned. All 
possible operations of all considered data types form the $\Operations$ set. 
All possible return values of all operations form the $\Values$ set. 
\emph{Relations} represent orderings or groupings of events, and thus can be 
understood as \emph{arcs} or \emph{edges} of the graph.

Event graphs are meant to carry information that is independent of the actual
elements of $\Events$ chosen to represent the events (the attributes and 
relations in $G$ encode all relevant information regarding the execution). Let 
$G = (E, d_1, ...., d_n)$ and $G' = (E', d'_1, ...., d'_n)$ be two event 
graphs. $G$ and $G'$ are \emph{isomorphic}, written $G \simeq G'$, if (1) for 
all $i \ge 1$, $d_i$ and $d'_i$ are of the same kind (attribute vs. relation) 
and (2) there exists a bijection $\phi : E \rightarrow E'$ such that for all 
$d_i$, where $d_i$ is an attribute, and all $x \in E$, we have $d_i(x) = 
d'_i(\phi(x))$, and such that for all $d_i$ where $d_i$ is a relation, and all 
$x,y \in E$, we have $x \raaa{d_i} y \iff \phi(x) \raaa{d'_i} \phi(y)$.

\subsection{Histories}

We represent a high-level view of a system execution as a \emph{history}.
We omit implementation details such as message exchanges or internal steps 
executed by the replicas. We include only the observable behaviour of the 
system, as perceived by the clients through received responses. Formally, we 
define a \emph{history} as an event graph $H = (E, \op, \rval, \rb, \sss, 
\lvl)$, where:
\begin{itemize}
\item $\op : E \rightarrow \Operations$, specifies the operation invoked in a 
particular event, e.g., $\op(e) = \wwr(3)$, 
\item $\rval : E \rightarrow \Values \cup \{\nabla\}$, specifies the value 
returned by the operation, e.g., $\rval(e) = 3$, or $\rval(e') = \nabla$, if 
the operation never returns ($e'$ is \emph{pending} in $H$), 
\item $\rb$, the \emph{returns-before} relation, is a natural partial order 
over $E$, which specifies the ordering of \emph{non-overlapping} operations 
(one operation returns before the other starts, in real-time), 
\item $\sss$, the \emph{same session} relation, is an equivalence relation 
which groups events executed within the same session (the same client), and 
finally
\item $\lvl : E \rightarrow \{ \weak, \strong \}$, specifies the consistency
level demanded for the operation invoked in the event.
\end{itemize}

We consider only \emph{well-formed} histories, for which the following holds:
\begin{itemize}
\item $\forall a,b \in E : (a \rarb b \Rightarrow \rval(a) \neq \nabla)$ (a 
pending operation does not return),
\item $\forall a,b,c,d \in E : (a \rarb b \wedge c \rarb d) \Rightarrow (a 
\rarb d \vee c \rarb b)$ ($\rb$ is an \emph{interval order} \cite{G76}),
\item for each event $e \in E$ and its session $S = \{e' \in E : e \raaa{\sss} 
e' \}$, the restriction $\rb|_S$ is an enumeration (clients issue operations 
sequentially).
\end{itemize}

\subsection{Abstract executions}

In order to \emph{explain} the history, i.e., the observed return values, and 
reason about the system properties, we need to extend the history with 
information about the abstract relationships between events. For strongly 
consistent systems typically we do so by finding a \emph{serialization} 
\cite{L79} (an enumeration of all events) that satisfies certain criteria.
For weaker consistency models, such as \emph{eventual consistency} or 
\emph{causal consistency}, it is more natural to reason about partial ordering 
of events. Hence, we resort to \emph{abstract executions}.

An \emph{abstract execution} is an event graph $A = (E, \op, \rval, \rb, \sss, 
\lvl, \vis, \ar, \perc)$, such that:
\begin{itemize}
\item $(E, \op, \rval, \rb, \sss, \lvl)$ is some history $H$, 
\item $\vis$ is an acyclic and natural relation, 
\item $\ar$ is a total order relation, and
\item $\perc : E \rightarrow 2^{E \times E}$ is a function which returns
a binary relation in $E$.
\end{itemize}
For brevity, we often use a shorter notation $A = (H, \vis, \ar, \perc)$ and 
let 
$\HH(A) = H$. Just as serializations are used to explain and justify 
operations' return values reported in a history, so are the \emph{visibility} 
($\vis$) and \emph{arbitration} ($\ar$) relations. \emph{Perceived arbitration} 
($\perc$) is a function which is necessary to formalize temporary operation 
reordering.

\emph{Visibility} ($\vis$) describes the relative influence of operation 
executions in a history on each others' return values: if $a$ is visible to $b$ 
(denoted $a \ravis b$), then the effect of $a$ is visible to the replica 
performing $b$ (and thus reflected in the $b$'s return value). Visibility often 
mirrors how updates propagate through the system, but it is not tied to the 
low-level phenomena, such as message delivery. It is an acyclic and natural 
relation, which may or may not be transitive.
Two events are \emph{concurrent} if they are not ordered by visibility.

\emph{Arbitration} ($\ar$) is an additional ordering of events which is 
necessary in case of non-commutative operations. It describes how the effects 
of these operations should be applied. If $a$ is arbitrated before $b$ (denoted 
$a \raar b$), then $a$ is considered to have been executed earlier than $b$. 
Arbitration is essential for resolving conflicts between concurrent events, but 
it is defined as a total-order over \emph{all} operation executions in a 
history. It usually matches whatever conflict resolution scheme is used in an 
actual system, be it physical time-based timestamps, or logical clocks.

\emph{Perceived arbitration} ($\perc$) describes the relative order of 
operation executions, as perceived by each operation ($\perc(e)$ defines the 
total order of all operations, as perceived by event $e$). If $\forall e \in E 
: \perc(e) = \ar$, then there is no temporary operation reordering in $A$. 

\subsection{Correctness predicates}

A \emph{consistency guarantee} $\PP(A)$ is a set of conditions on an abstract 
execution $A$, which depend on the particulars of $A$ up to isomorphism. For 
brevity we usually omit the argument $A$. We write $A \models \PP$ if $A$ 
\emph{satisfies} $\PP$. More precisely: $A \models \PP \iffdef \forall A' : A' 
\simeq A : \PP(A')$. A history $H$ is \emph{correct} according to some 
consistency guarantee $\PP$ (written $H \models \PP$) if it can be extended 
with some $\vis$ and $\ar$ relations to an abstract execution $A = (H, 
\vis, \ar, \perc)$ that satisfies $\PP$. 
We say that a system is correct according to some consistency guarantee $\PP$ 
if all of its histories satisfy $\PP$.

We say that a consistency guarantee $\PP_i$ is at least as strong as a
consistency guarantee $\PP_j$, denoted $\PP_i \geq \PP_j$, if $\forall H : H 
\models \PP_i \Rightarrow H \models \PP_j$. If $\PP_i \geq \PP_j$ and $\PP_j 
\not\geq \PP_i$ then $\PP_i$ is stronger than $\PP_j$, denoted $\PP_i > \PP_j$. 
If $\PP_i \not\geq \PP_j$ and $\PP_j \not\geq \PP_i$, then $\PP_i$ and $\PP_j$ 
are incomparable, denoted $\PP_i \lessgtr \PP_j$.

\subsection{Replicated data type} \label{sec:framework:rdt}

\begin{figure*}
\scalebox{0.8}{
\begin{minipage}{0.8\textwidth}
\begin{myeq}
\footnotesize
  \FFmvr(\wwr(v), (E, \op, \vis, \ar)) &= \ook\\
  \FFmvr(\rrd, (E, \op, \vis, \ar)) &=   \{v : \exists e \in E : \op(e) =
\wwr(v) \wedge \nexists e' \in E : \op(e') = \wwr(v') \wedge e \ravis e'\} \\
\\
  \FFseq(\app(s), (E, \op, \vis, \ar)) &= \ok\\
  \FFseq(\rd, (E, \op, \vis, \ar)) &= 
  \left\{
  \begin{array}{@{}ll@{}}
    \epsilon, & \text{if}\ \mathsf{appends}(E) = \emptyset \\
    s_1 \cdot s_2 \cdot \ldots \cdot s_n & \text{otherwise}, \text{where}\ 
    \forall{i \leq n}: 
    \exists{e \in E}: \rank(\mathsf{appends}(E), \ar, e) = i \\
    & \hspace{5.2cm} \wedge\ \op(e) = \app(s_i) \wedge s_i \in \$^*
  \end{array}\right. \\
  \text{where}\ \$ = \{a, b, c, \ldots, z\}\ & \text{and} 
  \ \mathsf{appends}(E) = \{e \in E : \exists v \in \$^* : \op(e) = \app(v)\} \\
  \\
  \FFnnc(\mathit{add}(v), (E, \op, \vis, \ar)) &= \ok\\
  \FFnnc(\mathit{subtract}(v), (E, \op, \vis, \ar)) &= 
  \left\{
  \begin{array}{@{}ll@{}}
    \true & \text{if}\ \foldr(0, \fnnc, E.sort(\ar)) \geq v \\
    \false & \text{otherwise} 
  \end{array}\right. \\
  \FFnnc(\mathit{get}, (E, \op, \vis, \ar)) &= \foldr(0, \fnnc, E.\sort(\ar))\\
  \text{where}\ v \in \mathbb{N}\ \text{and}\ \fnnc &= 
  \left\{
  \begin{array}{@{}ll@{}}
    \fnnc(x, \mathit{add}(v)) &= x + v  \\
    \fnnc(x, \mathit{subtract}(v)) &= x - v\ \text{if}\ x \geq v\ \text{or}\ 
    x\ \text{otherwise} \\
    \fnnc(x, \mathit{get}) &= x \\
  \end{array}\right. \\
\end{myeq}
\end{minipage}
}
\caption{Formal specifications of a multi-value register data type $\FFmvr$,  
an append-only sequence data type $\FFseq$, and a non-negative counter $\FFnnc$. 
An instance of $\FFmvr$ stores multiple values when there are concurrent 
$\wwr()$ operations ($\wwr()$ operations not ordered by the $\vis$ relation). 
$\FFseq$ can be used to create a sequence of characters (a word), where the set 
of characters is limited to $a$ through $z$. $\FFseq$ features two operations: 
$\app(x)$, which appends $x$ to the end of the sequence and returns $\ok \in 
\Values$, and $\rd()$, which returns a sequence (a word) $w \in \$^*$. $\FFnnc$ 
stores an integer value, that can be increased using the $\mathit{add}$ 
operation or decreased using the $\mathit{subtract}$ operation, but only if 
the value of the counter will not decrease below 0. The $\mathit{get}$ 
operation simply returns the current value of the counter. See the definition 
of operators $\sort$ and $\foldr$ in Section~\ref{sec:framework:preliminaries}.}
\label{fig:ffseq}
\end{figure*}

In order to specify semantics of operations invoked by the clients on the 
replicas, we model the whole system as a single replicated object (as in case 
of Algorithm~\ref{alg:actnnc}). Even though we use only a single 
object, this approach is general, as multiple objects can be viewed as a single 
instance of a more complicated type, e.g. multiple registers constitute a 
single \emph{key-value store}. Defining the semantics of the replicated object 
through a sequential specification \cite{HW90} is not sufficient for 
replicated objects which expose concurrency to the client, e.g. multi-value 
register (MVR) \cite{SPBZ11} or observed-remove set (OR-Set) \cite{SPBZ11a}. 
Hence, we utilize \emph{replicated data types} specification \cite{BGYZ14}. 

In this approach, the state on which an operation $\opn \in \Operations$ 
executes, called the \emph{operation context}, is formalized by the event graph 
of the prior operations visible to $\opn$. Formally, for any event $e$ in 
an abstract execution $A = (E, \op, \rval, \rb, \sss, \lvl, \vis, \ar, \perc)$, 
the operation context of $e$ in $A$ is the event graph $\context(A,e) \eqdef 
(\vis^{-1}(e), \op, \vis, \ar)$. Note that an operation context lacks return 
values, the returns-before relation, \iftoggle{proofs}{as well as}{and} the 
information about sessions. The set of previously invoked 
operations\iftoggle{proofs}{, coupled with}{ and} their relative visibility and 
arbitration unambiguously \iftoggle{proofs}{should define}{defines} the output 
of \iftoggle{proofs}{the operation invoked in the considered event}{each 
operation}. This brings us to the formal definition of a replicated data type.

A \emph{replicated data type} $\FF$ is a function that, for each operation 
$\opn \in \ops(\FF)$ (where $\ops(\FF) \subseteq \Operations$) and operation 
context $C$, defines the expected return value $v = \FF(\opn, C) \in \Values$, 
such that $v$ does not depend on events, i.e., is the same for isomorphic 
contexts: $C \simeq C' \Rightarrow \FF(\opn, C) = \FF(\opn, C')$ for all 
$\opn$, $C$, $C'$. We say that $\opn \in \ops(\FF)$ is a read-only operation 
(denoted $\opn \in \readonlyops(\FF)$), if and only if, for any operation 
$\opn'$, context $C = (E, \op, \vis, \ar)$ and event $e \in E$, such that 
$\op(e) = \opn$, $\FF(\opn', C) = \FF(\opn', C')$, where $C' = (E \setminus 
\{e\}, \op, \vis, \ar)$. In other words, read-only operations can be 
excluded from any context $C$, producing $C'$, and the result of any operation 
$\opn'$ will not change. 

In Figure~\ref{fig:ffseq} we give the specification of three replicated data 
types: $\FFmvr$ (a multi-value register), $\FFseq$ (an append-only sequence), 
and $\FFnnc$ (a non-negative counter). We use $\FFseq$ in the subsequent 
sections to illustrate various consistency models. 


\subsection{ACT specification} \label{sec:framework:act_spec}

To accommodate for the mixed-consistency nature of ACTs we 
extend 
replicated data type specification with the information on supported 
consistency levels for a given operation. Thus, we define \emph{ACT 
specification} as a pair $(\FF, \lvlmap)$, where $\FF$ is a replicated data 
type specification and $\lvlmap : \Operations \rightarrow 2^{\{\weak, 
\strong\}}$ is a function which specifies for each $\opn \in \Operations$ with 
which consistency levels it can be executed. We assume that clients follow this 
contract, and thus, when considering a history $H = (E, \op, \rval, \rb, \sss, 
\lvl)$ of an ACT compliant with the specification $(\FF, \lvlmap)$, we assume 
that for each $e \in E : \lvl(e) \in \lvlmap(\op(e))$.

Then, \ACTnnc's specification is 
is 
$(\FFnnc, \lvlmap_\mathit{NNC})$, where 
$\lvlmap_\mathit{NNC}(\mathit{add}) = \lvlmap_\mathit{NNC}(\mathit{get}) = 
\{\weak\}$ and $\lvlmap_\mathit{NNC}(\mathit{subtract}) = \{\strong\}$.

%% file: 4-correctness_criteria.tex

\section{Correctness guarantees} \label{sec:guarantees}

In this section we define various correctness guarantees for ACTs. We define 
them as conjunctions of several basic predicates. We start with two simple 
requirements that should naturally be present in any eventually consistent 
system. For the discussion below we assume some arbitrary abstract execution
$A = (E, \op, \rval, \rb, \sss, \lvl, \vis, \ar, \perc)$.

\subsection{Key requirements for eventual consistency} \label{sec:key_reqs}

The first requirement is the \emph{eventual visibility} ($\EV$) of events. 
$\EV$ requires that for any event $e$ in $A$, there is only a 
finite number of events in $E$ that do not observe $e$. Formally: 
$$\EV \eqdef \forall e \in E : |\{ e' \in E : e \rarb e' \wedge e \not\ravis e' 
\}| < \infty$$
Intuitively, $\EV$ implies progress in the system because replicas must 
synchronize and exchange knowledge about operations submitted to the system.

The second requirement concerns avoiding circular causality, as discussed in 
Section~\ref{sec:bayou:anomalies}. To this end we define two auxiliary 
relations: \emph{session order} and \emph{happens-before}. The session order 
relation $\so \eqdef \rb \cap \sss$ represents the order of operations in each 
session. The happens-before relation $\hb \eqdef (\so \cup \vis)^+$ (a 
transitive closure of session order and visibility) allows us to express the 
causal dependency between events. Intuitively, if $e \raaa{\hb} e'$, then $e'$ 
potentially depends on $e$. We simply require \emph{no circular causality}:
$$\NCC \eqdef \acyclic(\hb)$$

In the following sections we add requirements on the return values of the 
operations in $A$. Formalizing the properties of ACTs which, similarly to 
\ACTBayou, admit temporary operation reordering, requires a new approach. We 
start, however with the traditional one.

\subsection{Basic Eventual Consistency} \label{sec:guarantees:bec}

\begin{figure*}
\scalebox{0.8}{
\hspace{-0.5cm}
\begin{minipage}{1.0\textwidth}
\begin{tabular}{c p{3.95in} c p{3.95in}}
\Abec & \raisebox{-.5\height}{\input{histories/bec.tex}} &
\Afec & \raisebox{-.5\height}{\input{histories/fec.tex}}\\
\Asc & \raisebox{-.5\height}{\input{histories/seq.tex}} & 
\Alin & \raisebox{-.5\height}{\input{histories/lin.tex}}\\
\end{tabular}
\end{minipage}
}
\caption{Example abstract executions of systems with a list semantics that 
satisfy $\BEC(\FFseq)$, $\FEC(\FFseq)$, $\SC(\FFseq)$, and $\LIN(\FFseq)$
respectively (for brevity, we omit the level parameter $l$ and assume 
that all operations belong to the same class $l$). We use solid and dashed 
underlines to depict which updating operations are visible (through relation 
$\vis$) in $A$ to the $\rd()$ operations (we assume that every $\rd()$ 
operation observes all other $\rd()$ operations that happened prior to it). In 
the arbitration order, $\app(a)$ precedes $\app(b)$, and both updates are 
followed by all the reads in the order they appear on the timeline.}
\label{fig:aex}
\end{figure*}
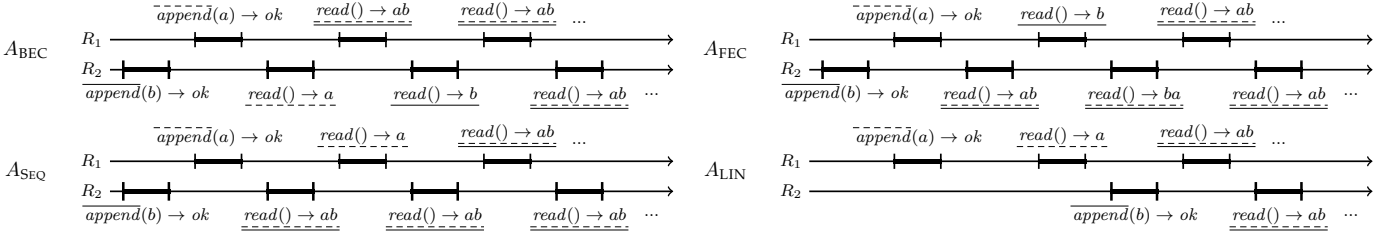

Intuitively, \emph{basic eventual consistency} ($\BEC$) \cite{BGY13} \cite{B14},
in addition to $\EV$ and $\NCC$, requires that the return value of each 
invoked operation can be explained using the specification of the 
replicated data type $\FF$, what is formalized as follows: 
$$\RVAL(\FF) \eqdef \forall e \in E : \rval(e) = \FF(\op(e), \context(A,e))$$ 
Then: 
$$\BEC(\FF) \eqdef \EV \wedge \NCC \wedge \RVAL(\FF)$$

An example abstract execution \Abec that satisfies $\BEC(\FFseq)$ is shown in 
Figure~\ref{fig:aex}. In \Abec, firstly replicas $R_1$ and $R_2$ concurrently 
execute two $\app()$ operations, and then each replica executes an 
infinite number of $\rd()$ operations. Consider the $\rd()$ operations on 
$R_2$: the first one observes only $\app(a)$ 
(which is in the operation 
context of $\rd()$), whereas the second observes only $\app(b)$. \BEC admits
this kind of execution, because it does not make any requirements in terms of 
session guarantees \cite{TDPS+94}. Eventually, both $\app(a)$ and $\app(b)$ 
become visible to all subsequent $\rd()$ operations, thus satisfying \EV. 

By the definition of the $\context$ function (Section~\ref{sec:framework:rdt}), 
when $A$ satisfies $\RVAL(\FF)$, the return value of each operation is 
calculated according to the $\ar$ relation. It is then easy to see that there 
are executions of \ACTBayou (or other ACTs that admit temporary operation 
reordering) that do not satisfy $\RVAL(\FF)$. It is because weak operations (as 
shown in Section~\ref{sec:bayou:anomalies}), might observe past operations in 
an order that differs from the final operation execution order ($\ar$). Hence 
\ACTBayou does not satisfy $\BEC(\FF)$ for arbitrary $\FF$. It still, though, 
could satisfy $\BEC(\FF)$ for a sufficiently simple $\FF$, such as a 
conflict-free counter, in which all operations always commute (as opposed to 
$\FFnnc$). It is so, because then, even if \ACTBayou reorders some operations 
internally, the final result never changes and thus the reordering cannot be 
observed by the clients.

\subsection{Fluctuating Eventual Consistency} \label{sec:guarantees:fec}

In order to admit temporary operation reordering, we give a slightly different
definition of the $\context$ function, in which the arbitration order 
\emph{fluctuates}, i.e., it changes from one event to another. Let 
$\fcontext(A,e) \eqdef (\vis^{-1}(e),\op, \vis,\perc(e))$, which means that now 
we consider the operation execution order as perceived by $e$, and not the 
final one. The definition of the fluctuating variant of $\RVAL$ 
is straightforward:
$$\FRVAL(\FF) \eqdef \forall e \in E : \rval(e) = \FF(\op(e), \fcontext(A,e))$$
To define the fluctuating variant of $\BEC$, that could be used to formalize 
the guarantees provided by ACTs we additionally must ensure, that the 
arbitration order perceived by events is not completely unrestricted, but that 
it gradually \emph{converges} to $\ar$ for each subsequent event. It means that 
each $e \in E$ can be temporarily observed by the subsequent events $e'$ 
according to an order that differs from $\ar$ (but is consistent with 
$\perc(e')$). However, from some moment on, every event $e'$ will observe $e$ 
according to $\ar$. To define this requirement, we use the $\rank$ function 
(defined in Section~\ref{sec:framework:preliminaries}). Let $E_e = \{e' \in E: 
e \ravis e'\}$. This intuition is formalized by \emph{convergent perceived 
arbitration}:
\begin{multline*}
\CPERC \eqdef \forall e \in E : |\{e' \in E_e : 
\rank(\vis^{-1}(e'), \perc(e'), e) \\ \neq \rank(\vis^{-1}(e'), \ar, e)\}| 
< \infty
\end{multline*}
If $A$ satisfies $\CPERC$, then for each event $e$, the set of events $e'$, 
which observe the position of $e$ not according to $\ar$ is finite. Thus, the 
position of $e$ in $\perc(e')$ for subsequent events $e'$ stabilizes, and 
$\perc(e')$ eventually converges to $\ar$.

Now we can define our new consistency criterion \emph{fluctuating eventual 
consistency} ($\FEC$):
$$\FEC(\FF) \eqdef \EV \wedge \NCC \wedge \FRVAL(\FF) \wedge \CPERC$$

An example abstract execution \Afec that satisfies \FEC is shown in 
Figure~\ref{fig:aex}. In \Afec, replica $R_2$ temporarily observes the $\app()$ 
operations in the order $\app(b), \app(a)$ which is different then the 
eventual operation execution order (as evidenced by the infinite number of 
$\rd() \rightarrow ab$ operations). We call this behaviour \emph{fluctuation}.

It is easy to see that $\FEC(\FF) < \BEC(\FF)$, in the sense that: for each 
$\FF$, $\FEC(\FF) \leq \BEC(\FF)$, and for some $\FF$, $\FEC(\FF) < \BEC(\FF)$.
It is so, because \FEC uses $\perc$ instead of $\ar$ to calculate the return 
values of operation executions, but $\perc$ eventually converges to $\ar$. 
Hence, $\BEC(\FF)$ is a special case of $\FEC(\FF)$, when $\forall e \in E : 
\perc(e) = \ar$. It is easy to see that \Abec from Figure~\ref{fig:aex} 
satisfies both \BEC and \FEC, whereas \Afec satisfies only \FEC.


\subsection{Operation levels}

The above definitions can be used to capture the guarantees provided by a 
wide variety of eventually consistent systems. 
However, our framework still lacks the capability to express the semantics of 
mixed-consistency systems. ACTs offer different guarantees for different 
classes of operations (e.g., consistency guarantees stronger than \BEC or \FEC 
are provided only for strong operations in \ACTBayou or \ACTnnc). Hence, we 
need to parametrize the consistency criteria with a level attribute (as 
indicated by the $\lvl$ function for each event). Since, consistency level is 
specified per operation invocation, we need to assure that the respective 
operations' responses reflect the demanded consistency level.

Let us revisit \BEC first. Let $L = \{ e \in E : \lvl(e) = l\}$ for a given 
$l$. Then: 
\begin{myeq}
\EV(l) \eqdef&\ \forall e \in E : |\{ e' \in L : e \rarb e' \wedge e \not\ravis 
e' \}| < \infty\\
\NCC(l) \eqdef&\ \acyclic(\hb \cap (L \times L))\\
\RVAL(l, \FF) \eqdef&\ \forall e \in L : \rval(e) = \FF(\op(e), 
\context(A,e))\\
\BEC(l, \FF) \eqdef&\ \EV(l) \wedge \NCC(l) \wedge \RVAL(l, \FF)
\end{myeq}
The above parametrized definition of \BEC restricts the $\RVAL$ predicate only 
to events issued with the given consistency level $l$ (the events that belong 
to the set $L$). It means that for any such event the response has to conform 
with the replicated data type specification $\FF$, and with the $\vis$ and 
$\ar$ relations (as defined by the definition of the $\context$ function). For 
all other events this requirement does not need to be satisfied, so they 
can return arbitrary responses (unless restricted by other predicates targeted 
for these events). Similarly, for $\EV$ and $\NCC$, the predicates are 
restricted to affect only the events from the set $L$. In case 
of $\EV$, each event eventually becomes visible to the operations executed with 
the level $l$. In case of $\NCC$, there must be no cycles in $\hb$ involving 
events from the set $L$.

The parametrized variant of \FEC is formulated analogously. Let $L$ be as 
defined before, and for any event $e \in E$, let $L_e = \{e' \in L: e \ravis 
e'\}$ be the subset of events from $L$ which observe $e$. Then:
\begin{myeq}
\FRVAL(l, \FF) \eqdef&\ \forall e \in L : \rval(e) = \FF(\op(e), 
\fcontext(A,e))\\
\CPERC(l) \eqdef&\ \forall e \in E : |\{e' \in L_e: \rank(\vis^{-1}(e'), 
\perc(e'), e) \\ 
&\neq \rank(\vis^{-1}(e'), \ar, e)\}| < \infty \\
\FEC(l, \FF) \eqdef&\ \EV(l) \wedge \NCC(l) \wedge \FRVAL(l, \FF) \wedge 
\CPERC(l)
\end{myeq}
As before, we restrict the return values only for the events from the set $L$. 
Additionally, we restrict the predicate $\CPERC$, so that $\perc(e)$ converges 
towards $\ar$ only for events $e \in L$. Other events can differently perceive 
the arbitration of events (in principle, the observed arbitration can be 
completely different from the final one, specified by $\ar$).

We can compare the parametrized variants of \BEC and \FEC as before: 
$\FEC(l, \FF) < \BEC(l, \FF)$. 

All of the strong consistency criteria which we are going to discuss next, we 
define already in the parametrized form with the given level $l$ in mind, so 
they can be used for, e.g., strong operations in \ACTBayou and \ACTnnc.

\subsection{Strong consistency} \label{sec:guarantees:sc}

A common feature of strong consistency criteria, such as \emph{sequential 
consistency} \cite{L79}, or \emph{linearizability} \cite{HW90}, is a single 
global serialization of all operations. It means that a history satisfies these 
criteria, if it is equivalent to some serial execution (\emph{serialization}) 
of all the operations. 
Additionally, depending on the particular criterion, the serialization must, 
e.g., respect program-order, or real-time order of operation executions. 
Although we provide a serialization of all operations (through the total order 
relation $\ar$, which is part of every abstract execution), the equivalence of 
a history to the serialization is not enforced in the correctness criteria we 
have defined so far. For example, given a sequence of three events $\langle a, 
b, c \rangle$, such that $a \raar b \raar c$, the response of $c$ according to 
\BEC, does not need to reflect neither $a$, nor $b$, as they may simply be not 
visible to $c$ ($a \not\ravis c \vee b \not\ravis c$). Thus, to guarantee 
conformance to a single global serialization, we must enforce that for any two 
events $e_1, e_2 \in E$, $e_1 \raar e_2 \iff e_1 \ravis e_2$ (unless $e_1$ is 
pending, since a pending operation might be arbitrated before a completed one, 
yet still be not visible). We express this through the following predicate, 
\emph{single order}:
\begin{myeq}
\SO \eqdef&\ \exists E' \subseteq \rval^{-1}(\nabla) : 
\vis = \ar \setminus (E' \times E)\\
\SO(l) \eqdef&\ \exists E' \subseteq \rval^{-1}(\nabla) : 
\vis_L = \ar_L \setminus (E' \times E)\\
\text{where}\ \vis_L =&\ \vis \cap (E \times L)\ \text{and} \
\ar_L = \ar \cap (E \times L)
\end{myeq}
In the parametrized form, the conformance to the serialization is required 
only for the events from the set $L$ (but the serialization includes all the 
events). 

In order to capture the eventual \emph{stabilization} of the operation 
execution order, which happens in \ACTBayou and in ACTs similar to it, we now 
define two additional correctness criteria that feature $\SO$.


\textbf{Sequential consistency.} 
Informally, \emph{sequential consistency} ($\SC$) \cite{L79} guarantees that 
the system behaves as if all operations were executed sequentially, but in an 
order that respects the \emph{program order}, i.e., the order in which 
operations were executed in each session. Hence, $\SC$ implies $\RVAL(\FF)$ 
together with $\SO$, and additionally, \emph{session arbitration} ($\SA$). 
$\SA$ simply requires that for any two events $e,e' \in E$, if $e \raso e'$, 
then $e \raar e'$. In the parametrized form we are interested only in the 
guarantees for events in the set $L$, and thus we use $\so_L = \so \cap (E 
\times L)$ instead of $\so$ (see Section~\ref{sec:key_reqs}). $\SO$ together 
with $\SA$ imply $\NCC$ and $\EV$ \cite{B14}, however this does not hold for 
the parametrized forms of these predicates. Thus, we define $\SC$ by extending 
$\BEC$ (which explicitly includes $\EV$ and $\NCC$):
\begin{myeq}
\SA(l) \eqdef&\ \so_L \subseteq \ar\\
\SC(l, \FF) \eqdef&\ \SO(l) \wedge \SA(l) \wedge \BEC(l, \FF)
\end{myeq}

An example abstract execution \Asc that satisfies \SC is shown in 
Figure~\ref{fig:aex}. According to \SC, since the $\app()$ operations are 
arbitrated $\app(a), \app(b)$ (as evidenced by any $\rd()$ operation that 
observes both $\app()$ operations), any $\rd()$ operation can either return 
$ab$ or $a$, a non-empty prefix of $ab$. 

\textbf{Linearizability.} The correctness condition of 
\emph{linearizability} ($\LIN$) \cite{HW90} is similar to $\SC$ but 
instead of program order it enforces a stronger requirement called 
\emph{real-time order}. Informally, a system that is linearizable guarantees 
that for any operation $\opn'$ that starts (in real-time) after any operation 
$\opn$ ends, $\opn'$ will observe the effects of $\opn$. 
We formalize $\LIN$ using the \emph{real-time order} ($\RT$) predicate, 
that uses the $\rb_L = \rb \cap (L \times L)$ relation in its parametrized 
form: 
\begin{myeq}
\RT(l) \eqdef&\ \rb_L \subseteq \ar \\
\LIN(l, \FF) \eqdef&\ \SO(l) \wedge \RT(l) \wedge \BEC(l, \FF)
\end{myeq}

Note that, $\SC$ and $\LIN$ are incomparable in their parametrized forms. 
While $\LIN(l, \FF)$ requires any two events to be arbitrated according to 
real-time if they both belong to $L$, $\SC(l, \FF)$ enforces real-time only 
within the same session, but only one of the events needs to belong to $L$.
By using a stronger definition of $\RT'(l)$ 
with $\rb'_L = \rb \cap (E \times L)$, 
we would force all operations to synchronize, which is incompatible with 
high availability of weak operations.

An example abstract execution \Alin that satisfies \LIN is shown in 
Figure~\ref{fig:aex}. According to \LIN, since $\app(a)$ ended before 
$\app(b)$ started, the operations must be arbitrated $\app(a), \app(b)$ (as 
evidenced by any $\rd()$ operation that observes both $\app()$ operations). If 
some $\rd()$ operation started after $\app(a)$ ended but executed concurrently 
with $\app(b)$ ($\app(b)$ would start before $\rd()$ ended), $\rd()$ could 
return either $a$ or $ab$.

%% file: histories/bec.tex
\begin{tikzpicture}[font=\sf\footnotesize] 

\tikzstyle{time line} = [->, thick];
\tikzstyle{dot} = [circle, fill, inner sep=1pt];

\node (r1) at (0,1) [left] {$R_1$};
\node (r2) at (0,0.5) [left] {$R_2$};

\newcommand{\dist}{1.2}

\draw [time line] (r2) -- (7.8 * \dist, 0.5);
\draw [time line] (r1) -- (7.8 * \dist, 1);

\abseg{1.5 * \dist, 1.0}{\odashed{$\app$}$(a) \rightarrow \ok$};
\abseg{3.5 * \dist, 1.0}{\udashedsolid{$\rd() \rightarrow ab$}};
\abseg{5.5 * \dist, 1.0}{\udashedsolid{$\rd() \rightarrow ab$}};

\belseg{0.5 * \dist, 0.5}{\osolid{$\app$}$(b) \rightarrow \ok$};
\belseg{2.5 * \dist, 0.5}{\udashed{$\rd() \rightarrow a$}};
\belseg{4.5 * \dist, 0.5}{\usolid{$\rd() \rightarrow b$}};
\belseg{6.5 * \dist, 0.5}{\udashedsolid{$\rd() \rightarrow ab$}};

\node (dots1) at (6.5 * \dist, 1.3) {$...$};
\node (dots2) at (7.5 * \dist, 0.1) {$...$};

\end{tikzpicture}

%% file: histories/fec.tex
\begin{tikzpicture}[font=\sf\footnotesize] 

\tikzstyle{time line} = [->, thick];
\tikzstyle{dot} = [circle, fill, inner sep=1pt];

\node (r1) at (0,1) [left] {$R_1$};
\node (r2) at (0,0.5) [left] {$R_2$};

\newcommand{\dist}{1.2}

\draw [time line] (r2) -- (7.8 * \dist, 0.5);
\draw [time line] (r1) -- (7.8 * \dist, 1);

\abseg{1.5 * \dist, 1.0}{\odashed{$\app$}$(a) \rightarrow \ok$};
\abseg{3.5 * \dist, 1.0}{\usolid{$\rd() \rightarrow b$}};
\abseg{5.5 * \dist, 1.0}{\udashedsolid{$\rd() \rightarrow ab$}};

\belseg{0.5 * \dist, 0.5}{\osolid{$\app$}$(b) \rightarrow \ok$};
\belseg{2.5 * \dist, 0.5}{\udashedsolid{$\rd() \rightarrow ab$}};
\belseg{4.5 * \dist, 0.5}{\udashedsolid{$\rd() \rightarrow ba$}};
\belseg{6.5 * \dist, 0.5}{\udashedsolid{$\rd() \rightarrow ab$}};

\node (dots1) at (6.5 * \dist, 1.3) {$...$};
\node (dots2) at (7.5 * \dist, 0.1) {$...$};

\end{tikzpicture}

%% file: histories/seq.tex
\begin{tikzpicture}[font=\sf\footnotesize] 

\tikzstyle{time line} = [->, thick];
\tikzstyle{dot} = [circle, fill, inner sep=1pt];

\node (r1) at (0,1) [left] {$R_1$};
\node (r2) at (0,0.5) [left] {$R_2$};

\newcommand{\dist}{1.2}

\draw [time line] (r2) -- (7.8 * \dist, 0.5);
\draw [time line] (r1) -- (7.8 * \dist, 1);

\abseg{1.5 * \dist, 1.0}{\odashed{$\app$}$(a) \rightarrow \ok$};
\abseg{3.5 * \dist, 1.0}{\udashed{$\rd() \rightarrow a$}};
\abseg{5.5 * \dist, 1.0}{\udashedsolid{$\rd() \rightarrow ab$}};

\belseg{0.5 * \dist, 0.5}{\osolid{$\app$}$(b) \rightarrow \ok$};
\belseg{2.5 * \dist, 0.5}{\udashedsolid{$\rd() \rightarrow ab$}};
\belseg{4.5 * \dist, 0.5}{\udashedsolid{$\rd() \rightarrow ab$}};
\belseg{6.5 * \dist, 0.5}{\udashedsolid{$\rd() \rightarrow ab$}};

\node (dots1) at (6.5 * \dist, 1.3) {$...$};
\node (dots2) at (7.5 * \dist, 0.1) {$...$};

\end{tikzpicture}

%% file: histories/lin.tex
\begin{tikzpicture}[font=\sf\footnotesize] 

\tikzstyle{time line} = [->, thick];
\tikzstyle{dot} = [circle, fill, inner sep=1pt];

\node (r1) at (0,1) [left] {$R_1$};
\node (r2) at (0,0.5) [left] {$R_2$};

\newcommand{\dist}{1.2}

\draw [time line] (r2) -- (7.8 * \dist, 0.5);
\draw [time line] (r1) -- (7.8 * \dist, 1);

\abseg{1.5 * \dist, 1.0}{\odashed{$\app$}$(a) \rightarrow \ok$};
\abseg{3.5 * \dist, 1.0}{\udashed{$\rd() \rightarrow a$}};
\abseg{5.5 * \dist, 1.0}{\udashedsolid{$\rd() \rightarrow ab$}};

\belseg{4.5 * \dist, 0.5}{\osolid{$\app$}$(b) \rightarrow \ok$};
\belseg{6.5 * \dist, 0.5}{\udashedsolid{$\rd() \rightarrow ab$}};

\node (dots1) at (6.5 * \dist, 1.3) {$...$};
\node (dots2) at (7.5 * \dist, 0.1) {$...$};

\end{tikzpicture}

%% file: 5-correctness_of_bayou.tex
\subsection{Correctness of \ACTnnc and \ACTBayou} \label{sec:guarantees:bayou}

%

Having defined \BEC, \FEC and \LIN, we show four formal results: two regarding 
\ACTnnc and two regarding \ACTBayou. The proofs of all four theorems can be 
found in Appendix~\ref{sec:proofs}.

As we have discussed in Section~\ref{sec:act:model:network}, we are interested 
in the behaviour of systems, both in fully asynchronous environment, when 
timing assumptions are consistently broken (e.g., because of 
prevalent network partitions), and in a stable one, when the minimal 
amount of synchrony is available so that consensus eventually terminates. Thus, 
we consider two kinds of runs: asynchronous and stable. 

\begin{restatable}{theorem}{actnncstable}
In stable runs \ACTnnc satisfies $\BEC(\weak, \FFnnc) \wedge \LIN(\strong, 
\FFnnc)$.
\end{restatable}

\begin{restatable}{theorem}{actnncasynchronous}
In asynchronous runs \ACTnnc satisfies $\BEC(\weak, \FFnnc)$ and it does not 
satisfy 
$\LIN(\strong, \FFnnc)$.
\end{restatable}


\ACTnnc does not guarantee $\LIN(\strong, \FFnnc)$ in asynchronous runs, 
because strong operations in general (for arbitrary $\FF$) cannot be 
implemented without solving global agreement, and since in asynchronous runs 
TOB completion is not guaranteed, some of the operations may remain pending. It 
means that for some $e \in E$, such that $\lvl(e) = \strong$, $\rval(e) = 
\nabla$, even though it is not allowed by $\FF$ (recall from 
Section~\ref{sec:act:model:correctness} that we consider fair executions).

By satisfying $\BEC(\weak, \FFnnc)$, we proved that temporary operation 
reordering is not possible in \ACTnnc. As we discussed in 
Section~\ref{sec:bayou:anomalies}, it is not the case for \ACTBayou. However,
we can prove, that \ACTBayou satisfies our new correctness criterion 
$\FEC(\weak, \FF)$ (for arbitrary $\FF$).

%

\begin{restatable}{theorem}{bayoustable}
In stable runs \ACTBayou satisfies $\FEC(\weak, \FF) \wedge \LIN(\strong, \FF)$ 
for any arbitrary ACT specification $(\FF, \lvlmap)$.
\end{restatable}

\begin{restatable}{theorem}{bayouasynchronous}
In asynchronous runs \ACTBayou satisfies $\FEC(\weak, \FF)$ and it does not 
satisfy $\LIN(\strong, \FF)$ for any arbitrary ACT specification $(\FF, 
\lvlmap)$.
\end{restatable}

The observation that some undesired anomalies are not inherent to all ACTs
leads to interesting questions that we plan to investigate more closely in the 
future, e.g., what are the common characteristics of the replicated data types 
with mixed-consistency semantics that can be implemented as ACTs that are free 
of temporary operation reordering.

%% file: 6-impossibility.tex

\section{Impossibility} \label{sec:impossibility}

Now we proceed to our central contribution--we show that there exists an ACT 
specification for which it is impossible to propose an ACT implementation that 
avoids temporary operation reordering. 

If a mixed-consistency ACT that implements some replicated data type $\FF$ 
could avoid temporary operation reordering, it would mean that it ensures BEC 
for weak operations and also provides at least some criterion based on $\SO$ 
for strong operations (to ensure a global serialization of all operations).
Hence we state our main theorem:
\begin{theorem} \label{thm:impossibility}
There exists an ACT specification $(\FF, \lvlmap)$, for which there does not 
exist an implementation that satisfies $\SO(\strong) \wedge \BEC(\strong, \FF)$ 
in stable runs, and $\BEC(\weak, \FF)$ in both asynchronous and stable runs.
\end{theorem}

To prove the theorem, we take $\FFseq$ (defined in Figure~\ref{fig:ffseq}) 
as an example replicated data type specification $\FF$.
We consider an ACT specification, which features 
$\app$ and $\rd$ operations in both consistency levels, $\weak$, and $\strong$. 
Thus, $(\FF, \lvlmap) = (\FFseq, \lvlmap_\mathit{seq})$, where 
$\lvlmap_\mathit{seq}(\app) = \lvlmap_\mathit{seq}(\rd) = \{\weak, \strong \}$. 

Let us begin with an observation. Whenever any ACT implementation of $(\FFseq, 
\lvlmap_\mathit{seq})$ that satisfies $\BEC(\weak, \FFseq)$ in asynchronous 
runs, executes a weak $\app$ operation, it has to $\rbcast$ some message $m$.
Since the implementation satisfies $\EV$ (through $\BEC(\weak, 
\FFseq)$) we know that all replicas have to be informed about the invocation 
of $\app$. The replica executing the $\app$ operation may not postpone sending 
the message until some other invocation happens, because all the subsequent 
operation invocations on the replica may be operations, which do not 
grant the replica the right to send messages (e.g., RO operations, by the 
invisible reads requirement). Moreover, the replica may not depend on 
$\tobcast$ messages, because in asynchronous runs they are not guaranteed to be 
delivered to other replicas.\footnote{A replica may $\tobcast$ some messages due 
to the invocation of a weak $\app$ operation, but its correctness cannot depend 
on their delivery.} Thus, a message must be $\rbcast$. 
Since replicas cannot distinguish between asynchronous and stable runs, the 
same observation also holds for stable runs. We utilize this fact in our 
proof by considering asynchronous and stable executions and establishing 
certain invariants which need to hold in both kinds of runs.

We conduct the proof by 
contradiction using a specially constructed execution, in which a replica 
that executes a strong operation has to 
return a value without consulting all replicas. Thus, we consider an ACT 
implementation of $(\FFseq, \lvlmap_\mathit{seq})$ that satisfies $\BEC(\weak, 
\FFseq)$ in asynchronous runs, and $\BEC(\weak, \FFseq) \wedge \SO(\strong) 
\wedge \BEC(\strong, \FFseq)$ 
in both the asynchronous and stable runs.

\begin{proof}
We give a proof for a two-replica system and then show
how it can be easily to a system with $n > 2$ replicas.

We begin with an empty execution represented by a history $H = (E, \op, \rval, 
\rb, \sss, \lvl)$, which we will extend in subsequent steps. Initially replicas 
$R_1$ and $R_2$ are separated by a temporary network partition, which means 
that the messages broadcast by the replicas do not propagate (however, 
eventually they will be delivered once the partition heals). A weak $\app(a)$ 
operation is invoked on $R_1$ in the event $e_a$ and a weak $\app(b)$ operation 
is invoked on $R_2$ in the event $e_b$. By input-driven processing and 
highly available weak operations both replicas return responses for the 
operations and become passive afterwards. 
Let $\msgs_a^\mathit{RB}$ and $\msgs_b^\mathit{RB}$ denote the set of messages
$\rbcast$ by, respectively, $R_1$ and $R_2$, until this point.
Let $\msgs_a^\mathit{TOB}$ and $\msgs_b^\mathit{TOB}$ denote the set of 
messages $\tobcast$ by, respectively, $R_1$ and $R_2$, until this point.
Neither $R_1$ $\rbdeliver$s messages from the set $\msgs_b^\mathit{RB}$, nor 
$R_2$ $\rbdeliver$s messages from the set $\msgs_a^\mathit{RB}$ (due to the 
temporary network partition), but the replicas $\rbdeliver$ their own messages, 
and subsequently become passive (if $\msgs_a^\mathit{TOB} \neq \emptyset$ or 
$\msgs_b^\mathit{TOB} \neq \emptyset$, then these messages remain pending).

Consider an alternative execution represented by history $H' = (E', 
\op', \rval', \rb', \sss', \lvl')$ in which the network partition heals, and 
$R_1$ $\rbdeliver$s all messages in the set $\msgs_b^\mathit{RB}$, $R_2$ 
$\rbdeliver$s all messages in the set $\msgs_a^\mathit{RB}$, and then a weak 
$\rd$ operation is invoked on $R_1$ in the event $e'_c$ and a weak $\rd$ 
operation is invoked on $R_2$ in the event $e'_d$. By invisible reads and 
highly available operations, both replicas remain passive and immediately 
return a response.

\begin{claim} \label{claim:ab}
$\rval'(e'_c) = \rval'(e'_d) = v$, and $v = \mathit{ab}$ or $v = \mathit{ba}$.
\end{claim}

\begin{proof}
We extend $H'$ with infinitely many weak $\rd$ invocations on both $R_1$ 
and $R_2$, in events $e'_k$, for $k \ge 1$. Similarly to $e'_c$ and $e'_d$, the 
$\rd$ operations invoked in each $e'_k$ return immediately and leave the 
replicas $R_1$ and $R_2$ in the unmodified passive state. Since none 
of the $\rd$ operations generate any new messages, $H'$ represents a fair 
infinite execution that satisfies all network properties of an asynchronous 
run. Then, by our base assumption, there exists an abstract execution $A' = 
(H', \vis', \ar', \perc')$, such that $A' \models \BEC(\weak, \FFseq)$.

Because each replica remains in the same state since the execution of $e'_c$ 
and $e'_d$, respectively, each $\rd$ operation invoked in $e'_k$, returns the 
same response as $e'_c$ or $e'_d$, depending on which replica the given event 
was executed. By $\EV(\weak)$, the two updating events $e_a$ and $e_b$ have to 
be both observed by infinitely many of the $e'_k$ events. Let $e'_p$ be one 
such event executed on $R_1$ and $e'_q$ be one such event executed on $R_2$, 
then $(e_a \raaa{\vis'} e'_p \wedge e_a \raaa{\vis'} e'_q \wedge e_b 
\raaa{\vis'} e'_p \wedge e_b \raaa{\vis'} e'_q)$. There is either: $e_a 
\raaa{\ar'} e_b$, or $e_b \raaa{\ar'} e_a$. Now, by the definition of read-only 
operations we can exclude the RO operations from the context of any operation 
without affecting the return value of all operations. Thus $\FFseq(\rd(), 
\context(A', e'_p)) = \FFseq(\rd(), \context(A', e'_q)) = v'$ for some $v'$. 
Because of $\RVAL(\weak, \FFseq)$, $\rval'(e'_p) = v' = \rval'(e'_q)$. 
Therefore, all $\rd$ operations in $H'$ return the same value $v'$, including 
the earliest ones $e'_c$ and $e'_d$, which means that $v = v'$. By the 
definition of $\FFseq$, either $v = \mathit{ab}$ or $v = \mathit{ba}$ 
(depending on whether $e_a \raaa{\ar'} e_b$, or $e_b \raaa{\ar'} e_a$).
\end{proof}

Without loss of generality, let us assume that $v$ obtained in the history $H'$ 
equals $\mathit{ab}$. Let us return to our main history $H$. We extend it 
similarly to the way we extended $H'$, but we do not allow the network 
partition to heal completely. Instead, we just let $\msgs_b^\mathit{RB}$ to 
reach $R_1$, which $\rbdeliver$s them exactly as in $H'$. Since replicas are 
deterministic, the current state of $R_1$ must be the same as it was in $H'$ 
during the execution of $e'_c$. Thus, similarly to $H'$, we invoke a weak $\rd$ 
operation on $R_1$ in an event $e_r$, and $\rval(e_r) = v = \mathit{ab}$.

Consider yet another execution represented by history $H'' = 
(E'', \op'', \rval'', \rb'', \sss'', \lvl'')$ which is obtained from our main 
execution $H$ by removing any steps executed by $R_1$. The events executed on 
$R_2$ remain unchanged, since the two replicas were all the time separated by a 
network partition, and no messages from $R_1$ reached $R_2$. We let the network 
partition heal. $R_1$ $\rbdeliver$s messages from the set 
$\msgs_b^\mathit{RB}$, both replicas $\tobdeliver$ messages from the set 
$\msgs_b^\mathit{TOB}$, and afterward both replicas become passive. 

We now extend $H''$ by infinitely many times applying the following procedure 
(for $k$ from $1$ to infinity):
\begin{enumerate}
\item invoke a \emph{strong} $\rd$ 
on $R_2$ in the event $e''_{2k}$,
\item \label{item:passive} let $R_2$ execute its steps until it becomes passive,
\item on both $R_1$ and $R_2$, $\rbdeliver$ and $\tobdeliver$ all messages, 
respectively, $\rbcast$ or $\tobcast$, by $R_2$ in step~\ref{item:passive},
\item \label{item:passive2} let both replicas reach a passive state,
\item invoke a \emph{weak} $\rd$ 
on $R_1$ in the event $e''_{2k+1}$.
\end{enumerate}

The resulting execution is fair and satisfies all the network properties of a 
stable run. Note that the strong $\rd$ operations executed on $R_2$ are not 
restricted by invisible reads and thus may freely change the state of $R_2$. 
Moreover, they can cause $R_2$ to $\rbcast$ and $\tobcast$ messages. On the 
other hand, the weak $\rd$ operations executed on $R_1$ are always executed on 
a passive state, and leave the replica in the same state. Moreover, $R_1$ does 
not $\rbcast$, nor $\tobcast$ any messages. By non-blocking strong operations 
no strong $\rd$ operation may be pending in $H''$. This is so, because for each 
$k$, by step~\ref{item:passive2}, there is no pending message not yet 
$\tobdeliver$ed on $R_2$, and $R_2$ is in a passive state.

\begin{claim}
There exists an event $e''_x \in E''$, with $x = 2k$ for some natural $k$, such 
that $\rval''(e''_x) = b$.
\end{claim}

\begin{proof}
By our base assumption, there exists an abstract execution $A'' = (H'', \vis'', 
\ar'', \perc'')$, such that $A'' \models \SO(\strong) \wedge \BEC(\strong, 
\FFseq)$. Then, for each $k$, by $\RVAL(\strong, \FFseq)$, $\rval''(e''_{2k}) = 
\FFseq(\rd(), \context(A'', e''_{2k}))$. Moreover, because of $\EV(\strong)$, 
$e_b$ needs to be observed from some point on by every $e''_{2k}$. Thus, we let 
$e_b \raaa{\vis''} e''_x$. Since $e_b$ is the only $\app$ operation visible to 
$e''_x$ (there are no other $\app$ operations in $A''$), by definition of 
$\FFseq$, $\rval''(e''_x) = b$.
\end{proof}

Let us return to our main history $H$. Note that, when we restrict $H$ and 
$H''$ only to events on $R_2$, $H$ constitutes a prefix of $H''$. Moreover, the 
state of $R_2$ at the end of $H$ is the same as in $H''$ just before 
$\tobdeliver$ing messages from the set $\msgs_b^\mathit{TOB}$ (if any) and 
executing the first strong $\rd$ operation. We now extend $H$ by 
$\tobdeliver$ing the messages from the set $\msgs_b^\mathit{TOB}$ and then with 
steps executed on $R_2$ generated using the repeated procedure for $H''$, for 
$k$ from $1$ to $\frac{x}{2}$. We can freely omit the steps executed on $R_1$, 
since none of them influenced in any way $R_2$ ($R_2$ did not 
deliver any message from $R_1$).\footnote{With a typical TOB 
implementation, it might be impossible for $R_2$ to $\tobdeliver$ its own 
messages without the votes of $R_1$ to reach a quorum. However, as we 
have discussed earlier, we abstract away from the implementation details of the 
TOB mechanism. Crucially, no information was transferred from $R_1$ to $R_2$. 
Moreover, in a three replica system, $R_2$ could establish a majority with 
$R_3$ to finalize TOB.}
Thus, there exists an 
event $e_x \in E$ executed on $R_2$, an equivalent of the $e''_x$ event from 
$H''$, such that $\op(e_x) = \rd()$, $\lvl(e_x) = \strong$ and $\rval(e_x) = 
b$. 

Finally, we allow the network partition to heal. $R_2$ $\rbdeliver$s the 
messages from the set $\msgs_a^\mathit{RB}$, and $R_1$ $\rbdeliver$s and 
$\tobdeliver$s any outstanding messages generated by $R_2$ (naturally, $R_1$ 
$\tobdeliver$s messages in the same order as $R_2$ did). Then, we let the 
replicas reach a passive state, and in order to make our constructed execution 
fair, we extend it with infinitely many weak $\rd$ operations as we did with 
$H'$. By our base assumption, there exists an abstract execution $A = (H, \vis, 
\ar, \perc)$, such that $A \models \BEC(\weak, \FFseq) \wedge \SO(\strong) 
\wedge \BEC(\strong, \FFseq)$. There are only two $\app$ operations invoked in 
$A$ in the events $e_a$ and $e_b$. Since $\rval(e_r) = \mathit{ab}$ (which we 
have established after the Claim~\ref{claim:ab}), by $\RVAL(\weak, \FFseq)$ and 
the definition of $\FFseq$, it can be only that $e_a \raar e_b$. We also know 
that $\rval(e_x) = b$ ($e_x$ is a strong $\rd$ operation executed on $R_2$), 
which means that $e_b \ravis e_x \wedge e_a \not\ravis e_x$. By $\SO(\strong)$, 
$e_b \raar e_x \wedge e_a \not\raar e_x$, and thus $e_x \raar e_a$. Therefore, 
a cycle forms in the total order relation $\ar$: $e_a \raar e_b \raar e_x \raar 
e_a$, a contradiction. This concludes our proof for a system with two replicas.

We could easily extended our reasoning to account for any number of 
replicas $n > 2$: any additional replica $R_i$ performs an infinite number of  
read operations, in the same fashion as the replica $R_1$ or $R_2$, depending
on whether $R_i$ originally belonged to the same partition as $R_1$ or $R_2$.
\end{proof}

Since from Theorem~\ref{thm:impossibility} we know that there exists an ACT
specification $(\FF, \lvlmap)$ 
for which we cannot propose (even a specialized) 
implementation that satisfies $\BEC(\weak, \FF)$, we can formulate a more 
general result about generic ACTs:

\begin{corollary}
There does not exist a generic implementation that for arbitrary ACT 
specification $(\FF, \lvlmap)$ satisfies $\SO(\strong) \wedge 
\BEC(\strong, \FF)$ in stable runs, and $\BEC(\weak, \FF)$ 
both in asynchronous, and in stable runs.
\end{corollary}


Theorem~\ref{thm:impossibility} shows that it is impossible to devise a system 
similar to \ACTBayou (for arbitrary $\FF$) 
that never admits temporary 
operation reordering (so satisfies $\BEC(\weak, \FF)$ instead of $\FEC(\weak, 
\FF)$). Hence, admitting temporary operation reordering is the inherent cost of 
mixing eventual and strong consistency when we make no assumptions about the 
semantics of $\FF$. Naturally, for certain replicated data types, such as 
$\FFnnc$, achieving both $\BEC(\weak, \FF)$ and $\LIN(\strong, \FF)$ is 
possible, as we show
with \ACTnnc.

In the next section we discuss several approaches that avoid temporary 
operation reordering, albeit at the cost of compromising fault-tolerance 
(e.g., by requiring all replicas to be operational), or sacrificing 
high availability (e.g., by forcing replicas to synchronize 
on weak operations).

%% file: 7-related_work.tex

\section{Related work} \label{sec:related_work}


\subsection{Symmetric 
models with strong operations blocking upon a single crash}

We start with \emph{symmetric} mixed-consistency models, in which all replicas 
can process both weak and strong operations and communicate directly with each 
other (thus enabling processing of weak operations within network partitions), 
but either do not enable fully-fledged strong operations (there is no 
stabilization of operation execution order) or require all replicas to 
synchronize in order for a strong operation to complete. In turn, the way these 
models bind the execution of weak and strong operations can be understood as an 
asymmetric (1--$n$) variant of quorum-based synchronization. Hence, unlike in 
ACTs, strong operations cannot complete if even a single replica cannot respond 
(due to a machine or network failure), which is a major limitation.

\emph{Lazy Replication} \cite{LLSG92} features three operation levels: 
causal, forced (totally ordered with respect to one another) and immediate 
(totally ordered with respect to all other operations). In this approach,
it is possible that two replicas execute a causal operation
$\opn_c$ and a forced operation $\opn_f$ in different orders. Since $\opn_c$ 
is required to commute with $\opn_f$, replicas will converge to the same 
state. However, the user is never certain that even after the completion of 
$\opn_f$, on some other replica no weaker operation $\opn'_c$ will be executed 
prior to $\opn_f$. Hence the guarantees provided by forced operations are 
inadequate for certain use cases, which require write stabilization, e.g., an 
auction system \cite{PBS18} (see also Section~\ref{sec:introduction}).
On the other hand, immediate operations offer stronger guarantees, but their 
implementation is based on three-phase commit \cite{S81}, and thus requires all 
replicas to vote in order to proceed.

\emph{RedBlue consistency} \cite{LPC+12} extends Lazy Replication 
(with \emph{blue} and \emph{red} operations corresponding to the causal and 
forced ones), by allowing operations to be split into (side-effect free) 
\emph{generator} and (globally commutative) \emph{shadow} operations. This 
greatly increases the number of operations which commute, but red 
operations still do not guarantee write stabilization. To overcome this 
limitation, RedBlue consistency was extended with programmer-defined 
\emph{partial order restrictions} over operations \cite{LPR18}. The proposed 
implementation, \emph{Olisipo}, relies on a 
counter-based system to synchronize conflicting operations. Synchronization can 
be either symmetric (all potentially conflicting pairs of operations must 
synchronize, which means that weak operations are not highly available any 
more) or asymmetric (all replicas must be operational for strong operations to 
complete).

The formal framework of \cite{GYF+16} can be used to express various 
consistency guarantees, including those of Lazy Replication and RedBlue 
consistency, but not as strong as, e.g., linearizability. Conflicts 
resulting from operations that do not commute are modelled through a set of 
tokens. On the other hand, in \emph{explicit consistency} \cite{BDFR+15}, 
stronger consistency guarantees are modelled through application-level 
invariants and can be achieved using multi-level locks (similar to 
readers-writer locks from shared memory).

All above models assume causal consistency (\CAC) as the base-line 
consistency criterion and thus do not account for weaker consistency 
guarantees, such as \FEC or \BEC, as our framework. \CAC is 
argued to be costly to ensure in real-life \cite{BFGH+12}, which makes our 
approach more general.

Finally, the model in \cite{BGY13} is similar to ours but treats strong 
operations as fences (barriers). It means that all replicas must vote in order 
for a strong operation to complete.

Temporary operation reordering is not possible in the models discussed above.
It is because they are either state-based (and thus their formalism
abstracts away from the operation return values which clients observe and 
interpret) and feature no write stabilization, or they require all replicas to 
vote in order to process strong operations.

\subsection{Symmetric Bayou-like 
models}

In Section~\ref{sec:bayou} we have already discussed the relationship between 
the seminal Bayou protocol \cite{TTPD+95} and ACTs. 

In \emph{eventually-serializable data service (ESDS)} \cite{FGL+96}, 
operations are executed speculatively before they are stabilized, similarly to 
Bayou. However, ESDS additionally allows a programmer to attach to an operation 
an arbitrary causal context that must be satisfied before the operation is 
executed. 
Zeno \cite{SFKR+09} is similar to Bayou 
but has been designed to tolerate Byzantine failures. 

All three systems (Bayou, ESDS, Zeno) are eventually consistent, but ensure 
that eventually there exists a single serialization of all operations, and the 
client may wait for a notification that certain operation was stabilized. 
Since these systems enable an execution of arbitrarily complex operations (as 
ACTs), they admit temporary operation reordering. 

Several researchers attempted a formal analysis of the guarantees provided by 
Bayou or systems similar to it. E.g., the authors of Zeno \cite{SFKR+09} 
describe its behaviour using I/O 
automata. In \cite{BB02} the authors analyse Bayou and explain it through a 
formal framework that is tailored to Bayou. Both of these approaches are not as 
general as ours and do not 
enable comparison of the guarantees provided by other systems.
Finally, the framework in \cite{BEH14} enables reasoning about eventually 
consistent systems that enable speculative executions and rollbacks and so also 
\ACTBayou. 
However, the framework does not formalize strong consistency models, which 
means it is not suitable for our purposes.

\subsection{Asymmetric 
models with cloud as a proxy}

Contrary to our approach, the work described below assumes an 
\emph{asymmetric} model in which external clients maintain local copies of 
primary objects that reside in a centralized (replicated) system, referred to 
as \emph{the cloud}. Clients perform weak operations on local copies and 
only synchronize with the cloud lazily or to complete strong operations. Since 
the cloud functions as a communication proxy between the clients, when it is is 
unavailable (e.g., due to failures of majority of replicas or a partition), 
clients cannot observe even each others new weak operations. Hence, this 
approach is less flexible than ours. However, since the cloud serves the role 
of a single \emph{source of truth}, conflicts between concurrent updates can be 
resolved before they are propagated to the clients, so temporary operation 
reordering is not possible.

In \emph{cloud types} \cite{BFLW12}, clients issue operations on replicated 
objects stored in the \emph{local revision} and occasionally synchronize with 
the \emph{main revision} stored in the cloud, in a way similar as in version 
control systems. The synchronization happens either eagerly or lazily, 
depending on the used mode of synchronization. The authors use \emph{revision 
consistency} \cite{BLFS12} as the target correctness criterion. In a subsequent 
work \cite{BLPF15} a \emph{global sequence protocol (GSP)} was introduced, which 
refines the programming model of cloud types, and replaces revision consistency 
with an abstract data model, as revisions and revision consistency were
deemed too complicated for non-expert users.
\emph{Global sequence consistency (GSC)} \cite{GB17} is a consistency model 
that generalizes GSP and a few other approaches that assume external clients 
that either eagerly or lazily push or pull data from the cloud.

\subsection{Asymmetric master-slave 
models}

There are systems which relax strong consistency by allowing clients 
to read stale data, either on demand (the client may forgo recency guarantees 
by choosing a weak consistency level for an operation), or depending on 
the replica localization (in a geo-replicated system the client accessing the 
nearest replica can read stale data that are pertinent to a different region). 
However, in such systems all updating operations (including the weak ones) 
must pass through the primary server designated for each particular data item. 
Thus, similarly to the \emph{asymmetric, cloud as a proxy} models, in this 
approach weak operations are not freely disseminated among the replicas. Since 
all updates (of a concrete data item) are serialized by the primary, temporary 
operation reordering is not possible.

Examples of systems which follow this design and allow users to select an 
appropriate consistency level include PNUTS \cite{CRSS+08}, Pileus 
\cite{TPKB+13}, and also the widely popular contemporary cloud data stores, 
such as AmazonDB \cite{AmazonDynamoDBConsistentReads} and CosmosDB 
\cite{CosmosDBConsistencyLevels}. Systems that guarantee strong consistency 
within a single site and causal consistency between sites include Walter 
\cite{SPAL11}, COPS \cite{LFKA11}, Eiger \cite{LFKA13} and Occult 
\cite{MLCA+17}.

\subsection{Other approaches}

Certain eventually consistent NoSQL data stores enable strongly consistent 
operations on-demand . E.g., Riak allows some data to be kept in \emph{strongly 
consistent buckets} \cite{RiakConsistencyLevels}, which is a namespace 
completely separate from the one used for data accessed in a regular, 
eventually-consistent way.
Apache Cassandra provides compare-and-set-like operations, called 
\emph{light-weight transactions (LWTs)} \cite{CassandraLWTs}, which can be 
executed on any data, but the user is forbidden from executing weakly 
consistent updates on that data at the same time. Concurrent updates and LWTs 
result in undefined behaviour \cite{CASS11000}, which means that 
mixed-consistency semantics of LWTs can be considered broken.

In Lynx \cite{ZPZS+13} and Salt \cite{XSKW+14} mixed-consistency transactions 
are translated into a chain of subtransactions, each committed at a 
different primary site. Thus such transactions can block or raise an error if 
a specific site is unavailable.


Recently some work has been published on the programming language perspective
of mixed-consistency semantics. Since this research is not directly related
to our work, we briefly discuss only a few papers. \emph{Correctables} 
\cite{GPS16} are abstractions similar to futures, that can be used to obtain 
multiple, incremental views on the operation return value (e.g., a result of a 
speculative execution of the operation and then the final return value). 
Correctables are used as an interface for the modified variants of Apache 
Cassandra and ZooKeeper \cite{ZooKeeper} (a strongly consistent system). In 
\emph{MixT} \cite{MM18} each data item is marked with a consistency level that 
will be used upon access. A transaction that accesses data marked with 
different consistency levels is split into multiple independently executed 
subtransactions, each corresponding to a concrete consistency level. The 
compilation-time code-level verification ensures that operations performed on 
data marked with weaker consistency levels do not influence the operations on 
data marked with stronger consistency levels. Understandably, the execution of 
a mixed-level transaction can be blocking. Finally, in \cite{GKNG+20} the 
authors advocate the use of the release-acquire semantics (adapted from 
low-level concurrent programming) and propose \emph{Kite}, a mixed-consistency 
key-value store utilizing this consistency model. In 
Kite weak read 
operations occasionally require inter-replica synchronization and 
block on network communication,
thus they are not highly available.



%% file: 8-conclusions.tex

\section{Conclusions} \label{sec:conclusions}

In this paper we defined \emph{acute cloud types}, a class of replicated 
systems that aim at seamless mixing of eventual and strong consistency. ACTs 
are primarily designed to execute client-submitted operations in a 
highly available, eventually-consistent fashion, similarly to CRDTs. However, 
for tasks that cannot be performed in that way, ACTs at the same time support 
operations that require some form of distributed consensus-based 
synchronization.

We defined ACTs and the guarantees they provide in our novel framework which is 
suited for modeling mixed-consistency systems. We also proposed a new 
consistency criterion called \emph{fluctuating eventual consistency}, which 
captures a common trait of many ACTs, namely \emph{temporary operation 
reordering}. 
Interestingly, temporary operation reordering appears neither in systems 
that are purely eventually consistent (e.g., NoSQL data stores) nor purely 
strongly consistent (e.g., traditional DBMS). Moreover, it is not necessarily 
present in all ACTs, but as we formally prove, it cannot be avoided in ACTs 
that feature arbitrarily complex (but deterministic) semantics (e.g., arbitrary 
SQL transactions).

%% file: a-appendix.tex

\section{}

In this appendix we present additional material that could not be included in 
the article due to space considerations. In Section~\ref{sec:bayou:details} we
give a detailed description of the seminal Bayou protocol, and in 
Section~\ref{sec:bayou:liveness} we discuss its liveness guarantees. Next, in 
Section~\ref{sec:bayou_modified} we supplement the details on how the Bayou 
protocol can be improved to form the general-purpose ACT \ACTBayou. In 
Section~\ref{sec:state_object} we formalize the properties of the $\boxx$ 
object, the \emph{black box} component responsible for the semantics of the 
implemented data type in the algorithms of Bayou and \ACTBayou. Finally, in 
Section~\ref{sec:proofs} we provide the formal proofs of correcntess for 
\ACTnnc and \ACTBayou.

\input{a0-bayou_extended.tex}

\input{a0-bayou_liveness.tex}

\input{a1-bayou_modified.tex}

\input{a1-state_object.tex}
\input{a5-proofs.tex}


%% file: a0-bayou_extended.tex
\subsection{Bayou--detailed description} \label{sec:bayou:details}

{\renewcommand\small{\footnotesize}%

\input{algorithms/bayou_alg.tex}

    \input{algorithms/box_alg.tex}

}

The pseudocode in Algorithm~\ref{alg:bayou} specifies the Bayou protocol for 
replica $R_i$. Replicas are independent and communicate solely by 
message passing. When a client submits an operation $\opn$ to a replica, $\opn$ 
is broadcast within a $\Request$ message using a gossip protocol. In our 
pseudocode, we use regular \emph{reliable broadcast, RB} (line 
\ref{alg:bayou:rbcast}; we say that $\opn$ has been $\rbcast$). Through the 
code in line \ref{alg:bayou:adjustTentativeOrder_early} we simulate immediate 
local $\rbdeliver$y of $\opn$. 

Each Bayou replica totally-orders all operations it knows about (executed 
locally or received through RB). In order to keep track of the total order, a 
replica maintains two lists of operations: $\committed$ and $\tentative$. The 
$\committed$ list encompasses the \emph{stabilized} operations, i.e., 
operations whose final execution order has been established by the primary. On 
the other hand, the $\tentative$ list encompasses operations whose final 
execution order has not yet been determined. The operations on the $\tentative$ 
list are sorted using the operations' timestamps (to resolve any ties, the 
replica identifiers and per replica sequence numbers are used). A timestamp is 
assigned to an operation as soon as a Bayou replica receives it from a client. 

A Bayou replica continually executes operations one by one in the order 
determined by the concatenation of the two lists: $\committed \cdot 
\tentative$ (line~\ref{alg:bayou:execute}). The replica keeps additional 
data structures, such as $\executed$ and $\toBeExecuted$, to keep track of its 
progress. An operation $\opn \in \committed$, once executed, will not be 
executed again as its final operation execution order is determined. 
On the other hand, an operation in the $\tentative$ list might be executed and 
rolled back multiple times. It is because a replica adds operations to the 
$\tentative$ list (rearranging it if necessary; lines 
\ref{alg:bayou:rbdeliver_tentative}-\ref{alg:bayou:rbdeliver_tentative_start})
as they are delivered by a gossip protocol. Hence, a replica might execute 
some operation $\opn$, and then, in order to maintain the proper execution 
order consistent with the modified $\tentative$ list, the replica might be 
forced to roll $\opn$ back (line~\ref{alg:bayou:rollback}), execute a just 
received operation $\opn'$ (which has lower timestamp than $\opn$), and execute 
$\opn$ again. 
We maintain the $\toBeRolledBack$ list of operations scheduled for rollback
(operations are kept in the order reverse to the one in which they were 
executed, line~\ref{alg:bayou:schedule_rollback}). An operation execution can 
proceed only once all the scheduled rollbacks have been performed.

One of the replicas, called the \emph{primary}, periodically \emph{commits} 
operations from its $\tentative$ list by moving them to the end of the 
$\committed$ list, thus establishing their final execution order 
(line~\ref{alg:bayou:periodically}). The primary announces the commit of 
operations by $\rbcast$ing commit messages, so that each replica can also commit 
the appropriate operations. Note that the primary uses the FIFO variant of RB to 
ensure that all replicas commit the same set of operations in the same order. 

Intuitively, the replicas converge to the same state, which is reflected by the 
$\committed \cdot \tentative$ list of operations. More precisely, when the 
stream of operations incoming to the system ceases and there are no network 
partitions (the replicas can communicate with the primary), the $\committed$ 
lists at all replicas will be the same, whereas the $\tentative$ lists will be 
empty. On the other hand, when there are partitions, some operations might not 
be successfully committed by the primary, but will be disseminated within a 
partition using RB. Then all replicas within the same partition will have the 
same $\committed$ and (non-empty) $\tentative$ lists.

Operations are executed on the $\boxx$ object (line~\ref{alg:bayou:box}), which 
encapsulates the state of the local database. At any moment, the value of 
$\boxx$ corresponds to a sequence $s$ of the already executed operations on a 
replica given, where $s$ is a prefix of $\committed \cdot \tentative$. Note
that $\boxx$ allows us to easily rollback a suffix of $s$ 
(line~\ref{alg:bayou:rollback}). We discuss the properties of the $\boxx$ 
object in more detail in Section~\ref{sec:state_object}.


Algorithm~\ref{alg:box} shows a pseudocode of a referential implementation of 
the \StateObject for arbitrary operations of any sequential data type (a 
specialized one can be used to take advantage of specific data type's 
characteristics or to enable non-sequential semantics for certain replicated 
data types which expose concurrency to the client). We assume that each 
operation can be specified as a composition of read and write operations 
on registers (objects) together with some local computation. The assumption is 
sensible, as the operations are executed locally, in a sequential manner, and 
thus no stronger primitives than registers (such as CAS, fetch-and-add, etc.) 
are necessary. The \StateObject keeps an undo log which allows it to revoke the 
effects of any operation executed so far (the log can be truncated to include 
only the operations on the $\tentative$ list).

%% file: algorithms/bayou_alg.tex
\begin{algorithm*}
\caption{The Bayou protocol for replica $R_i$}
\label{alg:bayou}  
  \scriptsize
\vspace{-0.15cm}
\begin{multicols*}{2}
\begin{algorithmic}[1]
\State{\Struct $\Request$($\timestamp$ : int, $\dot$ : pair$\langle$int, int$\rangle$, \par
         \hskip 4.7em $\strongOp$ : boolean, $\opn$ : $\ops(\FF)$)}

\Operator{$<$}{$\rr : \Request$, $\rr' : \Request$}
    \State{\Return $(\rr.\timestamp, \rr.\dot) < (\rr'.\timestamp, \rr'.\dot)$}
\EndOperator

\State{\Var $\boxx$ : \StateObject} \label{alg:bayou:box}
\State{\Var $\currentEventNumber$ : int}
\State{\Var $\committed$, $\tentative$ : list$\langle\Request\rangle$}
\State{\Var $\executed$, $\toBeExecuted$, $\toBeRolledBack$ : list$\langle\Request\rangle$}
\State{\Var $\requestsAwaitingResponse$ : map$\langle\Request, \Response\rangle$}

\Upon{$\invoke$}{$\opn$ : $\ops(\FF)$, $\strongOp$ : boolean} \label{alg:bayou:invoke}
    \State{$\currentEventNumber = \currentEventNumber + 1$} \label{alg:bayou:currInc}
    \State{\Var $r = \Request(\currentTime, (i, \currentEventNumber), \strongOp, \opn)$}
    \State{$\rbcast (\ISSUE, r)$} \label{alg:bayou:rbcast}
    \State{\Call{$\adjustTentativeOrder$}{$r$}} \label{alg:bayou:adjustTentativeOrder_early}
    \State{$\requestsAwaitingResponse.\pput(r, \bot)$}
\EndUpon

\Procedure{$\adjustTentativeOrder$}{$r$ : $\Request$} \label{alg:bayou:adjustTentativeOrder}
    \State{\Var $\previous = [ x | x \in \tentative \wedge x < r ]$} \label{alg:bayou:rbdeliver_tentative_start}
    \State{\Var $\subsequent = [ x | x \in \tentative \wedge r < x ]$}    
    \State{$\tentative = \previous \cdot [ r ] \cdot \subsequent$} \label{alg:bayou:rbdeliver_tentative}
    \State{\Var $\newOrder = \committed \cdot \tentative$}
    \State{\Call{$\adjustExecution$}{$\newOrder$}}
\EndProcedure

\Upon{$\rbdeliver$}{$\ISSUE, r$ : $\Request$} \label{alg:bayou:rbdeliver}
    \If{$r.\dot.\mathit{first} = i$}\LineComment{$r$ issued locally}
        \State{\Return}
    \EndIf
    \If{$\rr \not\in \committed$}
        \State{\Call{$\adjustTentativeOrder$}{$r$}} \label{alg:bayou:adjustTentativeOrder_late}
    \EndIf
\EndUpon

\Upon{$\frbdeliver$}{$\COMMIT, r$ : $\Request$} \label{alg:bayou:tobdeliver}
     \If{$i = \primary$}
        \State{\Return}
     \EndIf
    \State{$\commit(r)$}
\EndUpon

\Procedure{$\commit$}{$r$ : $\Request$}
     \State{$\committed = \committed \cdot [ r ]$} \label{alg:bayou:tobdeliver_committed}
    \State{$\tentative = [ x | x \in \tentative \wedge x \neq r ]$} \label{alg:bayou:tobdeliver_tentative}
     \State{\Var $\newOrder = \committed \cdot \tentative$} \label{alg:bayou:tobdeliver_neworder}
     \State{\Call{$\adjustExecution$}{$\newOrder$}}
     
     \If{$\requestsAwaitingResponse.\ccontains(r) \wedge r \in \executed$}
         \State{return $\requestsAwaitingResponse.\gget(r)$ to client} 
         \State{$\requestsAwaitingResponse.\rremove(r)$}
     \EndIf
\EndProcedure

\Periodically{$\primaryCommit()$}{} \label{alg:bayou:periodically}
  \If{$i = \primary \wedge \tentative \neq []$}
    \State{\Var $[\head] \cdot \tail = \tentative$}
    \State{$\commit(\head)$}
    \State{$\frbcast (\COMMIT, \head)$}
  \EndIf
\EndPeriodically

\Procedure{$\adjustExecution$}{$\newOrder$ : list$\langle\Request\rangle$}
    \State{\Var $\inOrder =\ $longestCommonPrefix$(\executed, \newOrder)$}
    \State{\Var $\outOfOrder = [ x | x \in \executed \wedge x \not\in \inOrder ]$} \label{alg:bayou:outoforder}
    
        \State{$\executed = \inOrder$}
        \State{$\toBeExecuted = [ x | x \in \newOrder \wedge x \not\in \executed ]$} \label{alg:bayou:tobeexecuted}
        \State{$\toBeRolledBack = \toBeRolledBack \cdot \reverse(\outOfOrder)$} \label{alg:bayou:schedule_rollback}
\EndProcedure

\Upon{$\toBeRolledBack \neq []$}{}
    \State{\Var $[ \head ] \cdot \tail = \toBeRolledBack$}
    \State{$\boxx.\rollback(\head)$} \label{alg:bayou:rollback}
    \State{$\toBeRolledBack = \tail$}
\EndUpon
    
\Upon{$\toBeRolledBack = [] \wedge \toBeExecuted \neq []$}{}
    \State{\Var $[ \head ] \cdot \tail = \toBeExecuted$}
    \State{\Var $\response = \boxx.\execute(\head)$} \label{alg:bayou:execute}
    
    \If{$\requestsAwaitingResponse.\ccontains(\head)$}
        \If{$\neg \head.\strongOp \vee \head \in \committed$}
            \State{return $\response$ to client}
             \State{$\requestsAwaitingResponse.\rremove(\head)$}
        \Else
           \State{$\requestsAwaitingResponse.\pput(\head, \response)$}
        \EndIf
    \EndIf
    
    \State{$\executed = \executed \cdot [ \head ]$}
    \State{$\toBeExecuted = \tail$}
\EndUpon

\end{algorithmic}
\end{multicols*}
\vspace{-0.4cm}
\end{algorithm*}

%% file: algorithms/box_alg.tex
\begin{algorithm}[t] 
\caption{\StateObject}
\label{alg:box}  
  \footnotesize
\begin{algorithmic}[1]
\State{\Var $\db$ : map$\langle\Id, \Value\rangle$}
\State{\Var $\undoLog$ : map$\langle\Request$, map$\langle\Id, \Value\rangle\rangle$}

\Function{$\execute$}{$r$ : $\Request$}
    \State{\Var $\undoMap :\ $map$\langle\Id, \Value\rangle$}
    \For{$\instruction$ in $r.\op$}
    \State{evaluate $\instruction$}
    \If{$\instruction = \mathit{read}(\idd : \Id)$}
        \State{$read(\idd : \Id) = \db[\idd]$}
    \EndIf
    \If{$\instruction = \mathit{write}(\idd : \Id, v : \Value)$}
        \If{$\undoMap[\idd] = \bot$}
           \State{$\undoMap[\idd] = \db[\idd]$}
        \EndIf
        \State{$\db[\idd] = v$}
    \EndIf
    \If{$\instruction = \mathit{return}(\response : \Response)$}
    \State{$\undoLog[r] = \undoMap$}
    \State{\Return $\response$}
    \EndIf
    \EndFor
\EndFunction

\Function{$\rollback$}{$r$ : $\Request$}
    \State{\Var $\undoMap = \undoLog[r]$}
    \For{$(k,v) \in \undoMap$}
        \State{$\db[k] = v$}
        \State{$\undoLog = \undoLog \setminus (r, \undoMap)$}
    \EndFor
\EndFunction
\end{algorithmic}
\end{algorithm}

%% file: a0-bayou_liveness.tex
\subsection{Liveness guarantees in Bayou} \label{sec:bayou:liveness}

Eventually consistent systems are aimed at providing high availability. It  
means that a replica is supposed to respond to a request even in the presence 
of network partitions in the system. This requirement can be differently 
formalized. In the model considered by Brewer \cite{B00}, a network partition 
can last infinitely. Then, high availability can be formalized as wait-freedom 
\cite{H91}, which means that each request is eventually processed by the system 
and the response is returned to the client. In the more commonly assumed model 
that admits only temporary network partitions (we also adopt this model, 
similarly to, e.g., \cite{BGY13} \cite{AEM15}), that requirement is not strong 
enough, since a replica could trivially just wait until the partitions are 
repaired before executing a request and responding to the client. Therefore,
in such a model the requirement of high availability must be formulated 
differently. It can be done as follows: a system is highly available if it 
executes each request in a finite number of steps even when no messages are 
exchanged between the replicas (the replica cannot indefinitely postpone 
execution of a request or returning the response to the client, see 
Section~\ref{sec:act:restrictions} for a formal definition). In this sense, 
Bayou is highly available. However, this definition of high availability
does not preclude situations in which, e.g., the number of steps the execution 
of each request takes grows over time and thus is unbounded. Hence, one could 
formulate a slightly stronger requirement, i.e., \emph{bounded} wait-freedom 
\cite{H91}, which states that there is a possibly unknown but bounded number 
of protocol steps that the replica takes before a response is returned to the 
client upon invocation of an operation. Interestingly, unlike many popular 
NoSQL data stores, such as \cite{DHJK+07} or \cite{LM10}, Bayou does not 
guarantee bounded wait-freedom even for weak operations, as we 
now demonstrate.

Consider a Bayou system with $n$ replicas, one of which, $R_s$, processes 
requests slower compared to all other replicas. Assume also that every fixed 
period of time $\Delta t$ there are $n$ new weak requests issued, one on each 
replica, and the processing capabilities of all replicas are saturated. In 
every $\Delta t$, $R_s$ should process all $n$ requests (as do other replicas), 
but it starts to lag behind, with its backlog constantly growing. Intuitively, 
every new operation invoked on $R_s$ will be scheduled for execution after all 
operations in the backlog, as they were issued with lower timestamps. Hence the 
response time will increase with every new invocation on $R_s$. One could try to 
overcome the problem of the increasing latency on $R_s$ by artificially slowing 
the clock on $R_s$, thus giving unfair priority to the operations issued on 
$R_s$, compared to operations issued on other replicas. But then any operation 
invoked on $R_s$ would appear on other replicas as an operation from a distant 
past. In turn, any such operation would cause a growing number of rollbacks on 
the other replicas.

Strong operations cannot be (bounded) wait-free simply because in order for 
them to complete, the primary must be operational, which cannot be guaranteed 
in a fault-prone environment.

Interestingly, in \ACTBayou (see Sections~\ref{sec:bayou:modified} and 
\ref{sec:bayou_modified}) the execution of weak operations is trivially bounded 
wait-free, as they are executed immediately upon their invocations.

%% file: a1-bayou_modified.tex

\subsection{Bayou improved} \label{sec:bayou_modified}

As we discussed in Section~\ref{sec:bayou:modified}, we can improve the
Bayou protocol to make it more fault-tolerant and free of cicular causality,
and thus obtain \ACTBayou. In Algorithm~\ref{alg:bayou_modified} we present the 
modifications to the Algorithm~\ref{alg:bayou}, which give us \ACTBayou. Note 
that, in accordance with the ACT restrictions (see 
Section~\ref{sec:act:restrictions}) we also improve the execution of weak 
read-only (\RO) operations (since any \RO operation $\opn$ does not change the 
logical state of the $\boxx$, $\opn$ can be executed only locally\footnote{We 
assume that \StateObject features an overloaded $\execute$ function which takes 
a plain operation as an argument, instead of a $\Request$ record, when 
executing \RO operations.}).

Firstly, we use TOB in place of the primary to establish the final operation 
execution order. More precisely, every (weak, updating) operation is broadcast 
using RB (as before) as well as TOB (lines 
\ref{alg:bayou_modified:rb}--\ref{alg:bayou_modified:tob}). When a replica 
$\tobdeliver$s an operation $\opn$ (line~\ref{alg:bayou_modified:tobdeliver}), 
it stabilizes $\opn$. Since TOB guarantees that all replicas $\tobdeliver$ the 
same set of messages in the same order, all replicas will stabilize the same 
set of operations in the same order. As we have argued, TOB can be implemented 
in a way that avoids a single point of failure \cite{Lam98}.

Further changes are aimed at eliminating circular causality in Bayou
as well as improving the response time for weak operations. 
To this end (1) any strong operation is broadcast using TOB only 
(line~\ref{alg:bayou_modified:tob2}), and (2) upon being submitted, any weak 
operation is executed immediately on the current state, and then rolled back 
(lines \ref{alg:bayou_modified:execute} and \ref{alg:bayou_modified:rollback}).
It is easy to see that the modification (2) means the incoming stream of 
weak operations from other replicas cannot delay the execution of weak 
operations submitted locally. Below we argue why the two above modifications 
allow us to avoid circular causality in Bayou. 

{\renewcommand\small{\footnotesize}%

\input{algorithms/bayou_alg2_short.tex}

}

The change (1) means that for any pair of a strong $s$ and a weak 
operation $w$, if the return value of any operation $e$ depends 
on both $s$ and $w$ ($e$ \emph{observes} $s$ and $w$), they will be observed 
in an order consistent with the final operation execution order. We prove 
it through the following observations: 
\begin{enumerate}
\item for $e$ to observe $s$, $s$ must be committed (in the modified algorithm 
$s$ never appears on the $\tentative$ list),
\item if $e$ is a strong operation, then $w$ must also be committed, 
because upon execution strong operations do not observe operations on the 
$\tentative$ list; hence both operations are observed according to their
final execution order,
\item otherwise ($e$ is a weak operation):
\begin{enumerate}
\item $w$ is updating (not \RO), because otherwise it would not logically 
impact the return value of $e$,  
\item if $w$ is already committed, it is similar to case 2,
\item if $w$ is not yet committed, $e$ will observe the operations in 
the order $s,w$; on the other hand, once $w$ is delivered by TOB and committed, 
it will appear on the $\committed$ list after $s$, and so $e$ also observes $s$ 
and $w$ in the same order $s,w$.
\end{enumerate}
\end{enumerate}

The change (2) is necessary to prevent circular causality between two (or more) 
weak operations (the case depicted in Figure~\ref{fig:tor}. It is because the 
modified algorithm executes a weak (updating) operation $\opn$ without waiting 
for the $\rbcast$/$\tobcast$ message to arrive. It means that no concurrent 
operation $\op'$ will be executed prior to the first execution of $\opn$, whose 
return value observes the client. Otherwise $\opn$ could observe $\opn'$ even 
though the final execution order is $\opn, \opn'$.

Finally, we redefine the $\Request$ record to include the \emph{execution 
context} $\ctxsf$, i.e., the identifiers of requests already executed upon the
invocation of the current operation and which have influenced the $\boxx$ 
object (those on the $\executed$ list and those on the $\toBeRolledBack$ list). 
Note that in practice such identifiers can be efficiently represented using 
\emph{Dotted Version Vectors} \cite{PBA+10}. With the augmented $\Request$ 
record the implementation of \StateObject can take advantage of the relative 
visibility between operations to achieve the non-sequential semantics of such 
replicated data types as MVRs or ORsets.

%% file: algorithms/bayou_alg2_short.tex
\begin{algorithm}[t]
\caption{Modifications to the Bayou protocol that produce \ACTBayou}
\label{alg:bayou_modified}  
  \scriptsize
\begin{algorithmic}[1]
\item[] \emph{// redefines the $\Request$ record}
\State{\Struct $\Request$($\timestamp$ : int, $\dot$ : pair$\langle$int, int$\rangle$, $\ctxsf$ : set$\langle$pair$\langle$int, int$\rangle\rangle$ \par
         \hskip 4.7em $\strongOp$ : boolean, $\opn$ : $\ops(\FF)$)}
         
\item[] \emph{// replaces $\invoke(...)$ and $\primaryCommit()$}
\Upon{$\invoke$}{$\opn$ : $\ops(\FF)$, $\strongOp$ : boolean} 
    \If{$\neg\strongOp \wedge \opn \in \readonlyops(\FF)$}
        \State{\Var $\response = \boxx.\execute(\opn)$} \label{alg:bayou_modified:ro_execute}
        \State{return $\response$ to client}
    \Else
        \State{$\currentEventNumber = \currentEventNumber + 1$}
        \State{\Var $\ctxsf = (\bigcup_{r \in \executed} r.\dot) \cup (\bigcup_{r \in \toBeRolledBack} r.\dot)$}
        \State{\Var $r = \Request(\currentTime, (i, \currentEventNumber), \ctxsf, \strongOp, \opn)$} \label{alg:bayou_modified:new_request}
        \If{$\neg\strongOp$}
            \State{\Var $\response = \boxx.\execute(r)$} \label{alg:bayou_modified:execute}
            \State{$\boxx.\rollback(r)$} \label{alg:bayou_modified:rollback}
            \State{$\rbcast (\ISSUE, r)$} \label{alg:bayou_modified:rb}
            \State{$\tobcast (\COMMIT, r)$} \label{alg:bayou_modified:tob}
            \State{\Call{$\adjustTentativeOrder$}{$r$}}
            \State{return $\response$ to client} \label{alg:bayou_modified:return_to_client}
        \Else
            \State{$\requestsAwaitingResponse.\pput(r, \bot)$}
            \State{$\tobcast (\COMMIT, r)$} \label{alg:bayou_modified:tob2}
        \EndIf
    \EndIf
\EndUpon

\item[] \emph{// replaces} upon $\frbdeliver(\COMMIT, r : \Request)$
\Upon{$\tobdeliver$}{$\COMMIT, r$ : $\Request$} \label{alg:bayou_modified:tobdeliver} 
    \State{$\commit(r)$} \label{alg:bayou_modified:commit}
\EndUpon
\end{algorithmic}
\end{algorithm}

%% file: a1-state_object.tex
\subsection{\StateObject properties} \label{sec:state_object}

Although in Algorithm~\ref{alg:box} we present a referential implementation of 
\StateObject, in general we treat the $\boxx$ object as a black box with 
unknown implementation. The corretness of \ACTBayou depends on the properties 
of the $\boxx$ object which we formalize below.

Take the list of requests that were executed on the $\boxx$, and remove the 
requests which were rolled back; we call the resulting sequence $\alpha$ the 
current trace of the $\boxx$.\footnote{We omit weak RO operations executed in 
Algorithm~\ref{alg:bayou_modified} line~\ref{alg:bayou_modified:ro_execute}, 
which are not associated with any $\Request$ record.} Since the $\boxx$ 
encapsulates the state of the system after locally executing and revoking 
requests, we require that the $\boxx$'s responses are consistent with a 
deterministic serial execution of $\alpha$ as specified by the type 
specification $\FF$ when taking into account the relative visibility between 
requests encoded in the $\ctxsf$ field of the $\Request$ record. In case of 
any strong operation $\opn$ (in a request $r$), we assume that all requests 
$r' \in \alpha$ prior to $r$ are visible to $r$ (regardless of $\ctxsf$). 
This is because $\opn$ is executed only once $r$ is on the $\committed$ list 
and thus its position relative to all other operations is fixed and corresponds
to the TOB order.

More precisely, for any given trace $\alpha$, the $\boxx$ object 
deterministically holds the state $S_\alpha$, and for any operation 
$\opn \in \ops(\FF)$, the response of the $\boxx.\execute$ function invoked on 
the $\boxx$ object in state $S_\alpha$ equals $\FF(\opn, C_\alpha)$, where 
$C_\alpha = (E_\alpha, \op_\alpha, \vis_\alpha, \ar_\alpha)$ is a context such 
that:
\begin{itemize}
\item $E_\alpha$ consists of all the requests in $\alpha$,
\item $\op_\alpha(r) = r.\opn$, for any request $r \in E_\alpha$,
\item $\vis_\alpha$ is the visibility relation based on the $\ctxsf$ fields of 
the $\Request$ record for the weak operations and on the order in $\alpha$ for 
strong operations, i.e. for any $r,r' \in E_\alpha$ such that $r 
\raaa{\vis_\alpha} r'$:
\begin{itemize}
\item if $r'.\strongOp = \false$, then $r.\dot \in r'.\ctxsf$; 
\item if $r'.\strongOp = \true$, then $r \raaa{\ar_\alpha} r'$;
\end{itemize}
\item $\ar_\alpha$ is the enumeration of requests in $E_\alpha$ according to 
their position in $\alpha$.
\end{itemize}

In \ACTBayou, $\alpha = \executed \cdot \reverse(\toBeRolledBack)$, because:
\begin{itemize}
\item requests are executed only if $\toBeRolledBack$ is empty,\footnote{Weak 
requests are also executed in the invoke block, independently of the 
$\toBeExecuted$ and $\toBeRolledBack$ lists, but they are immediately 
afterwards rolled back, so they do not influence the trace.}
\item whenever a request is executed it is added to the $\executed$ list,
thus it is appended to the end of $\alpha$,
\item in the $\adjustExecution$ function, some requests 
move from the 
$\executed$ list to the end of the $\toBeRolledBack$ list, thus not changing 
their position in $\alpha$,
\item whenever a request is rolled back, it is removed from the head of the 
$\toBeRolledBack$ list, and thus removed from the end of $\alpha$, consistently 
with the definition of a trace.
\end{itemize}

%% file: a5-proofs.tex
\subsection{Proofs of correctness} \label{sec:proofs}



In this section we provide the formal proofs of correcntess for \ACTnnc and 
\ACTBayou anticipated in Section~\ref{sec:guarantees:bayou}. We start with an 
overview of proofs' structures. 

In order to prove correctness of either protocol, we take a single arbitrary 
execution of the protocol, and without making any specific assumptions about 
it, we show how the visibility and arbitration relations can be defined so that 
the appropriate correctness guarantees can be proven. Below we briefly outline 
our approach.

In both \ACTnnc and \ACTBayou, strong operations are disseminated solely by 
TOB, and weak updating operations are sent using both RB and TOB. On the other 
hand weak RO operations are executed completely locally and do not involve any 
network communication (strong RO operations are present only in \ACTBayou and 
are treated as regular strong operations). Thus, in the proofs, for the purpose 
of constructing the arbitration relation ($\ar$), we order all updating (strong 
or weak) operations based on the order of the delivery of their respective 
messages broadcast using TOB. In the case of updating operations whose messages 
were not $\tobdeliver$ed (which can happen in the asynchronous runs), we order 
them in $\ar$ after all the operations whose messages were $\tobdeliver$ed. 
Their relative order can be arbitrary in \ACTnnc, and in \ACTBayou it has to 
conform to the order imposed by the $\Request$ records. Finally, 
for completeness, $\ar$ needs to include also weak RO operations. We carefully 
interleave them with updating operations in such a way to guarantee no circular 
causality as well as equivalence between visibility and arbitration for strong 
operations.

We construct the visibility relation ($\vis$) by choosing for any two events 
$e, e'$ whether one should be observed by the other. We include an edge $e 
\ravis e'$ under two, broad conditions: the edge is \emph{essential}, i.e., $e$ 
could have influenced the return value of $e'$, or the edge is 
\emph{non-essential}, i.e., $e$ could not have influenced the return value of 
$e'$ (because, e.g., $e$ is an RO operation), but $e$ occurs before $e'$ in 
real-time or arbitration. Non-essential edges are important to guarantee 
eventual visibility for all events.

Now let us make some observations regarding network properties during 
synchronous and asynchronous runs. Since we consider infinite fair executions, 
in both types of runs each message $\rbcast$ is guaranteed to be $\rbdeliver$ed 
by each replica. On the other hand, the same delivery guarantee, but for 
messages $\tobcast$, holds only in the stable runs, and in the asynchronous 
runs, some messages can be $\tobdeliver$ed while others may remain pending. 
However, asynchronous runs still obey other guarantees, which means that, 
crucially, no messages $\tobcast$ will be $\tobdeliver$ed by any replica out of 
order. Moreover, if some message was $\tobdeliver$ed by one replica, then it 
will be $\tobdeliver$ed by all replicas. Also, if one replica manages to 
$\tobcast$ infinitely many messages which are then $\tobdeliver$ed, then each 
replica can succesfully $\tobcast$ and $\tobdeliver$ its messages. Thus, in the 
asynchronous runs, we expect a finite number of $\tobcast$ messages to be 
$\tobdeliver$ed, while all other to remain pending.

For each event $e$ let us denote by $\msg_\TOB(e)$ and $\msg_\RB(e)$, 
respectively, the message $\tobcast$ in the event $e$ and the message $\rbcast$ 
in the event $e$ (both $\msg_\TOB(e)$ and $\msg_\RB(e)$ can be undefined for a 
given event $e$, denoted $\msg_\TOB(e) = \bot$ or $\msg_\RB(e) = \bot$). For 
any two events $e, e'$, such that $\msg_\TOB(e) = m$, $\msg_\TOB(e') = m'$ and 
$\tobNo(m) < \tobNo(m')$ we introduce the following notation: $e \raaa{\tobNo} 
e'$, which defines the \emph{$\tobNo$ order} (based on the $\tobNo$ function). 
Additionally, for any two events $e, e'$, such that $\msg_\TOB(e) = m$ (or 
respectively $\msg_\RB(e) = m$), we write $e \raaa{\TOBdel} e'$ ($e 
\raaa{\RBdel} e'$), if $e'$ executes on a replica that has $\tobdeliver$ed 
($\rbdeliver$ed) $m$ prior to its execution.

Finally, let us observe that we model replicas as deterministic state machines 
(as discussed in Section~\ref{sec:act:model}), whose specification we give 
through pseudocode. The variables declared in the algorithms of \ACTnnc and 
\ACTBayou represent the state of the replicas, while the code blocks represent 
atomic steps that transition the replicas from one state to another.
It means that each such block executes completely before any of its effects 
become visible. This allows us to infere the following rule (in both \ACTnnc 
and \ACTBayou) for weak operations which execute in one atomic transition in 
some event $e$, which is either in the $\TOBdel$ or $\RBdel$ relation with any 
other event $e'$: $\lvl(e) = \weak \wedge (e \raaa{\TOBdel} e' \vee e 
\raaa{\RBdel} e') \Rightarrow e \rarb e'$ ($e$ \emph{returns before} $e'$).

\input{a5a-annc_proofs.tex}

\input{a5b-bayou_proofs.tex}

%% file: a5a-annc_proofs.tex
\subsubsection{ANNC correcntess proofs}

Let us proceed with the proof of the guarantees offered by \ACTnnc in the 
stable runs.

\actnncstable*

\begin{proof}
For any given arbitrary stable run of \ACTnnc represented by a history $H = 
(E, \op, \rval, \rb, \sss, \lvl)$ we have to find suitable $\vis$, $\ar$ and 
$\perc$, such that $A = (H, \vis, \ar, \perc)$ is such that $A \models 
\BEC(\weak, \FFnnc) \wedge \LIN(\strong, \FFnnc)$.

\noindent \textbf{Additional observations.}
Note that each $\subtractit$ operation executed in some event $e$ finishes when 
the replica $\tobdeliver$s the message $m = \msg_\TOB(e)$. It means that for 
every operation executed in event $e'$, such that $e \rarb e'$, if 
$\msg_\TOB(e') = m'$ ($m' \neq \bot$), then $\tobNo(m) < \tobNo(m')$.

\noindent \textbf{Arbitration.}
We construct the total order relation $\ar$ by sorting all updating events 
(additions and subtractions) based on the order in which their respective 
$\tobcast$ messages are $\tobdeliver$ed, i.e., respecting the $\tobNo$ order.

Next, we interleave the updating events with RO events (gets) in the following 
way: each such an RO event $e$ occurs in $\ar$ after the last subtract event 
$e'$ such that $e \not\rarb e'$. Thus, for each subtract event $e'$ the 
following holds $e \raar e' \Rightarrow e \rarb e'$. The relative order of RO 
operations is irrelevant.

As \ACTnnc does not feature operation reordering, for each event $e$ we simply 
let $\perc(e) = \ar$. 

\noindent \textbf{Visibility.}
For any two events $e, e' \in E$, we include an edge $e \ravis e'$ in our 
construction of $\vis$, if:
\begin{enumerate}
\item \label{enum:actvis:tob}
$\op(e) = \addit(v)$ or $\op(e) = \subtractit(v)$, $\op(e') = \subtractit(v')$ 
and $e \raaa{\tobNo} e'$,
\item \label{enum:actvis:subget}
$\op(e) = \subtractit(v)$, $\op(e') = \getit$ and $e \raaa{\TOBdel} e'$,
\item \label{enum:actvis:addgettob}
$\op(e) = \addit(v)$, $\op(e') = \getit$, and $e \raaa{\TOBdel} e'$,
\item \label{enum:actvis:addgetrb}
$\op(e) = \addit(v)$, $\op(e') = \getit$, and $e \raaa{\RBdel} e'$,
\item \label{enum:actvis:getget}
$\op(e) = \getit$, $\op(e') = \getit$ and $e \rarb e'$,
\item \label{enum:actvis:getsub}
$\op(e) = \getit$, $\op(e') = \subtractit(v')$ and $e \raar e'$,
\item \label{enum:actvis:add}
$\op(e') = \addit(v')$ and $e \rarb e'$,
\end{enumerate}
(for some $v,v' \in \mathbb{N}$).

The edges~\ref{enum:actvis:tob}-\ref{enum:actvis:addgetrb} are essential, while
the edges~\ref{enum:actvis:getget}-\ref{enum:actvis:add} are non-essential.
The updates that are visible to a $\subtractit$ operation depends solely on the 
$\tobNo$ order, while in case of a $\getit$ operation, the $\TOBdel$ and 
$\RBdel$ relations play a role. It does not matter which updates are visible to 
an $\addit$ operation because it always responds with a simple $\ok$ 
acknowledgment, hence the edge~\ref{enum:actvis:add} is non-essential.

Note that in case of 
edges~\ref{enum:actvis:addgettob}-\ref{enum:actvis:addgetrb}, $e \rarb e'$ is 
implied (see the general observations in Section~\ref{sec:proofs}), and in case 
of the edge~\ref{enum:actvis:getsub}, $e \rarb e'$ follows directly from the 
construction of $\ar$. Thus, for all 
edges~\ref{enum:actvis:addgettob}-\ref{enum:actvis:add}, $e \rarb e'$.

Having defined $A$ (through $\vis$, $\ar$ and $\perc$), it now remains to show
that $A \models \BEC(\weak, \FFnnc) \wedge \LIN(\strong, \FFnnc)$, or more 
specifically $A \models \EV(\weak) \wedge \EV(\strong) \wedge \NCC(\weak) 
\wedge \NCC(\strong) \wedge \RVAL(\weak) \wedge \RVAL(\strong) \wedge 
\SO(\strong) \wedge \RT(\strong)$.

\noindent \textbf{Eventual visibility.}
We prove now that eventual visibility is satisfied for all events:
\begin{itemize}
\item each $\addit$ or $\subtractit$ event $e$ is visible to all subsequent 
$\subtractit$ events from some point, because there is only a finite number of 
updating events $e'$ such that $e \not\raaa{\tobNo} e'$ 
(\ref{enum:actvis:tob}),
\item each $\addit$ or $\subtractit$ event $e$ is visible to all subsequent 
$\getit$ events from some point, because both $\msg_\RB(e)$ and $\msg_\TOB(e)$ 
are eventually delivered on all replicas, (\ref{enum:actvis:subget}, 
\ref{enum:actvis:addgettob} and \ref{enum:actvis:addgetrb}),
\item each $\getit$ event $e$ is visible to all subsequent $\getit$ events from 
some point (\ref{enum:actvis:getget}),
\item each $\getit$ event $e$ is visible to all subsequent $\subtractit$ events
from some point, because by construction of $\ar$ there is only a finite number 
of events $e'$ such that $e \not\raar e'$ (\ref{enum:actvis:getsub}),
\item each event is visible to all subsequent $\addit$ events from some point
(\ref{enum:actvis:add}).
\end{itemize}

\noindent \textbf{No circular causality.}
We need to show that $\acyclic(\hb \cap (W \times W))$ and $\acyclic(\hb \cap 
(S \times S))$, where $W \subseteq E, S \subseteq E$, are, respectively, the 
sets of all weak, and strong events. We elect to prove a more general case of 
$\acyclic(\hb)$.

Recall that $\hb = (\vis \cup \so)^+$. If $\acyclic(\vis \cup \so)$, then 
$\acyclic(\hb)$, because transitive edges cannot introduce cycles. Thus, we 
have eight types of edges to consider: 
edges~\ref{enum:actvis:tob}-\ref{enum:actvis:add} from $\vis$ and the eight
edge $e \raso e'$. We divide them into two groups: the first one consists of
edges~\ref{enum:actvis:tob}-\ref{enum:actvis:subget}, while the second one
consists of edges~\ref{enum:actvis:addgettob}-8. Note that for the second group
$e \rarb e'$ always holds.

There can be no cycles when we restrict the edges only to the ones from the 
first group, as edge~\ref{enum:actvis:tob} is constrained by the $\tobNo$ 
order, and edge~\ref{enum:actvis:subget} leads to a $\getit$ event which cannot 
be followed using only edges from the first group.

Also, there can be no cycles when we restrict the edges only to the ones from 
the second group, as all the edges are constrained by the $\rb$ relation, which 
is naturally acyclic.

Thus, a potential cycle could only form when we mix edges from both groups.
Let us assume that the cycle contains the following chain of edges: $a 
\raaa{\hb} b \raaa{\hb} ... \raaa{\hb} c \raaa{\hb} ... \raaa{\hb} d$, where 
$a,b,c,d \in E$, all the edges between $b$ and $c$ belong to the second group, 
while the other ones belong to the first group. Notice that $b \rarb c$, and
that $\op(a), \op(c) \in \{\addit(v) : v \in \mathbb{N}\} \cup \{\subtractit(v) 
: v \in \mathbb{N}\}$ while $\op(b), \op(d) \in \{\subtractit(v) : v \in 
\mathbb{N}\} \cup \{\getit\}$. Thus, the chain consists of a series of edges 
from the first group and a series of edges from the second group. The whole 
cycle can be combined from multiple such chains, but for simplicity, let us 
assume that it contains only one such chain and that $d = a$ (the same 
reasoning as below can be applied iteratively for multiple interleavings of 
edges from the two groups).

If $\op(b) = \subtractit(v)$, for some $v \in \mathbb{N}$, then $a 
\raaa{\tobNo} b$ (edge~\ref{enum:actvis:tob}), and since $b \rarb c$, also $b 
\raaa{\tobNo} c$ (see the additional observations in the begining of the 
proof). A contradiction: $a \raaa{\tobNo} b \raaa{\tobNo} c \raaa{\tobNo} a$.

If $\op(b) = \getit$, then $\op(a) = \subtractit(v)$, for some $v \in 
\mathbb{N}$, and $a \raaa{\TOBdel} b$ (edge~\ref{enum:actvis:subget}). Either 
$a \raaa{\tobNo} c$, or $c \raaa{\tobNo} a$. In the former case we end up with 
a similar contradiction as above: $a \raaa{\tobNo} c \raaa{\tobNo} a$. In the 
latter case, since $c \raaa{\tobNo} a$, also $c \raaa{\TOBdel} b$ (the message 
$\msg_\TOB(c)$ is $\tobdeliver$ed before the message $\msg_\TOB(a)$). However, 
$b \rarb c$, which means that the message $\msg_\TOB(c)$ was not even 
$\tobcast$ yet when $b$ executed. A contradiction.

\noindent \textbf{Return value consistency.}
We need to show that for each event $e \in E$: $\rval(e) = \FFnnc(\op(e), 
\context(A, e))$. We base our reasoning below on essential $\vis$ edges and 
$\ar$ order.

Trivially, the condition is satisfied for all $\addit$ events, which always 
return $\ok$. For all $\subtractit$ and $\getit$ events, we can exclude from
$\context(A, e)$ all $\getit$ events which by the definition of an RO 
operation are irrelevant for the computation of $\FFnnc$.

In case of a $\subtractit(v)$ operation, for some $v \in \mathbb{N}$, executed 
in some event $e$, $\context(A, e)$ includes all the $\addit$ and $\subtractit$ 
events that precede $e$ in the $\tobNo$ order. When applying the $\foldr$ 
function from the definition of $\FFnnc$, these $\addit$ and $\subtractit$ 
operations are processed one by one, in the order of their $\tobdeliver$y (by 
construction of $\ar$). Each $\addit(v)$ operation increases the accumulator by 
$v$, and each $\subtractit(v)$ operation decreases the accumulator by $v$, but 
only if it is greater or equal $v$. This matches the pseudocode 
(lines~\ref{alg:actpc:tobinc} and 
\ref{alg:actpc:tobdecstart}-\ref{alg:actpc:tobdecend}) with the accumulator 
corresponding to the difference between $\si$ and $\sd$ variables. Thus, the 
computed value of the $\foldr$ function corresponds to the difference between 
$\si$ and $\sd$ variables at the time the response to $e$ is computed in 
line~\ref{alg:actpc:subcalcstart}. If that value is greater or equal $v$ then 
$\true$ is returned, which matches the pseudocode's behaviour.

In case of a $\getit$ operation executed in some event $e$, $\context(A, e)$
includes all the $\addit$ and $\subtractit$ events that were $\tobdeliver$ed 
before the execution of $e$, as well as, (possibly) some $\addit$ events which 
were not $\tobdeliver$ed, but only $\rbdeliver$ed before the execution of $e$. 
Note that all the latter $\addit$ events are ordered according to $\ar$, after 
all the former $\addit$ and $\subtractit$ events (have they had been ordered 
earlier due to lower $\tobNo$ value of their respective $\tobcast$ message, 
they would also be $\tobdeliver$ed). When processing the $\foldr$ function up 
to the last $\tobdeliver$ed event, the value of the accumulator corresponds, 
similarly as in case of $\subtractit$ events above, to the difference between 
$\si$ and $\sd$ variables. Then, when processing the remaining $\addit$ events 
the final computed value of the $\foldr$ function grows by an amount $V$, which 
is equal to the sum of all these $\addit$ operations' arguments. Due to the 
fact that each $\tobdeliver$ed message is first $\rbdeliver$ed or is processed 
as if it were $\rbdeliver$ed
(lines~\ref{alg:actpc:tobrbinvstart}-\ref{alg:actpc:tobrbinvend}), the value of 
$\wi$ is always greater or equal $\si$. The difference between $\wi$ and $\si$ 
variables corresponds exactly to $V$, because it includes events which were 
$\rbdeliver$ed, but not $\tobdeliver$ed. Thus, the computed value of 
$\FFnnc(\getit, \context(A, e))$ equals $\si - \sd + V = \si - \sd + \wi - \si 
= \wi - \sd$ at the time of executing $e$, which matches $\rval(e)$.

\noindent \textbf{Single order.}
Since there are no pending $\subtractit$ operations (because eventually every 
message is $\tobdeliver$ed and the operations finish), we have to simply prove 
that $\vis \cap (E \times S) = \ar \cap (E \times S)$, where $S = \{e : \lvl(e) 
= \strong\}$. In other words, for any two events $e \in E, e' \in S$: $e \ravis 
e' \iff e \raar e'$.

Let us begin with $e \ravis e' \Rightarrow e \raar e'$. Either $e 
\raaa{\tobNo} e'$ (edge~\ref{enum:actvis:tob}), or $\op(e) = \getit$ 
(edge~\ref{enum:actvis:getsub}). In both cases $e \raar e'$.

Now let us consider $e \raar e' \Rightarrow e \ravis e'$. Either $\op(e) \in 
\{\addit(v) : v \in \mathbb{N}\} \cup \{\subtractit(v) : v \in \mathbb{N}\}$, 
or $\op(e) = \getit$. In the former case, $e \raaa{\tobNo} e'$, and thus $e 
\ravis e'$ (edge~\ref{enum:actvis:tob}). In the latter case, $e \rarb e'$ (by 
construction of $\ar$), and thus $e \ravis e'$ (edge~\ref{enum:actvis:getsub}).

\noindent \textbf{Real-time order.}
We need to show that arbitration order respects the real-time order of strong 
operations, i.e., $\rb \cap (S \times S) \subseteq \ar$. In other words, for 
any two $e,e' \in S$: $e \rarb e' \Rightarrow e \raar e'$.

Clearly, if $e \rarb e'$, then $e \raaa{\tobNo} e'$ (see the additional 
observations in the begining of the proof). Thus, $e \raar e'$ (by construction 
of $\ar$).
\end{proof}

Now, let us continue with the proof of the guarantees offered by \ACTnnc in the 
asynchronous runs.

\actnncasynchronous*

\begin{proof}
To show the inability of \ACTnnc to satisfy $\LIN(\strong, \FFnnc)$ in 
asynchronous runs, it is sufficient to observe that due to some of the 
$\tobcast$ messages not being $\tobdeliver$ed, some of the $\subtractit$ 
operations remain pending. A pending operation's return value equals $\nabla$ 
which is unreconcilable with the requirements of the predicate $\RVAL(\FFnnc)$.

The proof regarding the guarantees of the weak operations is similar to the one 
for the stable runs, thus we rely on it and focus only on differences between 
stable and asynchronous runs that need to be addressed. Now for any given
arbitrary asynchronous run of \ACTnnc represented by a history $H = (E, \op, 
\rval, \rb, \sss, \lvl)$ we have to find suitable $\vis$, $\ar$ and $\perc$, 
such that $A = (H, \vis, \ar, \perc)$ is such that $A \models \BEC(\weak, 
\FFnnc)$.

\noindent \textbf{Arbitration.}
We construct the total order relation $\ar$ by sorting all updating events 
(additions and subtractions) based on the order in which their respective 
$\tobcast$ messages are $\tobdeliver$ed, i.e., respecting the $\tobNo$ order. 
Updating events whose messages are not $\tobdeliver$ed are ordered after those 
whose messages are $\tobdeliver$ed.

Next, we interleave the updating events with RO events (gets) in the following 
way: each such an RO event $e$ occurs in $\ar$ after the last 
\emph{non-pending} subtract event $e'$ such that $e \not\rarb e'$. Thus, for 
each non-pending subtract event $e'$ the following holds $e \raar e' 
\Rightarrow e \rarb e'$. The relative order of RO operations is irrelevant.

As \ACTnnc does not feature operation reordering, for each event $e$ we simply 
let $\perc(e) = \ar$. 

\noindent \textbf{Visibility.}
We construct the visibility relation in the same way as in the stable 
runs case. However, we remove edges to and from pending $\subtractit$ events. 
Since pending operations do not provide a return value, no edge to a pending 
event is essential. Also, as we guarantee only eventual visibility for weak 
events, edges to $\subtractit$ events are not necessary to satisfy 
$\EV(\weak)$. Moreover, edges \emph{from} pending events are not needed either, 
because by definition a pending event is never followed in $\rb$ by any other 
event (which is a requirement to fail the test for $\EV$). Again, for all 
edges~\ref{enum:actvis:addgettob}-\ref{enum:actvis:add}, $e \rarb e'$.

Having defined $A$ (through $\vis$, $\ar$ and $\perc$), it now remains to show
that $A \models \BEC(\weak, \FFnnc)$, or more specifically $A \models 
\EV(\weak) \wedge \NCC(\weak) \wedge \RVAL(\weak)$.

\noindent \textbf{Eventual visibility.}
We prove now that eventual visibility is satisfied for all \emph{weak} events:
\begin{itemize}
\item each $\addit$ or \emph{non-pending} $\subtractit$ event $e$ is visible to 
all subsequent $\getit$ events from some point, because $\msg_\RB(e)$ or 
$\msg_\TOB(e)$ are eventually delivered on all replicas 
(\ref{enum:actvis:subget}, \ref{enum:actvis:addgettob} and 
\ref{enum:actvis:addgetrb}),
\item each $\getit$ event $e$ is visible to all subsequent $\getit$ events from 
some point (\ref{enum:actvis:getget}),
\item each \emph{non-pending} event is visible to all subsequent $\addit$ 
events from some point (\ref{enum:actvis:add}).
\end{itemize}

\noindent \textbf{No circular causality.}
We use exactly the same reasoning as in the stable runs case to show that 
$\acyclic(\hb)$ holds true.

\noindent \textbf{Return value consistency.}
Again, we use exactly the same reasoning as in the stable runs case to show
that for each \emph{weak} event $e \in E$: $\rval(e) = \FFnnc(\op(e), 
\context(A, e))$. Although this time we only need to prove return value 
consistency for $\addit$ and $\getit$ operations, it can be shown that it also 
holds for \emph{non-pending} subtract events.

\end{proof}

%% file: a5b-bayou_proofs.tex
\subsubsection{\ACTBayou correcntess proofs}

The proofs for \ACTBayou are analogous to those for \ACTnnc, but are slighly 
more complex due to operation reordering and the more general nature of 
\ACTBayou with unconstrained operations' semantics (in contrast \ACTnnc 
features weak updating operations that always return $\ok$). Because we strive 
in this section for self-contained proofs we do not refer to the proofs for 
\ACTnnc even when doing so would allow us to omit some repetitions.

We begin with the proof of guarantees offered by \ACTBayou in the stable runs.

\bayoustable*

\begin{proof}
For any given arbitrary stable run of \ACTBayou represented by a history $H = 
(E, \op, \rval, \rb, \sss, \lvl)$ we have to find suitable $\vis$, $\ar$ and 
$\perc$, such that $A = (H, \vis, \ar, \perc)$ is such that $A \models 
\FEC(\weak, \FF) \wedge \LIN(\strong, \FF)$.

\noindent \textbf{Additional observations.}
All events besides weak RO ones, have an associated unique $\Request$ record 
which is disseminated using $\rbcast$ and $\tobcast$; let us denote by 
$\req(e)$ the $\Request$ record of the event $e$.\footnote{Thus a trace of the 
$\boxx$ object, which consists of such records, can be translated into a 
sequence of events.} Since, the handling of weak RO events, which are 
\emph{local} to a replica, differ significantly from other events, which are 
\emph{shared}, we divide the set of all events $E$ into two subsets: $\Psi 
\subseteq E$, consisting of weak updating, strong updating and strong RO 
events; and $\Omega \subseteq E$, consisting of weak RO events. We also further
divide $\Psi$ into subsets $\Psi_w$ and $\Psi_s$, consisting of, respectively, 
weak and strong events.


Upon $\tobdeliver$y of a $\COMMIT$ message, the receiced request $r$ is 
\emph{committed}
(Algorithm~\ref{alg:bayou_modified} line~\ref{alg:bayou_modified:commit}), i.e.,
it is appended at the end of the $\committed$ list, and removed from the 
$\tentative$ list (if present there). Note that the position of $r$ established 
on the $\committed$ list never changes as the list is only appended to. Once 
the request is committed, the operation associated with the request is 
eventually executed (unless the request was already executed in the order 
consistent with the commit order) and then the request is never rolled back. 
This is so, because:
\begin{itemize}
\item the $\committed$ list is included in the $\newOrder$ list as a prefix in 
the $\commit$ procedure
(Algorithm~\ref{alg:bayou} line~\ref{alg:bayou:tobdeliver_neworder}), 
\item until the request $r$ executes it has to feature on the list 
$\toBeExecuted$ (Algorithm~\ref{alg:bayou} line~\ref{alg:bayou:tobeexecuted}) 
and there can be only a finite number of items preceding it on that list, 
\item the $\toBeRolledBack$ list cannot grow indefinitely without executing 
some of the requests from the $\toBeExecuted$ list, which means that $r$ is
eventually executed (Algorithm~\ref{alg:bayou} line~\ref{alg:bayou:execute}),
\item and finally a request which is included in both the $\committed$ and 
$\executed$ lists is never part of the $\outOfOrder$ list 
(Algorithm~\ref{alg:bayou} line~\ref{alg:bayou:outoforder}), which means it 
will not be scheduled for rollback. 
\end{itemize}

Weak operations execute atomically in the invoke code block where the 
response is always returned immediately to the client.\footnote{If due to 
operation reexecutions multiple responses are returned to the client we discard 
the additional ones.} For a given weak event $e$ the response is computed on 
the $\boxx$ object in some state $S_\alpha$, where $\alpha$ is the current 
trace of the $\boxx$ object at the time of the operation's invocation. We let 
$\trace(e)$ denote the trace $\alpha$.

On the other hand, strong operations follow a more complicated route. For a 
strong event $e$: firstly the $\COMMIT$ message is $\tobcast$, then upon 
its $\tobdeliver$y the request $r = \req(e)$ is committed. Since $r$ is not 
disseminated using $\rbcast$, it is never included in the $\tentative$ list, and 
so it executes for the first time after its commit. Thus, each strong operation 
is executed on each replica exactly once, on a $\boxx$ object in some state 
$S_\alpha$, where $\alpha$ is the current trace of the $\boxx$ object at the 
time of the execution. Note that the trace $\alpha$ is exactly the same on each 
replica and it consists exactly of all the requests preceding $\req(e)$ in the 
$\committed$ list (which due to the properties of $\tobdeliver$y has the same 
value on each replica upon $r$'s commit). Again, as in case of weak events, we 
let $\trace(e)$ denote the trace $\alpha$.

Note that each strong operation executed in some event $e$ finishes only after 
the replica $\tobdeliver$s the message $m = \msg_\TOB(e)$. It means that for 
every operation executed in event $e'$, such that $e \rarb e'$, if 
$\msg_\TOB(e') = m'$ ($m' \neq \bot$), then $\tobNo(m) < \tobNo(m')$.

\noindent \textbf{Arbitration.}
We construct the total order relation $\ar$ by sorting all shared events based 
on the order in which their respective $\tobcast$ messages are $\tobdeliver$ed, 
i.e., respecting the $\tobNo$ order.

Next, we interleave the shared events with local events in the following way: 
each local event $e$ occurs in $\ar$ after the last shared event $e'$ such that 
$e \not\rarb e'$. Thus, for each shared event $e'$ the following holds $e \raar 
e' \Rightarrow e \rarb e'$. The relative order of local operations is 
irrelevant.

We construct the perceived arbitration order $\perc(e)$, for each event $e$, 
using the trace $\alpha = \trace(e)$. More precisely, we add all the events 
whose requests appear in $\alpha$ in the order of occurence, next we add all 
the remaining shared events according to their order in $\ar$. Finally, we 
interleave the constructed sequence with local events in a similar way as in 
case of $\ar$, i.e., for each local event $f$ and each shared event $g$, the 
following holds $f \rapar{e} g \Rightarrow f \rarb g$.

Note that for a strong event $e$, $\perc(e) = \ar$. This is because $e$ 
executes once $\req(e)$ is on the $\committed$ list, and its position on the 
list is determined by the $\tobNo$ order, which means that the trace $\alpha$ 
contains exactly all the shared events preceding $e$ in $\ar$.

\noindent \textbf{Visibility.}
For any two events $e, e' \in E$, such that $\trace(e') = \alpha$, we include 
an edge $e \ravis e'$ in our construction of $\vis$, if:
\begin{enumerate}
\item \label{enum:bayouvis:wustrongstrong}
$e \in \Psi$, $e' \in \Psi_s$, and $\req(e) \in \alpha$,
\item \label{enum:bayouvis:strongwu}
$e \in \Psi_s$, $e' \in \Psi_w$, and $\req(e) \in \alpha$,
\item \label{enum:bayouvis:strongwro}
$e \in \Psi_s$, $e' \in \Omega$, and $\req(e) \in \alpha$,
\item \label{enum:bayouvis:wuweak}
$e \in \Psi_w$, $e' \in \Psi_w \cup \Omega$, and $\req(e) \in \alpha$,
\item \label{enum:bayouvis:wrowro}
$e, e' \in \Omega$, and $e \rarb e'$,
\item \label{enum:bayouvis:wrowustrong}
$e \in \Omega$, $e' \in \Psi$, and $e \raar e'$.
\end{enumerate}

The edges~\ref{enum:bayouvis:wustrongstrong}-\ref{enum:bayouvis:wuweak} are 
essential, while the 
edges~\ref{enum:bayouvis:wrowro}-\ref{enum:bayouvis:wrowustrong} are 
non-essential.

Note that in case of edge~\ref{enum:bayouvis:wuweak}, either $e \raaa{\TOBdel} 
e'$, or $e \raaa{\RBdel} e'$, and thus $e \rarb e'$ is implied (see the general 
observations in Section~\ref{sec:proofs}). Thus, for all 
edges~\ref{enum:bayouvis:wuweak}-\ref{enum:bayouvis:wrowustrong}, $e \rarb e'$.

Additionally, observe that in case of edge~\ref{enum:bayouvis:wustrongstrong}, 
$e \raaa{\TOBdel} e'$, because $\alpha$ contains only requests on the 
$\committed$ list (see the additional observations in the beginning of the 
proof), and thus $e \raaa{\tobNo} e'$. Similarly, in case of 
edges~\ref{enum:bayouvis:strongwu} and \ref{enum:bayouvis:strongwro}, $e 
\raaa{\TOBdel} e'$, because $\msg_\RB(e) = \bot$ and thus $\req(e)$ can appear 
in $\alpha$ only if it was $\tobdeliver$ed by the replica executing $e'$. Also 
in case of edge~\ref{enum:bayouvis:strongwu}, $e \raaa{\tobNo} e'$.

Having defined $A$ (through $\vis$, $\ar$ and $\perc$), it now remains to show
that $A \models \FEC(\weak, \FF) \wedge \LIN(\strong, \FF)$, or more 
specifically $A \models \EV(\weak) \wedge \EV(\strong) \wedge \NCC(\weak) 
\wedge \NCC(\strong) \wedge \FRVAL(\weak) \wedge \RVAL(\strong) \wedge 
\CPERC(\weak) \wedge \SO(\strong) \wedge \RT(\strong)$.

\noindent \textbf{Eventual visibility.}
We prove now that eventual visibility is satisfied for all events:
\begin{itemize}
\item each shared event $e$ is visible to all subsequent events from some 
point, because $\msg_\TOB(e)$ is eventually $\tobdeliver$ed and $r = \req(e)$ 
is placed on the $\committed$ list on each replica, thus $r$ is eventually 
executed and never rolled back, and is included in the trace of the $\boxx$ 
object from some point (\ref{enum:bayouvis:wustrongstrong}, 
\ref{enum:bayouvis:strongwu}, \ref{enum:bayouvis:strongwro} and 
\ref{enum:bayouvis:wuweak}),
\item each local event $e$ is visible to all subsequent local events from 
some point (\ref{enum:bayouvis:wrowro}),
\item each local event $e$ is visible to all subsequent shared events
from some point, because by construction of $\ar$ there is only a finite number 
of events $e'$ such that $e \not\raar e'$ (\ref{enum:bayouvis:wrowustrong}).
\end{itemize}

\noindent \textbf{No circular causality.}
We need to show that $\acyclic(\hb \cap (W \times W))$ and $\acyclic(\hb \cap 
(S \times S))$, where $W \subseteq E, S \subseteq E$, are, respectively, the 
sets of all weak, and strong events. We elect to prove a more general case of 
$\acyclic(\hb)$.

Recall that $\hb = (\vis \cup \so)^+$. If $\acyclic(\vis \cup \so)$, then 
$\acyclic(\hb)$, because transitive edges cannot introduce cycles. Thus, we 
have six types of edges to consider: 
edges~\ref{enum:bayouvis:wustrongstrong}-\ref{enum:bayouvis:wrowustrong} from 
$\vis$ and the seventh edge $e \raso e'$. We divide them into two groups: the 
first one consists of 
edges~\ref{enum:bayouvis:wustrongstrong}-\ref{enum:bayouvis:strongwro}, 
while the second one consists of edges~\ref{enum:bayouvis:wuweak}-7. Note that 
for the second group $e \rarb e'$ always holds.

There can be no cycles when we restrict the edges only to the ones from the 
first group, as the edges~\ref{enum:bayouvis:wustrongstrong} and 
\ref{enum:bayouvis:strongwu} are constrained by the $\tobNo$ order, and 
edge~\ref{enum:bayouvis:strongwro} leads to a local event which cannot be 
followed using only edges from the first group.

Also, there can be no cycles when we restrict the edges only to the ones from 
the second group, as all the edges are constrained by the $\rb$ relation, which 
is naturally acyclic.

Thus, a potential cycle could only form when we mix edges from both groups.
Let us assume that the cycle contains the following chain of edges: $a 
\raaa{\hb} b \raaa{\hb} ... \raaa{\hb} c \raaa{\hb} ... \raaa{\hb} d$, where 
$a,b,c,d \in E$, all the edges between $b$ and $c$ belong to the second group, 
while the other ones belong to the first group. Notice that $b \rarb c$, and
that $a,c \in \Psi$. Thus, the chain consists of a series of edges from the 
first group and a series of edges from the second group. The whole cycle can be 
combined from multiple such chains, but for simplicity, let 
us assume that it contains only one such chain and that $d = a$ (the same 
reasoning as below can be applied iteratively for multiple interleavings of 
edges from the two groups).

If $b \in \Psi$, then $a \raaa{\tobNo} b$ 
(edges~\ref{enum:bayouvis:wustrongstrong} and \ref{enum:bayouvis:strongwu}), 
and since $b \rarb c$, also $b \raaa{\tobNo} c$ (see the additional 
observations in the beginning of the proof). A contradiction: $a \raaa{\tobNo} 
b \raaa{\tobNo} c \raaa{\tobNo} a$.

If $b \in \Omega$, then $a \in \Psi_s$, and $a \raaa{\TOBdel} b$ 
(edge~\ref{enum:bayouvis:strongwro}). Either $a \raaa{\tobNo} c$, or $c 
\raaa{\tobNo} a$. In the former case we end up with a similar contradiction as 
above: $a \raaa{\tobNo} c \raaa{\tobNo} a$. In the latter case, since $c 
\raaa{\tobNo} a$, also $c \raaa{\TOBdel} b$ (the message $\msg_\TOB(c)$ is 
$\tobdeliver$ed before the message $\msg_\TOB(a)$). However, $b \rarb c$, which 
means that the message $\msg_\TOB(c)$ was not even $\tobcast$ yet when $b$ 
executed. A contradiction.

\noindent \textbf{Single order.}
Since there are no pending strong operations (because eventually every 
message is $\tobdeliver$ed and the operations finish), we have to simply prove 
that $\vis \cap (E \times S) = \ar \cap (E \times S)$, where $S = \{e : \lvl(e) 
= \strong\}$. In other words, for any two events $e \in E, e' \in S$: $e \ravis 
e' \iff e \raar e'$.

Let us begin with $e \ravis e' \Rightarrow e \raar e'$. Either $e \in \Psi$, 
and thus $e \raaa{\tobNo} e'$ (edge~\ref{enum:bayouvis:wustrongstrong}), or 
$e \in \Omega$ (edge~\ref{enum:bayouvis:wrowustrong}). In both cases $e \raar 
e'$.

Now let us consider $e \raar e' \Rightarrow e \ravis e'$. Either $e \in \Psi$, 
or $e \in \Omega$. In the former case, $e \raaa{\tobNo} e'$, and thus $e$ must 
be included in $\trace(e')$, which means that $e \ravis e'$ 
(edge~\ref{enum:bayouvis:wustrongstrong}). In the latter case, $e \rarb e'$ 
(by construction of $\ar$), and thus also $e \ravis e'$ 
(edge~\ref{enum:bayouvis:wrowustrong}).

\noindent \textbf{Return value consistency.}
Since for a strong event $e$, $\perc(e) = \ar$ and $\fcontext(A, e) = 
\context(A, e)$. Thus, for each event $e \in E$, we need to show that: 
$\rval(e) = \FF(\op(e), \fcontext(A, e))$. We base our reasoning below on 
essential $\vis$ edges and $\perc(e)$ order.

Firstly, observe that we can exclude from $\fcontext(A, e)$ all local events 
which by the definition of an RO operation are irrelevant for the computation 
of $\FF$. Thus, let $C = (E_C, \op, \vis, \perc(e))$, where $E_C = \{e' \in 
\Psi : e' \ravis e \}$.

Then, recall that $\rval(e)$ is obtained by calling $\boxx.\execute$ on the 
$\boxx$ object in state $S_\alpha$, where $\alpha = \trace(e)$, and that 
$\rval(e) = \FF(\opn, C_\alpha)$, where $C_\alpha = (E_\alpha, \op_\alpha, 
\vis_\alpha, \ar_\alpha)$ is a context constructed from $\alpha$ as defined in 
Section~\ref{sec:state_object}. It suffices to show that the context $C$ is 
isomorphic with $C_\alpha$, which we do below.

Clearly, by construction of $\vis$, if $e' \ravis e$ and $e' \in E_C$, then
$\req(e') \in \alpha$. Thus, $E_\alpha$ consists of the $\Request$ records of 
the events in $E_C$. By the way how $\Request$ records are constructed 
(Algorithm~\ref{alg:bayou_modified} line~\ref{alg:bayou_modified:new_request}), 
for any given event $e \in E_C$, $\op_\alpha(\req(e))$ equals $\op(e)$. Also,
for any two events $f,g \in E_C$, $f \rapar{e} g \iff \req(f) \raaa{\ar_\alpha} 
\req(g)$, which follows trivially from the construction of $\perc(e)$. It 
remains to show that for any two events $f,g \in E_C$, $f \ravis g \iff \req(f) 
\raaa{\vis_\alpha} \req(g)$.

If $g \in \Psi_w$ and $f \ravis g$, then $\req(f) \in \trace(g)$, and thus
$\req(f).\dot \in \req(g).\ctxsf$, which implies $\req(f) \raaa{\vis_\alpha} 
\req(g)$.

If $g \in \Psi_w$ and $\req(f) \raaa{\vis_\alpha} \req(g)$, then 
$\req(f).\dot \in \req(g).\ctxsf$, and thus $\req(f) \in \trace(g)$, which 
implies $f \ravis g$.

If $g \in \Psi_s$ and $f \ravis g$, then $f \raar g$ (by Single Order), and 
thus $f \raaa{\tobNo} g$. Since $\req(g)$ is committed at the time of $e$'s 
execution ($\req(g) \in \alpha$ and $\lvl(g) = \strong$), so is $\req(f)$ but 
its position on the $\committed$ list is earlier ($f \raaa{\tobNo} g$). Because 
the order of requests in the trace is based on the $\executed$ list, whose 
order is consistent with the order of the $\committed$ list, $\req(f)$ precedes 
$\req(g)$ in $\alpha$, which implies $\req(f) \raaa{\ar_\alpha} \req(g)$. Then, 
by construction of $C_\alpha$, $\req(f) \raaa{\vis_\alpha} \req(g)$.

If $g \in \Psi_s$ and $\req(f) \raaa{\vis_\alpha} \req(g)$, then $\req(f) 
\raaa{\ar_\alpha} \req(g)$, and thus $\req(f)$ precedes $\req(g)$ in $\alpha$.
Since $\req(g)$ is committed at the time of $e$'s execution, both $\req(f)$ and
$\req(g)$ belong to the $\committed$ list during the $e$'s execution, which 
implies that $f \raaa{\tobNo} g$. Thus, $f \raar g$, and by Single Order, $f 
\ravis g$.

Thus, $C$ is isomorphic with $C_\alpha$.

\noindent \textbf{Convergent perceived arbitration.}
We now show, that for each event $e \in E$ there exist only a finite number of 
weak events $e'$, such that the prefixes of $\perc(e')$ and $\ar$ up to the 
event $e$ differ, which is a sufficient condition to prove $\CPERC(\weak)$.

If $e \in \Psi$, then eventually on each replica $\msg_\TOB(e)$ is 
$\tobdeliver$ed, and $\req(e)$ is committed and executed. Thus, from some 
point, the trace of each subsequent event $e'$ contains $\req(e)$, preceded by 
requests of events $e''$ committed earlier, such that $e'' \raaa{\tobNo} e$. 
Both $\ar$ and $\perc(e')$ are constructed by first ordering shared events and 
then interleaving them with local events using the same procedure. In both 
$\ar$ and $\perc(e')$, $e$ is preceded by the same shared events $e''$, such 
that $e'' \raaa{\tobNo} e$. Then, it is also preceded by the same local events, 
which means the prefixes of $\perc(e')$ and $\ar$ up to $e$ are equal. 

If $e \in \Omega$, then eventually the requests of all shared events $e''$, 
such that $e'' \raar e$, are committed and executed on each replica. Then, from 
some point, the trace of each subsequent event $e'$ contains the requests of 
events $e''$, ordered by $\tobNo$. Thus, $e$ is preceded in both $\ar$ and 
$\perc(e')$ by the same shared events $e''$. Because both $\ar$ and $\perc(e')$ 
are interleaved with local events using the same procedure, $e$ is also 
preceded in both $\ar$ and $\perc(e')$ by the same local events, which means 
the prefixes of $\perc(e')$ and $\ar$ up to $e$ are equal.

\noindent \textbf{Real-time order.}
We need to show that arbitration order respects real-time order of strong 
operations, i.e., $\rb \cap (S \times S) \subseteq \ar$, where $S = \{e : 
\lvl(e) = \strong\}$. In other words, for any two $e,e' \in S$: $e \rarb e' 
\Rightarrow e \raar e'$.

Clearly, if $e \rarb e'$, then $e \raaa{\tobNo} e'$ (see the additional 
observations in the beginning of the proof). Thus, $e \raar e'$ (by 
construction of $\ar$).
\end{proof}

Now, let us continue with the proof of the guarantees offered by \ACTBayou in 
the asynchronous runs.

\bayouasynchronous*

\begin{proof}
To show the inability of \ACTBayou to satisfy $\LIN(\strong, \FF)$ in 
asynchronous runs, it is sufficient to observe that due to some of the 
$\tobcast$ messages not being $\tobdeliver$ed, some of the strong operations 
remain pending. A pending operation's return value equals $\nabla$ 
which is unreconcilable with the requirements of the predicate $\RVAL(\FF)$.

The proof regarding the guarantees of the weak operations is similar to the one 
for the stable runs, thus we rely on it and focus only on differences between 
stable and asynchronous runs that need to be addressed. Now for any given
arbitrary asynchronous run of \ACTBayou represented by a history $H = (E, \op, 
\rval, \rb, \sss, \lvl)$ we have to find suitable $\vis$, $\ar$ and $\perc$, 
such that $A = (H, \vis, \ar, \perc)$ is such that $A \models \FEC(\weak, \FF)$.

\noindent \textbf{Additional observations.}
The same observations apply as in case of stable runs, with the only 
distinction that some strong events $e$ remain pending due to the lack of 
$\tobdeliver$y of $\msg_\TOB(e)$. In such cases $\trace(e)$ is undefined.

Now let us make one more observation: the request of a weak updating event $e$ 
whose $\msg_\TOB(e)$ is never $\tobdeliver$ed, even though it never commits,  
eventually \emph{settles}, i.e. it is eventually executed and is never 
rolled back after that execution. It is so, because after $r = \req(e)$ is 
$\rbdeliver$ed by each replica and placed on the $\tentative$ list, only a 
finite number of other requests can commit (due to the properties of TOB in 
asynchronous runs), and also only a finite number of other requests can have a 
lesser $\Request$ record (as defined by the operator $<$ in 
Algorithm~\ref{alg:bayou}) and thus precede $r$ in the $\tentative$ list (due 
to monotonically increasing clocks on each replica). Thus, once $r$ is placed 
on the $\toBeExecuted$ list, it eventually executes, and when executed $r$ can 
be rolled back at most a finite number of times, due to a commit of other 
request, or a lesser $\Request$ being inserted into the $\tentative$ list.

\noindent \textbf{Arbitration.}
We construct the total order relation $\ar$ by sorting all shared events 
based on the order in which their respective $\tobcast$ messages are 
$\tobdeliver$ed, i.e., respecting the $\tobNo$ order. Shared events whose 
messages are not $\tobdeliver$ed are ordered after those whose messages are 
$\tobdeliver$ed, with weak updating events appearing first, ordered relatively 
based on their $\Request$ records, followed by pending strong events.

Next, we interleave the shared events with local events in the following 
way: each local event $e$ occurs in $\ar$ after the last \emph{non-pending} 
shared event $e'$ such that $e \not\rarb e'$. Thus, for each non-pending shared 
event $e'$ the following holds $e \raar e' \Rightarrow e \rarb e'$. The 
relative order of local events is irrelevant.

We construct the perceived arbitration order $\perc(e)$ for each event $e$, in 
the same way as in case of stable runs, i.e. using $\trace(e)$, the 
remaining shared events from $\ar$, and finally interleaving the constructed 
sequence with local events as in case of $\ar$ (so that for each local event 
$e'$ and each non-pending shared event $e''$, the following holds $e' \rapar{e} 
e'' \Rightarrow e' \rarb e''$.

For a pending strong event $e$, which was not executed at all, we let 
$\perc(e) = \ar$.

Note that for a non-pending strong event $e$, $\perc(e) = \ar$. This is because 
$e$ executes once $\req(e)$ is on the $\committed$ list, and its position on 
the list is determined by the $\tobNo$ order, which means that its trace will 
contain exactly all the shared events preceding $e$ in $\ar$.

\noindent \textbf{Visibility.}
We construct the visibility relation in the same way as in the stable 
runs case. However, we remove edges to and from pending strong events. 
Since pending operations do not provide a return value, no edge to a pending 
event is essential. Also, as we guarantee only eventual visibility for weak 
events, edges to strong events are not necessary to satisfy $\EV(\weak)$. 
Moreover, edges \emph{from} pending events are not needed either, because by 
definition a pending event is never followed in $\rb$ by any other event (which 
is a requirement to fail the test for $\EV$). Again, for all 
edges~\ref{enum:bayouvis:wuweak}-\ref{enum:bayouvis:wrowustrong}, $e \rarb e'$.

Having defined $A$ (through $\vis$, $\ar$ and $\perc$), it now remains to show
that $A \models \FEC(\weak, \FFnnc)$, or more specifically $A \models 
\EV(\weak) \wedge \NCC(\weak) \wedge \FRVAL(\weak) \wedge \CPERC(\weak)$.

\noindent \textbf{Eventual visibility.}
We prove now that eventual visibility is satisfied for all \emph{weak} events:
\begin{itemize}
\item each non-pending shared event $e$, such that $\msg_\TOB(e)$ is 
eventually $\tobdeliver$ed on each replica, is visible to all subsequent 
non-pending events from some point, because $r = \req(e)$ is placed on the 
$\committed$ list on each replica, thus $r$ is eventually executed and never 
rolled back, and is included in the trace of the $\boxx$ object from some point 
(\ref{enum:bayouvis:strongwu}, \ref{enum:bayouvis:strongwro} and 
\ref{enum:bayouvis:wuweak}),
\item each weak updating event $e$, such that $\msg_\TOB(e)$ is \emph{not} 
eventually $\tobdeliver$ed on each replica, is visible to all subsequent 
non-pending events from some point, because it settles (see the additional 
observations in the beginning of the proof) and is included in the trace of the 
$\boxx$ object on each replica from some point (\ref{enum:bayouvis:wuweak}),
\item each local event $e$ is visible to all subsequent local events from 
some point (\ref{enum:bayouvis:wrowro}),
\item each local event $e$ is visible to all subsequent non-pending 
shared events from some point, because by construction of $\ar$ there is 
only a finite number of events $e'$ such that $e \not\raar e'$ 
(\ref{enum:bayouvis:wrowustrong}).
\end{itemize}

\noindent \textbf{No circular causality.}
We use exactly the same reasoning as in the stable runs case to show that 
$\acyclic(\hb)$ holds true.

\noindent \textbf{Return value consistency.}
Again, we use exactly the same reasoning as in the stable runs case to show
that for each \emph{weak} event $e \in E$: $\rval(e) = \FFnnc(\op(e), 
\fcontext(A, e))$. Although this time we only need to prove return value 
consistency for weak operations, it can be shown that it also holds for 
\emph{non-pending} strong events.

\noindent \textbf{Convergent perceived arbitration.}
We now show, that for each non-pending\footnote{We can exclude pending events, 
because according to the construction of $\vis$ they are not visible to any 
other event, and thus automatically satisfy the requirements of the $\CPERC$ 
predicate.} event $e \in E$ there exist only a finite number of weak events 
$e'$, such that the prefixes of $\perc(e')$ and $\ar$ up to the event $e$ 
differ, which is a sufficient condition to prove $\CPERC(\weak)$.

If $e \in \Psi$ and $\msg_\TOB(e)$ is eventually $\tobdeliver$ed, then the same 
logic can be applied as in case of stable runs to show that from some point for 
each subsequent event $e'$ the prefixes of $\perc(e')$ and $\ar$ up to $e$ are 
equal. 

If $e \in \Psi_w$ and $\msg_\TOB(e)$ is never $\tobdeliver$ed, then it 
eventually settles (see the additional observations in the beginning of the 
proof) and thus also the same logic can be applied as in case of stable runs, 
with the distinction that $e$ is preceded in $\ar$ and $\perc(e')$ not only by 
events $e''$ whose requests are committed, but also by events $e''$, such that 
$\req(e'') < \req(e)$.

If $e \in \Omega$, then eventually the requests of all shared events $e''$, 
such that $e'' \raar e$ (none of which are pending by the construction of 
$\ar$), are either committed, or settled, and executed on each replica. 
Then, from some point, the trace of each subsequent event $e'$ contains the 
requests of events $e''$, ordered by both $\tobNo$, and based on their 
$\Request$ records. Thus, $e$ is preceded in both $\ar$ and $\perc(e')$ by the 
same shared events $e''$. Because both $\ar$ and $\perc(e')$ are interleaved 
with local events using the same procedure, $e$ is also preceded in both $\ar$ 
and $\perc(e')$ by the same local events, which means the prefixes of 
$\perc(e')$ and $\ar$ up to $e$ are equal. 
\end{proof}